\documentclass[aps,a4paper,nofootinbib,showpacs,preprintnumbers,amsmath,amssymb]{article} 
\usepackage{latexsym}
\usepackage{amssymb}
\usepackage{graphicx}
\usepackage{amsmath}
\usepackage{geometry}
\usepackage{amsthm,amssymb,amstext,amscd,amsfonts,amsbsy,amsxtra,latexsym}

\usepackage{color}

\usepackage[colorlinks=true]{hyperref}

\definecolor{dark-red}{rgb}{.54,.0,.0}
\definecolor{dark-green}{rgb}{.0,.4,.0}
\definecolor{dark-blue}{rgb}{.04,.04,.4}

\hypersetup{linkcolor=dark-red, urlcolor=dark-blue, citecolor=dark-green}

\usepackage{pdflscape}
\usepackage{lscape}
\usepackage{tabularx}
\usepackage{rotating}
\usepackage{verbatim}
\usepackage{psfrag}
\usepackage{enumerate}

\renewenvironment{proof}   
{\noindent
	{\em Proof}.}   
{\nopagebreak\mbox{}\hfill $\Box$\par\addvspace{0.5cm}}   
\newenvironment{altproof}[1]   
{\noindent
	{\em Proof of {#1}}.}   
{\nopagebreak\mbox{}\hfill $\Box$\par\addvspace{0.5cm}}

\DeclareMathOperator{\arcsec}{arcsec}
\newcommand{\flt}{\frac{\Lambda}{3}}

\newcommand{\be}{\begin{equation}}
\newcommand{\ee}{\end{equation}}
\newcommand{\ben}{\begin{eqnarray*}}
\newcommand{\een}{\end{eqnarray*}}
\newcommand{\bea}{\begin{eqnarray}}
\newcommand{\eea}{\end{eqnarray}}
\newcommand{\md}{{\,\mathrm{d}}}
\newcommand{\mmd}{{\mathrm{d}}}
\newcommand{\beq}{\begin{equation}}
\newcommand{\eeq}{\end{equation}}
\DeclareMathAlphabet{\mathpzc}{OT1}{pzc}{m}{it}

\def\bbR{{\mathbb R}}
\def\bbS{{\mathbb S}}
\newcommand{\dV}{{\rm dVol}}

\newtheorem{thm}{Theorem}[section]

\newtheorem{lem}[thm]{Lemma}

\newtheorem{cor}[thm]{Corollary}

\newtheorem{rmk}[thm]{Remark}
\newtheorem{ex}[thm]{Example}

\let\truegamma=\gamma

\newcommand{\nabb}{\mbox{$\nabla \mkern-13mu /$\,}}
\newcommand{\GGamma}{\mbox{$\Gamma \mkern-12mu /$\,}}

\newcommand{\gin}{g \mkern-8.7mu /}

\renewcommand{\gamma}{{\gin}}

\newcommand{\vch}{v_{\mbox{\tiny${\cal CH}^+$}}}
\newcommand{\bch}{b_{\mbox{\tiny${\cal CH}^+$}}^{\phi_\star}}
\newcommand{\phich}{\phi_{\star,{\cal CH}^+}}
\newcommand{\grad}{\mbox{grad}\,}
\newcommand{\bphi}{b^{\phi_{\star}}}
\newcommand{\bphir}{\left.b^{\phi_{\star}}\right|_{r=r_-}}
\newcommand{\Omegasqch}{\Omega^2_{\mbox{\tiny${\cal CH}^+$}}}

\newcommand{\Omegafch}{\Omega^4_{\mbox{\tiny${\cal CH}^+$}}}

\newcommand{\bphib}{b_{\phi_\star}}
\newcommand{\bthetab}{b_{\theta_\star}}

\newcommand{\normb}{\|\bphi\|^2}
\newcommand{\normbch}{\|\bch\|^2}
\newcommand{\f}{\tilde{f}}
\newcommand{\g}{\tilde{g}}
\newcommand{\h}{\tilde{h}}
\newcommand{\tldej}{\tilde{\mbox{\j}}}
\newcommand{\pu}{\partial_u}
\newcommand{\pv}{\partial_{\vch}}
\newcommand{\pt}{\partial_{\theta_\star}}
\newcommand{\pp}{\partial_{\phi_\star}}
\newcommand{\pvpnormal}{{\textstyle\frac{1}{\Omega^2}}
	\left(\partial_{v}+\bphi\partial_{\phi_\star}\right)}
\newcommand{\pvpnormaldois}{\frac{\partial_{v}+\bphi\partial_{\phi_\star}}{\Omega^2}}
\newcommand{\pvp}{{\textstyle\frac{1}{\Omegasqch}}
	\left(\partial_{\vch}\!\!+\bch\partial_{\phi_\star}\right)}
\newcommand{\pvpdois}{
	\frac{\partial_{\vch}\!\!+\bch\partial_{\phi_\star}}{\Omegasqch}}
\newcommand{\pr}{\partial_{r_\star}}

\newcommand{\vh}{v_{{\cal H}^+}}
\newcommand{\uh}{u_{{\cal H}^+}}
\newcommand{\tuh}{\tilde{u}_{{\cal H}^+}}
\newcommand{\bh}{b_{\mbox{\tiny${\cal H}^+$}}^{\phi_\star}}
\newcommand{\phih}{\phi_{\star,{\cal H}^+}}
\newcommand{\bphirh}{\left.b^{\phi_{\star}}\right|_{r=r_+}}
\newcommand{\Omegasqh}{\Omega^2_{\mbox{\tiny${\cal H}^+$}}}

\newcommand{\Omegafh}{\Omega^4_{\mbox{\tiny${\cal H}^+$}}}

\newcommand{\pvh}{\partial_{\vh}}
\newcommand{\pvph}{
	\left(\partial_{\vh}\!\!+\bh\partial_{\phi_\star}\right)}

\newcommand{\puh}{\partial_{\uh}}
\newcommand{\puph}{{\textstyle\frac{1}{\Omegasqh}}\partial_{\uh}}

\newcommand{\nCv}{n_{\underline{\cal C}_v}}
\newcommand{\nCu}{n_{{\cal C}_u}}
\newcommand{\VCv}{\dV_{\underline{\cal C}_v}}
\newcommand{\VCu}{\dV_{{{\cal C}_u}}}
\newcommand{\VS}{\dV_{\Sigma_{r_\star}}}
\newcommand{\tinyt}{\mbox{\tiny $\overline{\mbox{\tt T}}$}}

\newcommand{\smallt}{\mbox{\small $\overline{\mbox{\tt T}}$}}
\newcommand{\tinyy}{\mbox{\tiny $\overline{\mbox{\tt Y}}$}}

\newcommand{\smally}{\mbox{\small $\overline{\mbox{\tt Y}}$}}
\newcommand{\omt}{\omega_{\tinyt}}
\newcommand{\omy}{\omega_{\tinyy}}
\newcommand{\omtt}{\omega_{\theta_\star}}
\newcommand{\omp}{\omega_{\phi_\star}}
\newcommand{\ot}{\otimes}

\newcommand{\rhos}{\hat{\rho}}

\newcounter{mnotecount}[section]

\renewcommand{\themnotecount}{\thesection.\arabic{mnotecount}}

\newcommand{\mnote}[1]
{\protect{\stepcounter{mnotecount}}$^{\mbox{\footnotesize $%
\!\!\!\!\!\!\,\bullet$\themnotecount}}$ \marginpar{
\raggedright\tiny\em $\!\!\!\!\!\!\,\bullet$\themnotecount: #1} }

\def\XXint#1#2#3{{\setbox0=\hbox{$#1{#2#3}{\int}$ }
\vcenter{\hbox{$#2#3$ }}\kern-.6\wd0}}

\def\thesection{\arabic{section}}

\makeatletter
\def\p@subsection{}
\def\p@subsubsection{}
\def\p@paragraph{}
\makeatother

\renewcommand{\thefootnote}{\fnsymbol{footnote}}

\begin{document}
\begin{center}

{\bf {\large Regularity of a double null coordinate system\\ for Kerr--Newman--de Sitter spacetimes}}

\bigskip
Anne T.\ Franzen\footnote{e-mail address: anne.franzen@tecnico.ulisboa.pt.} and
Pedro M.\ Gir\~ao\footnote{e-mail address: pgirao@math.ist.utl.pt.}

\bigskip
{Center for Mathematical Analysis, Geometry and Dynamical Systems,}\\
{Instituto Superior T\'ecnico, Universidade de Lisboa,}\\
{Av.\ Rovisco Pais, 1049-001 Lisbon, Portugal}

\bigskip
{\bf Dedicated to Giorgio Fusco}

\newcommand\blfootnote[1]{%
	\begingroup
	\renewcommand\thefootnote{}\footnote{#1}%
	\addtocounter{footnote}{-1}%
	\endgroup
}



\end{center}
\medskip

\centerline{\bf Abstract}

\noindent We construct a double null coordinate system $(u,v,\theta_\star,\phi_\star)$ for 
Kerr--Newman--de Sitter black hole interior spacetimes and prove that
the two dimensional spheres given by the intersection of the
hypersurfaces $u=\mbox{constant}$ and $v=\mbox{constant}$ are $C^\infty$
in Boyer--Lindquist coordinates (including at the ``poles").
The null coordinates allow one to immediately extend some results 
previously proven for Kerr. As an example,
we 
illustrate how Sbierski's result in~\cite{Sbierski},
for the wave equation on the black hole interior,
for Reissner--Nordstr\"{o}m and Kerr spacetimes,
applies to Kerr--Newman--de Sitter spacetimes.

\medskip

\tableofcontents

\section{Introduction}

The Kerr--Newman--de Sitter (KNdS) metric is a solution of the Einstein--Maxwell
equations with a positive cosmological constant $\Lambda$:
\begin{eqnarray*}
&&R_{\mu\nu}-\,\frac{1}{2}Rg_{\mu\nu}+\Lambda g_{\mu\nu}=2T_{\mu\nu},\\
&&dF=d\star F=0,\\
&&T_{\mu\nu}=F_{\mu\alpha}F_\nu^{\ \alpha}-\,\frac{1}{4}F_{\alpha\beta}
F^{\alpha\beta}g_{\mu\nu}.
\end{eqnarray*}
Here $R_{\mu\nu}$ are the components of the Ricci tensor of the spacetime
metric $g$, $R$ is the scalar curvature, $\star$ is the Hodge star operator, and
$F_{\mu\nu}$ is the Faraday electromagnetic 2-form. So, 
the KNdS metric is an electrovacuum solution of the Einstein field equations,
i.e.\ it is a solution in which the only nongravitational mass-energy present is an electromagnetic field.
The spacetime is the four-dimensional manifold $\bbR^2\times\bbS^2$ with metric
given by 
\bea
g&=&\frac{\rho^2}{\Delta_r}\md r^2+\frac{\rho^2}{\Delta_{\theta}}\md \theta^2
+\frac{1}{\rho^2}(a^2\sin^2\theta\Delta_{\theta}-\Delta_r)\md t^2\nonumber \\
&& +\frac{\sin^2\theta}{\Xi^2\rho^2}((r^2+a^2)^2\Delta_{\theta}-a^2\sin^2\theta\Delta_r)\md \phi^2
-\,\frac{2a\sin^2\theta}{\Xi\rho^2}((r^2+a^2)\Delta_{\theta}-\Delta_r)\md \phi \md t
\label{rt_metric}
\eea 
in Boyer--Lindquist coordinates,
where
\begin{eqnarray*}
\Delta_r&=&(r^2+a^2)\left(1-\,\frac{\Lambda}{3}r^2\right)-2Mr+e^2,\nonumber\\
\Delta_{\theta}&=& 1+\frac{\Lambda}{3}a^2\cos^2\theta,\nonumber\\
\rho^2&=&r^2+a^2\cos^2\theta,\nonumber\\
\Xi&=&1+\frac{\Lambda}{3}a^2,\nonumber
\end{eqnarray*}
$\theta$ is the colatitude, with
$$
0\leq\theta\leq\pi,
$$
and $\phi$ is the longitude, with 
$$
\phi\in \bbS^1,
$$
(see
Carter~\cite{carter1}, and Akcay and Matzner~\cite{akcay} and Kraniotis~\cite{krani}, for example).
Here $M>0$ and $a\neq 0$ are mass and angular momentum parameters, respectively,
 and $e$ is a charge parameter (which may be zero).
Without loss of generality, 
we assume that the magnetic charge is zero and $a$ is positive.
This metric is supposed to represent a rotating black hole,
with charge, in a universe which is expanding at an accelerated rate (as ours is).
We refer to Appendix~\ref{app-D},
where we calculate Komar integrals over the event horizon,
for the relation between these parameters
and the physical quantities of the black hole.
We wish to consider subextremal metrics, meaning that
$\Delta_r$ has four distinct real roots
$$
r_n<0<r_-<r_+<r_c.
$$
The event horizon ${\cal H}$ corresponds to the hypersurface where $r=r_+$,
the Cauchy horizon ${\cal CH}$ corresponds to the hypersurface where $r=r_-$,
and the cosmological horizon corresponds to the hypersurface $r=r_c$.
We are interested in studying solutions of the wave equation in the 
black hole region,
$r_-<r<r_+$, and we wish
to construct double null coordinates and understand their relation with
Boyer--Lindquist coordinates.
The coordinate $t$ takes values in $\mathbb{R}$.

This article follows closely the strategy and tools developed in~\cite{m-luk}.
Double null coordinate systems were constructed by 
Pretorius and Israel~\cite{pretorius} for Kerr spacetimes,
by Balushi and Mann~\cite{Mann1} for Kerr--(anti) de Sitter spacetimes,
and by Imseis, Balushi and Mann~\cite{mann2}
for Kerr--Newman--(anti) de Sitter spacetimes.
In~\cite{Mann1} and~\cite{mann2} the authors also study the formation of caustics.
In Section~\ref{null} we construct a double null coordinate
system for Kerr--Newman--de Sitter spacetimes. This construction only differs
from the one in~\cite{Mann1} and~\cite{mann2} (that we were
unaware of until the completion of this work)
 in the choice of 
$\lambda$ (our choice $\lambda=\sin^2\theta_\star$ is identical to the one in~\cite{m-luk} and~\cite{pretorius}).
Consider the transformation $(t,r,\theta,\phi)\mapsto(t,r_\star,\theta_\star,\phi)$
in the black hole region $r_-<r<r_+$, where $\Delta_r<0$.
The coordinate $\theta_{\star}$ is defined implicitly as the solution
of $F(r,\theta,\theta_{\star})=0$, where $F$ is given by 
$$
F(r, \theta,\theta_{\star})
=\int_{\theta_{\star}}^{\theta} \frac{\md \theta'}
{a\sqrt{\sin^2\theta_{\star}\Delta_{\theta'}-\sin^2\theta'}}
+\int_r^{r_+}\frac{\md r'}{\sqrt{((r')^2+a^2)^2-a^2\sin^2\theta_{\star}\Delta_{r'}}}.
$$
The coordinate 
$r_\star$
is defined by $r_\star=\varrho(r,\theta,\sin^2\theta_{\star}(r,\theta))$, 
where
$\varrho$ is given by 
\begin{eqnarray*}
\varrho(r,\theta,\lambda)&=&
\int^r_{r_0}\frac{(r^{\prime 2}+a^2)}{\Delta_{r'}}\md r'+
\int_r^{r_+}\frac{(r^{\prime 2}+a^2)-
	\sqrt{((r')^2+a^2)^2-a^2\lambda\Delta_{r'}}}{\Delta_{r'}}\md r'
\nonumber\\
&&+\int_0^\theta \frac{a\sqrt{\lambda\Delta_{\theta'}-\sin^2\theta'}}
{\Delta_{\theta'}}\md\theta'+\frac{a^2}{2}f(\lambda),
\end{eqnarray*}
for some fixed $r_0\in(r_-,r_+)$, and where
$f$ is the function which satisfies
$f(0)=0$ and
$$ 
f'(\lambda)=-
\int_0^{\arcsin\sqrt{\lambda}}
\frac{a\sqrt{\lambda\Delta_{\theta'}-\sin^2\theta'}}{\Delta_{\theta'}}{\md \theta'}
$$ 
(note the difference, $\varrho$ versus $\rho$ in~\eqref{rt_metric}). 

The regularity of transformation of coordinates  $(r_\star,\theta_\star)\mapsto(r,\theta)$ for Kerr spacetimes,
namely at $\theta_\star=0$,
was shown by Dafermos and Luk~\cite{m-luk}.  
We adapt their work to the setting of Kerr--Newman--de Sitter spacetimes.
We check that $\theta_{\star}$ is well defined and continuous
at $\theta=0$ and $\theta=\frac{\pi}{2}$ with
\be\label{aim}
\frac{\sin\theta_{\star}}{\sin\theta}+\frac{\cos\theta}{\cos\theta_{\star}}
\lesssim 1.
\ee
The proof of~\eqref{aim} requires that we use conditions characterizing subextremal 
black holes which are deduced in Appendix~\ref{charge_appendix}. Namely,
we use 
\be\label{Xii}
\Xi<1+\frac{1}{3}\,l\left(\frac{r_+}{r_-},\Lambda e^2\right),
\ee
where the function $l$ is given by~\eqref{l}.
This implies the inequality
$$
\Xi<\csc^2\left(\arctan\sqrt{\frac{r_+}{r_-}}+\arctan\sqrt{\frac{r_-}{r_+}}\right),
$$
which in turn implies~\eqref{aim}. We use the fact that $\Lambda$ is nonnegative
so that our computations are not immediately applicable to the setting of Kerr--AdS.

We also show
that $r$ and $\theta$ are smooth functions of $r_\star$ and $\theta_{\star}$.
When the cosmological constant is equal to zero
and $e=0$,
$\Delta_{\theta }$ is equal to $1$
and our formulas reduce to 
the ones for the Kerr spacetime in~\cite{m-luk}.
The trigonometric identity
$$
\sin ^2(2 \theta_{\star}) \Delta_\theta=-\sin(2\theta)\,\partial_\theta D(\theta,\theta_{\star})+
2 (\cos (2 \theta)+\cos (2 \theta_{\star}))D(\theta,\theta_{\star}),
$$
for $D(\theta,\theta_{\star})=\sin^2\theta_{\star}\Delta_{\theta }-\sin^2\theta$
is the key to 
the calculation of $\partial_r\theta_\star$ and $\partial_\theta\theta_\star$,
as well as 
the successful completion of some new identities, such as~\eqref{ddtheta} and~\eqref{d_sqrt}, which
 we need in order to calculate the derivatives of $r$
and $\theta$ with respect to $r_\star$ and $\theta_\star$.
 We would like to emphasize that 
our calculations are successful because the dependence
of $\Delta_{\theta }$ on $\theta$ occurs through
$\sin^2\theta$,
and not through $\sin\theta$, or on any other non-smooth
function of $\sin^2\theta$.
(This is a reflection of the fact that the Kerr--Newman--de Sitter metric 
is regular on a manifold diffeomorphic to $\mathbb{R}^2\times \mathbb{S}^2_{\theta,\phi}$,
i.e.\ it is regular on the full Boyer--Lindquist spheres of constant time coordinate $t$
and radial coordinate $r$.)
We also obtain bounds on the derivatives
of $r$ and $\theta$ which we need later on. These bounds 
parallel the ones in~\cite{m-luk}.

Using
$$ 
u=\frac{r_\star-t}{2}\qquad\mbox{and}\qquad v=\frac{t+r_\star}{2},
$$ 
the final transformation 
$$
\phi_\star=\phi-h(r_\star,\phi_\star),
$$
with $h$ given by
$$
\partial_{r_{\star}} h(r_{\star},\theta _{\star})=-\,\frac{\Xi a ((r^2+a^2) \Delta _{\theta }-\Delta _r)}
{(r^2+a^2)^2\Delta_{\theta}-a^2\sin^2\theta\Delta_{r}},\qquad\mbox{with}\ h(0,\theta_{\star})=0,
$$
allows one to
bring the metric to the double null form
\begin{eqnarray}
g&=&-2\Omega^2 \left(\mmd u\otimes \mmd v+\md v\otimes \mmd u\right)
\nonumber\\
&&+\gamma_{\theta_{\star}\theta_{\star}}
\md\theta_{\star}\otimes \mmd\theta_{\star}
+\gamma_{\theta_{\star}\phi_{\star}}\md\theta_{\star}\otimes(\mmd\phi_{\star}-b^{\phi_{\star}}\md v)
+\gamma_{\theta_{\star}\phi_{\star}}(\mmd\phi_{\star}-b^{\phi_{\star}}\md v)\otimes
\mmd\theta_{\star}\nonumber\\
&&+\gamma_{\phi_{\star}\phi_{\star}}
(\mmd\phi_{\star}-b^{\phi_{\star}}\md v)\otimes(\mmd\phi_{\star}-b^{\phi_{\star}}\md v),\label{double_null}
\end{eqnarray}
which one can use to carry out energy estimates.

Neither the Boyer--Lindquist coordinate system $(t,r,\theta,\phi)$ nor the double null
coordinate system $(u,v,\theta_\star,\phi_\star)$ cover the axis $\theta_\star=0$,
obviously. But they can be naturally extended to an atlas that does
cover $\theta_\star=0$ using 
$$
\left\{
\begin{array}{rcl}
x&=&\sin\theta_\star\cos\phi,\\
y&=&\sin\theta_\star\sin\phi,
\end{array}\right.\qquad\qquad\left\{
\begin{array}{rcl}
\tilde{x}&=&\sin\theta\cos\phi,\\
\tilde{y}&=&\sin\theta\sin\phi.
\end{array}\right.
$$
In Subsection~\ref{polo}, we prove that the two atlases 
${\cal A}_{\mbox{\tiny BL}}=\{(t,r,\theta,\phi),(t,r,\tilde{x},\tilde{y})\}$ and
${\cal A}_{\mbox{\tiny DN}}=\{(u,v,\theta_\star,\phi_\star),(u,v,x,y)\}$
are compatible (which would be clear if we were to exclude the points
where $\theta=\theta_\star=0$ and $\theta=\theta_\star=\pi$ from our manifold).
This implies that the two-spheres given by the intersection of the
hypersurfaces $u=\mbox{constant}$ and $v=\mbox{constant}$ are $C^\infty$
with respect to ${\cal A}_{\mbox{\tiny BL}}$ (see Theorem~\ref{thm1}).

Following~\cite{m-luk}, we analyze the decay of $\Omega^2$ at the future event and Cauchy horizons,
${\cal H}^+$ and~${\cal CH}^+$.
Finally, we give regular coordinates at ${\cal H}^+$ and~${\cal CH}^+$.
The Christoffel symbols of $g$ 
in the double null coordinates $(u,v,\theta_\star,\phi_\star)$
are given in Appendix~\ref{Chris},
along with some covariant derivatives that are needed to carry out energy 
estimates.

In Section~\ref{aplicacao},
using the vector field method, we study, in the black hole interior,
the energy of solutions of the wave equation 
which have compact support on ${\cal H}^+$.
We apply the form~\eqref{double_null} of the metric to construct certain blue-shift and red-shift 
vector fields and to calculate their covariant derivatives.
We obtain the usual inequalities relating the vector and scalar currents
associated to these vector fields. This allows us to 
illustrate how Sbierski's result in~\cite{Sbierski} 
applies to Kerr--Newman--de Sitter spacetimes.

This work is a first step of our broader project to generalize 
the results of~\cite{CF}, which provides a sufficient condition, in terms of surface gravities and a parameter for an exponential decaying Price Law, for energy of waves to remain bounded up to ${\cal CH}^+$.
The work~\cite{CF} used the fact that the generators of spherical symmetry are three
Killing vector fields, which is not true in the context of 
Kerr--Newman--de Sitter
spacetimes. We expect to address this in a forthcoming paper. 

An alternative approach towards extending the results of~\cite{CF} to the
Kerr--Newman--de Sitter setting would be to work in Boyer--Lindquist coordinates
as is done in the work~\cite{luk-sbierski} on Kerr black hole interiors.
However, as an example,
the double null coordinates allow one to immediately extend the results in~\cite{anne_kerr} to KNdS.

In Appendix~\ref{charge_appendix}, we characterize subextremal
Kerr--Newman--de Sitter black holes in terms of $(r_-, r_+,\Lambda a^2,\Lambda e^2)$, proving~\eqref{Xii}, in particular, as mentioned above.
The subset of $\bbR^3$ where one can choose
$\left(\frac{r_+}{r_-},\Lambda a^2,\Lambda e^2\right)$ is sketched in
Figure~\ref{region} on page~\pageref{region}. 
We make additional remarks concerning 
alternative choices of parameters,
namely $\left(\Lambda,\frac{r_+}{r_-},a,e\right)$ or
$(\Lambda,M,a,e)$. Related characterizations of the
parameters of subextremal Kerr--de Sitter solutions, for the case when there is
no charge, can be found in Lake and Zannias~\cite{lake} 
and Borthwick~\cite{borthwick}.

Hintz and Vasy give a uniform analysis of linear
waves up to the Cauchy horizon using methods from scattering
theory and microlocal analysis in~\cite{peter2}.
Moreover, Hintz proves non-linear stability of the Kerr--Newman--de Sitter family of charged black holes in~\cite{peter1}.

\section{A double null coordinate system}
\label{null}

\subsection{Construction of the double null coordinate system}

\subsubsection{The coordinates $r_\star$ and $\theta_{\star}$}

Given the manifold $\mathbb{R}^2\times \mathbb{S}^2$,
with metric~\eqref{rt_metric},
we look for a function $r_\star$ such that the axisymmetric hypersurface,
$$
v(t,r,\theta)=t\pm r_\star(r,\theta)=\mbox{constant}
$$
(ingoing when the plus sign is chosen, and outgoing
when when the minus sign is chosen), is lightlike. Then, the function
$v$ must satisfy the eikonal equation
\begin{eqnarray*}
	g^{\alpha \beta}(\partial_{\alpha} v)(\partial_{\beta} v)
	=\frac{1}{\rho^2}\left[\Delta_r (\partial_r {r}_{\star})^2+\Delta_{\theta}(\partial_{\theta} {r}_{\star})^2
	-\,\frac{1}{\Delta_{r}\Delta_{\theta}}\left((r^2+a^2)^2\Delta_{\theta}-a^2\sin^2\theta\Delta_{r}\right)\right]=0.
\end{eqnarray*}
We follow~\cite{pretorius} and construct particular separable solutions of
the eikonal equation.
We define $P$ and $Q$ by
\begin{eqnarray*}
	P(\theta,\theta_{\star})&=&a\sqrt{\sin^2\theta_{\star}\Delta_{\theta}-\sin^2\theta}, \\
	Q(r,\theta_{\star})&=&\sqrt{(r^2+a^2)^2-a^2\sin^2\theta_{\star}\Delta_r}.
\end{eqnarray*}
Note that
\be
\frac{1}{\Delta_{r}\Delta_{\theta}}\left((r^2+a^2)^2\Delta_{\theta}-a^2\sin^2\theta\Delta_{r}\right)=\frac{Q^2}{\Delta_{r}}+\frac{P^2}{\Delta_{\theta}},\label{sum}
\ee
and so the eikonal equation becomes
$$
\Delta_r (\partial_r {r}_{\star})^2+\Delta_{\theta}(\partial_{\theta} {r}_{\star})^2=\frac{Q^2}{\Delta_{r}}+\frac{P^2}{\Delta_{\theta}}.
$$
As $P$ is independent of $r$, and $Q$ is independent of $\theta$,
we look for special solutions $r_\star$ of this equation, where
\be\label{QP}
\partial_r {r}_{\star}=\frac{Q}{\Delta_r}\quad\mbox{and}\quad
\partial_{\theta} {r}_{\star}=\frac{P}{\Delta_\theta},
\ee
so that
\be\label{dr*}
\md r_{\star}=\frac{Q}{\Delta_r}\md r+\frac{P}{\Delta_\theta}\md\theta.
\ee
Both $P$ and $Q$ depend on (what is so far the parameter) $\theta_\star$,
which arises because of the degree of freedom one has in
breaking up the left-hand side of~\eqref{sum} to a sum. Indeed, 
to the left-hand side of~\eqref{sum}
we subtracted and added 
the quantity $a^2\sin^2\theta_{\star}$ (which is independent of
both $r$ and $\theta$) and then we decomposed the resulting
expression into a 
sum of a function depending solely on $r$ and a function depending
solely on $\theta$. 
We integrate~\eqref{dr*} and obtain
$$
r_\star=\int_{r_+}^r\frac{Q(r',\theta_{\star})}{\Delta_{r'}}\md r'+\int_0^\theta\frac{P(\theta',\theta_{\star})}{\Delta_{\theta'}}\md\theta'+
\frac{a^2}{2}f(\sin^2\theta_{\star}),
$$
where the function $f$ accounts for an integration constant.
Thus we have
\be\label{r*}
r_\star=\varrho(r,\theta,\lambda),
\ee
$\varrho:(r_-,r_+)\times[0,\pi]\times[0,1]\to\mathbb{R}$,
with
\begin{eqnarray*}
	\varrho(r,\theta,\lambda)&=&
	\int^r_{r_0}\frac{(r^{\prime 2}+a^2)}{\Delta_{r'}}\md r'+
	\int_r^{r_+}\frac{(r^{\prime 2}+a^2)-\hat{Q}(r',\lambda)}{\Delta_{r'}}\md r'\\
	&&+\int_0^\theta \frac{\hat{P}(\theta',\lambda)}{\Delta_{\theta'}}\md\theta'+\frac{a^2}{2}f(\lambda),
\end{eqnarray*}
for some fixed $r_0\in(r_-,r_+)$,
and
$$
\lambda=\sin^2\theta_{\star}.
$$
The expression for $\varrho$ is written so that the second integral converges.
Here 
\bea
\nonumber
\hat{P}^2(\theta,\lambda)&=&a^2\left(\lambda\Delta_{\theta}-\sin^2\theta\right), \\
\nonumber
\hat{Q}^2(r,\lambda)&=&\left((r^2+a^2)^2-a^2\lambda\Delta_r\right),
\eea
so that $\hat{P}(\theta,\lambda)=
P\bigl(\theta,\arcsin\sqrt{\lambda}\bigr)
=P(\theta,\theta_{\star})$ and
$\hat{Q}(r,\lambda)=
Q\bigl(r,\arcsin\sqrt{\lambda}\bigr)=Q(r,\theta_{\star})$.

For each fixed $\lambda$, \eqref{r*} is a solution of~\eqref{dr*}.
We now proceed to obtain another solution of~\eqref{dr*}.
Calculating the differential of $\varrho$, we obtain
$$
\md\varrho=\frac{Q}{\Delta_r}\md r+\frac{P}{\Delta_\theta}\md\theta+
\partial_\lambda\varrho\md\lambda,
$$
where
$$
\partial_\lambda\varrho=\frac{a^2}{2}\left(
\int_0^\theta\frac{\md\theta'}{\hat{P}(\theta',\lambda)}+
\int_r^{r_+}\frac{\md r'}{\hat{Q}(r',\lambda)}+f'(\lambda)
\right).
$$
Define the function
$$
F(r,\theta,\theta_\star):=
\int_0^\theta\frac{\md\theta'}{P(\theta',\theta_{\star})}+
\int_r^{r_+}\frac{\md r'}{Q(r',\theta_{\star})}+f'(\sin^2\theta_{\star}).
$$
Note that the function $f$ is still free. We choose $f:[0,1]\to\mathbb{R}$ to be the function which satisfies $f(0)=0$ and
$$ 
f'(\lambda)=-
\int_0^{\arcsin\sqrt{\lambda}}
\frac{\md\theta'}{\hat{P}\bigl(\theta', \lambda)},\qquad\mbox{i.e.}\
f'(\sin^2\theta_{\star})=-\int_{0}^{\theta_{\star}}\frac{\md \theta'}{P(\theta', \theta_{\star})}.
$$
The function $f$ is bounded.
Then the expression for $F$ becomes
\bea
\label{F_imp_def2}
F(r, \theta,\theta_{\star})
&=&\int_{\theta_{\star}}^{\theta} \frac{\md \theta'}{P(\theta', \theta_{\star})}+\int_r^{r_+}\frac{\md r'}{Q(r', \theta_{\star})}.
\eea

Adapting the construction of~\cite{m-luk} to our case,
for $r\in[r_-,r_+]$ and $\theta \in (0, \frac{\pi}{2})$, 
we define $\theta_{\star}\in [\theta, \frac{\pi}{2})$ implicitly to be the solution of  
\begin{equation}\label{F-equal-zero}
F(r, \theta,\theta_{\star})=0
\end{equation}
(see Lemma~\ref{theta-well-defined}).
Also, let 
$$
\theta_{\star}(r,0)=0,\quad \theta_{\star}\left(r,\frac{\pi}{2}\right)=
\frac{\pi}{2}
$$
and
\be\label{mirror}
\theta_{\star}(r,\theta)=\pi-\theta_{\star}(r,\pi-\theta)\ \mbox{for}\ 
\theta\in\left(\frac{\pi}{2},\pi\right].
\ee
Then, the function
$$
r_\star=\varrho(r,\theta,\sin^2\theta_{\star}(r,\theta))
$$
is another solution of~\eqref{dr*}. 
The functions 
\bea
\label{u_v}
\qquad u=\frac{r_{\star}(r,\theta)-t}{2}\quad\mbox{and}\quad v=\frac{t+{r}_{\star}(r, \theta)}{2}
\eea
are solutions of the eikonal equation.
Just as in the case of Kerr, it turns out
that $\theta_{\star}$ 
is an appropriate angle coordinate. 
This can be understood starting with the construction of~\cite{pretorius}:
when $\Lambda=M=a=e=0$ (so that we are reduced to the Minkowski spacetime),
$\theta_{\star}$ is the spherical polar angle. Moreover, 
for $r$ close to $r_+$, $\theta_{\star}$ is close to
$\theta$ (see~\eqref{F_imp_def2} and~\eqref{F-equal-zero}).
The function $\theta_\star$ is
interpreted as the spherical polar angle and the hypersurfaces where
$u$ and $v$ are constant are called quasi-spherical light cones.
From
\bea
g^{\alpha \beta}(\partial_{\alpha} r)(\partial_{\beta} r)
&=&\frac{1}{\rho^2}
(\Delta_r (\partial_r {r}_{\star})^2+\Delta_{\theta}(\partial_{\theta} {r}_{\star})^2)\ =\ 
\frac{1}{\rho^2\Delta_r\Delta_{\theta}}(\Delta_rP^2+\Delta_{\theta}Q^2)
\nonumber\\
&=&\frac{1}{\rho^2\Delta_{r}\Delta_{\theta}}\left((r^2+a^2)^2\Delta_{\theta}-a^2\sin^2\theta\Delta_{r}\right)\ <\ 0\nonumber
\eea
(recall that $\Delta_r<0$),
a hypersurface where $r_{\star}$ equals a constant is spacelike.
\begin{rmk}\label{swing}
	Note that $r_{\star}$ ranges between $-\infty$ and $+\infty$, as $r$ ranges between $r_+$ and $r_-$. More precisely, given $L>0$, there exists $\delta>0$
	such that $r_\star(r,\theta)>L$ for all $(r,\theta)\in(r_-,r_-+\delta)\times[0,\pi]$.  Moreover, given $\delta>0$,
	there exists $L>0$ such that $r_\star(r,\theta)>L$ implies that
	$r\in(r_-,r_-+\delta)$. In\/ {\rm Lemma~\ref{sharon}} we will prove that
	$(r,\theta)\mapsto(r_\star,\theta_\star)$ is invertible.
	So, we are observing that\/ $\lim_{r\searrow r_-}r_\star(r,\theta)=+\infty$
	and that\/ $\lim_{r_\star\nearrow +\infty}r(r_\star,\theta_\star)=r_-$, and that these
	limits are uniform in $\theta$ and in $\theta_\star$, respectively, for $\theta$ and $\theta_\star$ in $[0,\pi]$. Analogous statements can be made
	for the other endpoint.
\end{rmk}
The behavior of the coordinates $u$ and $v$ is sketched in Figure~\ref{ll}.
\begin{figure}[ht]
	\begin{psfrags}
		\psfrag{c}{{\small $u=+\infty$}}
		\psfrag{m}{{\small $v=-\infty$}}
		\psfrag{e}{{\small $u=-\infty$}}
		\psfrag{h}{{\small $v=+\infty$}}
		\centering
		\includegraphics[scale=.8]{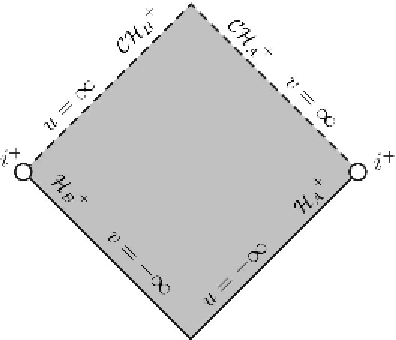}
		\caption{Behavior of the coordinates $u$ and $v$.}
		\label{ll}
	\end{psfrags}
\end{figure}

\begin{rmk}
	Because of the symmetry in~\eqref{mirror}, our statements 
	about $\theta$ and $\theta_\star$ will refer to the interval
	$\left(0,\frac{\pi}{2}\right)$ and it is understood that
	corresponding statements will hold in $\left(\frac{\pi}{2},\pi\right)$.
\end{rmk}

\subsubsection{The metric in $(t,r_\star,\theta_{\star},\phi)$ coordinates}
Denote 
\bea
\label{Upsilon}
\Upsilon:= \Delta_{r}P^2+\Delta_{\theta}Q^2= (r^2+a^2)^2\Delta_{\theta}-a^2\sin^2\theta\Delta_{r},
\eea
	$$
G:=\partial_{\theta_{\star}} F
$$  
and
\be
\label{L}
L:= -GPQ.
\ee
Differentiating both sides of~\eqref{F_imp_def2} with respect to $r$ and $\theta$ yields
\be\label{UQP}
\partial_r\theta_\star = \frac{1}{GQ},\qquad\partial_\theta\theta_\star=-\,\frac{1}{GP}.
\ee
Since
$$
\left|
\begin{array}{cc}
\partial_rr_\star&\partial_\theta r_\star\\
\partial_r\theta_\star&\partial_\theta \theta_\star
\end{array}
\right|=
\frac{\Upsilon}{\Delta_r\Delta_{\theta}L},
$$
the differentials of $r$ and $\theta$ are given by
\bea
\label{differentials}
\md r&=&\frac{\Delta_{r}\Delta_{\theta}Q}{\Upsilon}\md r_\star-\,\frac{\Delta_{r}LP}{\Upsilon}\md \theta_{\star},\\
\label{differentials2}
\md \theta &=& \frac{\Delta_{r}\Delta_{\theta}P}{\Upsilon}\md r_\star+\frac{\Delta_{\theta}LQ}{\Upsilon}\md \theta_{\star}.
\eea
To write the metric in $(t,r_{\star},\theta_{\star},\phi)$ coordinates one 
uses~\eqref{differentials} and~\eqref{differentials2} in~\eqref{rt_metric} 
and obtains
\begin{eqnarray*}
		g&=&\frac{\rho^2\Delta_r\Delta_{\theta}}{\Upsilon}\md r_{\star}^2+\frac{\rho^2L^2}{\Upsilon}\md \theta_{\star}^2
	+\frac{1}{\rho^2}(a^2\sin^2\theta\Delta_{\theta}-\Delta_r)\md t^2 +\frac{\sin^2\theta}{\Xi^2\rho^2}\Upsilon\md \phi^2\nonumber\\
	&&
	-\,\frac{2a\sin^2\theta}{\Xi\rho^2}((r^2+a^2)\Delta_{\theta}-\Delta_r)\md \phi \md t\\
	&=&g_{r_{\star} r_{\star}} \left(\mmd r_{\star}\otimes \mmd r_{\star}-\md t\otimes \mmd t\right)
	+g_{\theta _{\star} \theta _{\star}}\md\theta_{\star}\otimes \mmd\theta_{\star} 
	+g_{\phi \phi }(\mmd\phi-B\md t)\otimes(\mmd\phi-B\md t),
\end{eqnarray*}
with
$$
B=-\,\frac{g_{\phi t}}{g_{\phi\phi}}=\frac{\Xi a((r^2+a^2) \Delta _{\theta }-\Delta _r)}{\Upsilon}.
$$

\subsubsection{Definition of $\phi_{\star}$ and the metric in double null coordinates $(u,v,\theta_{\star},\phi_{\star})$}\label{cha}
From~\eqref{u_v}, one gets
$$
\left\{\begin{array}{lcl}
dr_{\star}&=&\md u+\md v,\\
dt&=&\md v-\md u.
\end{array}\right.
$$
Now introduce a new coordinate $\phi_{\star}$, defined by
$$
\phi_{\star}=\phi -h(r_{\star},\theta_{\star}),
$$
where
\be\label{prh}
\partial_{r_{\star}} h(r_{\star},\theta _{\star})=-B=-\,\frac{\Xi a ((r^2+a^2) \Delta _{\theta }-\Delta _r)}{\Upsilon },\qquad\mbox{with}\ h(0,\theta_{\star})=0.
\ee
For a general function $f$, one has
$$
\hat{f}(\hat{u},\hat{v},\hat{\theta}_\star,\phi)=
\hat{f}(\hat{u},\hat{v},\hat{\theta}_\star,\phi_\star+h(\hat{u}+\hat{v},\hat{\theta}_\star))=
f(u,v,\theta_\star,\phi_\star),
$$
where $\hat{f}$ is the function $f$ written in the $(\hat{u},\hat{v},\hat{\theta}_\star,\phi)$ coordinate system and
$\hat{u}=u$, $\hat{v}=v$, $\hat{\theta}_\star=\theta_\star$.
So, it follows that
\begin{eqnarray*}
	\partial_{u}&=&\partial_{\hat{u}}+\partial_{r_\star}h\,\partial_\phi,\\
	\partial_{v}&=&\partial_{\hat{v}}+\partial_{r_\star}h\,\partial_\phi,\\
	\partial_{\theta_\star}&=&\partial_{\hat{\theta}_\star}+\partial_{\theta_\star}h\,\partial_\phi,\\
	\partial_{\phi_\star}&=&\partial_\phi.
\end{eqnarray*}
These equations help us with the geometric interpretation of the change of 
coordinates operated by passing from $\phi$ to $\phi_\star$. 
Of course, for functions $f$ that do not depend on $\phi$, like the coefficients
of our metric, $\partial_uf=\partial_{\hat{u}}f$, $\partial_vf=\partial_{\hat{v}}f$
and $\partial_{\theta_\star}f=\partial_{\hat{\theta}_\star}f$.
Defining 
\begin{eqnarray}
\Omega^2&=&-g_{r_{\star}r_{\star}}=-\,\frac{\rho^2\Delta_r\Delta_\theta}{\Upsilon},\label{Omega-squared}\\
b^{\phi_{\star}}&=&2B=2\frac{\Xi a((r^2+a^2)\Delta_{\theta }-\Delta_r)}{\Upsilon},\label{b-phi}
\end{eqnarray}
the expression for the metric becomes
\begin{eqnarray}
g&=&-2\Omega^2 \left(\mmd u\otimes \mmd v+\md v\otimes \mmd u\right)
\nonumber\\
&&+\gamma_{\theta_{\star}\theta_{\star}}
\md\theta_{\star}\otimes \mmd\theta_{\star}
+\gamma_{\theta_{\star}\phi_{\star}}\md\theta_{\star}\otimes(\mmd\phi_{\star}-b^{\phi_{\star}}\md v)
+\gamma_{\theta_{\star}\phi_{\star}}(\mmd\phi_{\star}-b^{\phi_{\star}}\md v)\otimes
\mmd\theta_{\star}\nonumber\\
&&+\gamma_{\phi_{\star}\phi_{\star}}
(\mmd\phi_{\star}-b^{\phi_{\star}}\md v)\otimes(\mmd\phi_{\star}-b^{\phi_{\star}}\md v).\label{metric-final}
\end{eqnarray}
For each pair $(u,v)$, $\gamma$ is a metric defined on a two-sphere.
The calculation above shows that the coefficients of this metric are
\begin{eqnarray}
\gamma_{\theta_{\star}\theta_{\star}}&=&g_{\theta_{\star}\theta_{\star}}
+g_{\phi\phi}(\partial_{\theta_{\star}}h)^2
\ =\
\frac{\rho^2L^2}{\Upsilon}+\frac{\sin^2\theta}{\Xi^2\rho^2}\Upsilon(\partial_{\theta_{\star}}h)^2,\label{gamma-tt}\\
\gamma_{\theta_{\star}\phi_{\star}}&=&g_{\phi\phi}(\partial_{\theta_{\star}}h)
\ =\ 
\frac{\sin^2\theta}{\Xi^2\rho^2}\Upsilon(\partial_{\theta_{\star}}h),
\label{gamma-tf}\\
\gamma_{\phi_{\star}\phi_{\star}}&=&g_{\phi\phi}\ =\
\frac{\sin^2\theta}{\Xi^2\rho^2}\Upsilon.\label{gamma-ff}
\end{eqnarray}
The determinants of the metrics $\gamma$ and $g$ are
\be\label{det-gamma}
\rm{det}\,\gamma=\frac{L^2\sin^2\theta}{\Xi^2}
\ee
and
$$
\mbox{det}\,g=-
4\Omega^4 \frac{L^2 \sin^2\theta}{\Xi^2},
$$
and the inverse of the metric $g$ is 
\begin{eqnarray*}
	g^{-1}&=&-\,\frac{1}{2\Omega^2}\left(\partial_u\otimes(\partial_v+b^{\phi_{\star}}\,
	\partial_{\phi_{\star}})+(\partial_v+b^{\phi_{\star}}\,
	\partial_{\phi_{\star}})\otimes\partial_u\right)\\
	&&+\gamma^{\theta_{\star}\theta_{\star}}\,\partial_{\theta_{\star}}\otimes\partial_{\theta_{\star}}
	+\gamma^{\theta_{\star}\phi_{\star}}\,\partial_{\theta_{\star}}\otimes\partial_{\phi_{\star}}
	+\gamma^{\theta_{\star}\phi_{\star}}\,\partial_{\phi_{\star}}\otimes\partial_{\theta_{\star}}
	+\gamma^{\phi_{\star}\phi_{\star}}\,\partial_{\phi_{\star}}\otimes\partial_{\phi_{\star}},
\end{eqnarray*}
with coefficients given by
$$
\left[
\begin{array}{cc}
\gamma^{\theta_{\star}\theta_{\star}}&\gamma^{\theta_{\star}\phi_{\star}}\\
\gamma^{\theta_{\star}\phi_{\star}}&\gamma^{\phi_{\star}\phi_{\star}}\\
\end{array}
\right]
=
\left[
\begin{array}{cc}
\frac{\Upsilon }{L^2 \rho ^2} & -\,\frac{\Upsilon}{L^2 \rho ^2} (\partial_{\theta_{\star}}h) \\
-\,\frac{\Upsilon}{L^2 \rho ^2} (\partial_{\theta_{\star}}h)& \frac{\Xi^2 \rho ^2}{\Upsilon  \sin ^2\theta} +\frac{\Upsilon}{L^2 \rho ^2}(\partial_{\theta_{\star}}h)^2\\
\end{array}
\right].
$$

\subsubsection{Normals to hypersurfaces and volume elements}

We finish this subsection by writing down the volume elements of 
hypersurfaces 
$$
\Sigma_{r_\star}=\{r_\star=\mbox{constant}\},\quad
\underline{\cal C}_v=\{v=\mbox{constant}\}\quad\mbox{and}\quad
{\cal C}_u=\{u=\mbox{constant}\}
$$
of our spacetime, corresponding to constant $r_\star$,
$v$ and $u$. As
$$
2\Omega^2(\mmd u\otimes\mmd v+\md v\otimes\mmd u)=\Omega^2(\mmd r_\star\otimes\mmd r_\star-
\mmd t\otimes\mmd t),
$$
recalling~\eqref{det-gamma} for the determinant of $\gamma$,
the volume element for $\Sigma_{r_\star}$ is
$$
\VS=\Omega\frac{L \sin\theta}{\Xi}\md\theta_{\star}\md\phi_{\star} \md t_{\star}.
$$
Our choice for the normals to constant $v$ and $u$ hypersurfaces are
\be\label{ncu}
\nCv=\partial_u\qquad\mbox{and}\qquad \nCu=\partial_v+b^{\phi_{\star}}\partial_{\phi_{\star}}.
\ee
We have that
$$
g\left(\partial_u,\frac{1}{2\Omega^2}\partial_v\right)=-1
\qquad\mbox{and}\qquad
g\left(\partial_v+b^{\phi_{\star}}\partial_{\phi_{\star}},\frac{1}{2\Omega^2}\partial_u\right)=-1
$$
and so, since the volume element associated to the metric $g$ is
\be\label{volume}
\dV=2\Omega^2 \frac{L \sin\theta}{\Xi}
\md u \md v\md \theta_{\star}\md\phi_{\star},
\ee
the volume elements associated to constant $v$ and $u$ hypersurfaces are
$$
\VCv=\frac{L \sin\theta}{\Xi}\md u\md\theta_{\star}\md\phi_{\star} 
$$
and
$$
\VCu=\frac{L \sin\theta}{\Xi}\md v\md\theta_{\star} \md\phi_{\star}.
$$

As $r_\star=u+v$, the tangent space to a hypersurface 
$\Sigma_{r_\star}$, where $r_\star$ is constant, is spanned by 
$\partial_v-\partial_u$, $\partial_{\theta_{\star}}$ and $\partial_{\phi_{\star}}$,
and
$$
n_{\Sigma_{r_\star}}=\frac{\partial_u+\partial_v+\bphi\partial_{\phi_{\star}}}{2\Omega}=-\,\frac{\grad r_\star}{(-g(\grad r_\star,\grad r_\star))^{\frac{1}{2}}}.
$$
\begin{rmk}
	$\partial_{\phi_\star}$ is equal to zero when $\theta_\star$ is either
	$0$ or $\pi$.
\end{rmk}
Indeed, the vector field $\partial_{\phi_\star}$ is tangent to the spheres 
$u=\mbox{constant}$ and $v=\mbox{constant}$ which are contained in the spacelike
hypersurfaces $r_\star=\mbox{constant}$ and
$$
\lim_{\theta_\star\to 0}g(\partial_{\phi_\star},\partial_{\phi_\star})
=\lim_{\theta_\star\to 0}\gamma_{\phi_\star\phi_\star}=0.
$$ 

\subsection{Regularity of the change of coordinates}
\label{regularity}

The regularity of transformation of coordinates  $(r_\star,\theta_\star)\mapsto(r,\theta)$ for Kerr spacetimes,
namely at $\theta_\star=0$ and $\theta_\star=\frac{\pi}{2}$,
was shown by Dafermos and Luk~\cite{m-luk}.  In this subsection
we adapt their work to the setting of Kerr--Newman--de Sitter spacetimes.

\subsubsection{The coordinate $\theta_{\star}$ is well defined and continuous}

\begin{lem}\label{theta-well-defined}
$\theta_{\star}$ is well defined.
\end{lem}
\begin{proof}
Clearly $$F(r,\theta,\theta)\geq 0.$$ 
Moreover
\begin{eqnarray*}
\int_{\theta}^{\theta_{\star}}\frac{\md \theta'}{P(\theta', \theta_{\star})}
&\geq&
\int_{\theta}^{\theta_{\star}}\frac{\md \theta'}{ a\sqrt{\Delta_{\theta'}-\sin^2\theta'}}\nonumber\\
&=&
\int_{\theta}^{\theta_{\star}}\frac{\md \theta'}{\Xi^{\frac{1}{2}}a\cos\theta'}\to+\infty\
\ \mbox{as}\ \theta_{\star}\to\frac{\pi}{2},
\end{eqnarray*}
which implies that
$$
\lim_{\theta_{\star}\to\frac{\pi}2}F(r,\theta,\theta_{\star})=-\infty.
$$
The continuity of $F$ implies that for each pair $(r,\theta)$, with $\theta\in\bigl(0,\frac{\pi}{2}\bigr)$, there exists $\theta_{\star}=\theta_{\star}(r,\theta)\in
\bigl(\theta,\frac{\pi}{2}\bigr)$
such that $F(r,\theta,\theta_{\star})=0$.

The function $F$ is decreasing in $\theta_{\star}$.
This follows from the fact that
$$
\theta_{\star}\to
\int_0^{\theta_\star}\frac{\md\theta'}{\sqrt{\sin^2\theta_{\star}\Delta_{\theta'}-\sin^2\theta'}}
$$
is strictly increasing. Indeed, suppose $\theta_{\star}<\tilde\theta_{\star}$. Setting
$\sin\tilde{\theta}'=\frac{\sin\tilde{\theta}_\star}{\sin\theta_{\star}}\sin\theta'$,
as $\theta'<\tilde{\theta}'$, we have $\Delta_{\theta'}>\Delta_{\tilde\theta'}$ and
\bea
&&\int_0^{\tilde\theta_\star}\frac{\md\tilde\theta'}{\sqrt{\sin^2\tilde\theta_{\star}\Delta_{\tilde\theta'}-\sin^2\tilde\theta'}}
-\int_0^{\theta_\star}\frac{\md\theta'}{\sqrt{\sin^2\theta_{\star}\Delta_{\theta'}-\sin^2\theta'}}
\nonumber\\
&&\ \ \geq \int_0^{\tilde\theta_\star}\frac{\md\tilde\theta'}{\sqrt{\sin^2\tilde\theta_{\star}\Delta_{\theta'}-\sin^2\tilde\theta'}}
-\int_0^{\theta_\star}\frac{\md\theta'}{\sqrt{\sin^2\theta_{\star}\Delta_{\theta'}-\sin^2\theta'}}
\nonumber\\
&&\ \ =\int_0^{\theta_\star}
\left(\frac{\cos\theta'}{\cos\tilde{\theta}'}-1\right)
\frac{\md\theta'}{\sqrt{\sin^2\theta_{\star}\Delta_{\theta'}-\sin^2\theta'}}\ \geq\ 0.
\nonumber
\eea
\end{proof}

\begin{lem}
$\theta_{\star}$ is continuous at $\theta=0$ and $\theta=\frac{\pi}{2}$ with
\be
\frac{\sin\theta_{\star}}{\sin\theta}+\frac{\cos\theta}{\cos\theta_{\star}}\lesssim 1.
\label{bound}
\ee
\end{lem}
\begin{proof}
The second integral in~\eqref{F_imp_def2} is bounded above by
\bea
\int_{r}^{r_+}\frac{\md r'}{Q(r',\theta_{\star})}&\leq&\int_{r_-}^{r_+}
\frac{\md r'}{r'^2+a^2}\nonumber\\
&=&\frac{1}{a}\left(\arctan\left(\frac{r_+}{a}\right)-
\arctan\left(\frac{r_-}{a}\right)\right)\nonumber\\
&=&\frac{1}{a}\left(\frac{\pi}{2}-c_{\Lambda,M,a,e}\right),\label{i_o_Q}
\eea
for some positive $c_{\Lambda,M,a,e}$.
On the other hand, 
using the substitution \mbox{$\sin\theta'=\sin\theta_{\star}\sin\tilde{\theta}$} we see that
the negative of the first integral in~\eqref{F_imp_def2} is bounded below by
\bea
\int_{\theta}^{\theta_{\star}}\frac{\md \theta'}{P(\theta', \theta_{\star})}
&=&\frac{1}{a}
\int_{\theta}^{\theta_{\star}}\frac{\md \theta'}{\sqrt{\sin^2\theta_{\star}\Delta_{\theta'}-
\sin^2\theta'}}\label{i_o_P}\\
&=&\frac{1}{a}\int_{\arcsin\left(\frac{\sin\theta}{\sin\theta_{\star}}\right)}^{\frac{\pi}{2}}
\frac{\cos\tilde{\theta}\,\md\tilde{\theta}}{\cos\theta'\sqrt{\Delta_{\theta'}-\sin^2\tilde{\theta}}}
\nonumber\\
&=&\frac{1}{a}\int_{\arcsin\left(\frac{\sin\theta}{\sin\theta_{\star}}\right)}^{\frac{\pi}{2}}
\frac{\cos\tilde{\theta}\,\md\tilde{\theta}}
{\sqrt{\left(1+\frac{\Lambda}{3}a^2\right)
-\left(1+\frac{\Lambda}{3}a^2\sin^2\theta_{\star}\right)\sin^2\tilde{\theta}}}
\nonumber\\
&\geq&\frac{1}{a}\int_{\frac{\sin\theta}{\sin\theta_{\star}}}^{1}\frac{\md x}{\sqrt{\Xi-x^2}}
\nonumber\\
&=&\frac{1}{a}\int_{\frac{\sin\theta}{\sqrt{\Xi}\sin\theta_{\star}}}^{
\frac{1}{\sqrt{\Xi}}}\frac{\md y}{\sqrt{1-y^2}}\nonumber\\
&=&\frac{1}{a}\left(
\arcsin\left(\frac{1}{\sqrt{\Xi}}\right)-\,
\arcsin\left(\frac{\sin\theta}{\sqrt{\Xi}\sin\theta_{\star}}\right)
\right)\nonumber\\
&\geq&\frac{1}{a}\left(
\arcsin\left(\frac{1}{\sqrt{\Xi}}\right)-\,
\arcsin\left(\frac{\sin\theta}{\sin\theta_{\star}}\right)
\right).\nonumber
\eea
Since $F(r,\theta,\theta_{\star})$ is equal zero, we obtain the estimate
\be\label{sinsin}
\arcsin\left(\frac{\sin\theta}{\sin\theta_{\star}}\right)\geq
\arcsin\left(\frac{1}{\sqrt{\Xi}}\right)+c_{\Lambda,M,a,e}-\,\frac{\pi}{2},
\ee
provided the right-hand side is positive, i.e.\
\bea
\Xi&<&\csc^2\left(\frac{\pi}{2}-c_{\Lambda,M,a,e}\right)\nonumber\\
&=&\csc^2
\left(\arctan\left(\frac{r_+}{a}\right)-
\arctan\left(\frac{r_-}{a}\right)\right).\label{lira}
\eea
For fixed $r_-$ and $r_+$, the minimum of the last expression is attained
when $a=\sqrt{r_-r_+}$. So the last inequality is implied
by the stronger restriction 
\be\label{Xi}
\Xi<\csc^2\left(\arctan\left(\sqrt{\alpha}\right)-
\arctan\left(\frac{1}{\sqrt{\alpha}}\right)\right),\qquad\alpha=\frac{r_+}{r_-}.
\ee
According to Lemma~\ref{subextremal}, for all subextremal black holes, 
we have
$$
\Xi<1+\frac{l(\alpha,\Lambda e^2)}{3}\leq 1+\frac{l(\alpha,0)}{3},
$$
where $l$ is given by~\eqref{l}, so that
\begin{eqnarray*}
l(\alpha,0)&:=&
\frac{3}{2 \alpha^2 (1+2 \alpha)}
\Bigl(\left(1+4 \alpha+12 \alpha^2+16 \alpha^3+9 \alpha^4\right)\nonumber\\
&&\qquad\qquad
-\left(1+3 \alpha+5 \alpha^2+3 \alpha^3\right) \sqrt{1+2\alpha+9\alpha^2}\Bigr).
\end{eqnarray*}

We claim that
$$
1+\frac{l(\alpha,0)}{3}<\csc^2\left(\arctan\left(\sqrt{\alpha}\right)-
\arctan\left(\frac{1}{\sqrt{\alpha}}\right)\right)=\left(\frac{\alpha+1}{\alpha-1}\right)^2,
$$
for all $\alpha$ greater than $1$.
To show this, we define the function
$$
\rhos(\alpha):=\frac{\bigl(\frac{\alpha+1}{\alpha-1}\bigr)^2-1}{\frac{l(\alpha,0)}{3}}=
	\frac{\frac{4\alpha}{(\alpha-1)^2}}{\frac{l(\alpha,0)}{3}},
$$
which we wish to check is always greater than one.
By direct 
computation one may check that 
$$
l(\alpha,0)\leq\frac{2}{3\alpha}.
$$
Indeed, \begin{eqnarray*}
l(\alpha,0)\leq\frac{2}{3\alpha}
&
\Leftrightarrow&
		\left(1+4 \alpha+12 \alpha^2+16 \alpha^3+9 \alpha^4\right)
		-\left(1+3 \alpha+5 \alpha^2+3 \alpha^3\right) \sqrt{1+2\alpha+9\alpha^2}\leq\frac{4}{9}\alpha(1+2\alpha)\\
		&\Leftrightarrow&
		1+\frac{32\alpha}{9}+\frac{100\alpha^2}{9}+16\alpha^3+9\alpha^4\leq
		\left(1+3 \alpha+5 \alpha^2+3 \alpha^3\right) \sqrt{1+2\alpha+9\alpha^2}.
	\end{eqnarray*}
	Since both sides of the last inequality are positive, the inequality
	holds if and only if
	the square of the right-hand side minus the
	square of the left-hand side is nonnegative.
	Calculating this difference, one finds out that the terms in
	$\alpha^8$, $\alpha^7$ and $\alpha^6$ cancel out. One is left with
	the polynomial
	$$
	\frac{4\alpha}{81}(18+104\alpha+344\alpha^2+623\alpha^3+414\alpha^4),
	$$
	whose coefficients are all positive. Hence, this polynomial is positive
	(for $\alpha>1$).
So,
$$
\rhos(\alpha)\geq\frac{18\alpha^2}{(\alpha-1)^2}\geq 18>1.
$$
The claim is proven.

The previous paragraph implies that~\eqref{Xi} is satisfied, and so
the right-hand side of~\eqref{sinsin} is positive. Then, we have 
\be\label{good2}
\frac{\sin\theta}{\sin\theta_{\star}}\geq c>0.
\ee

If we use the substitution $\cos\theta'=\cos\theta_{\star}\sec\tilde{\theta}$ in~\eqref{i_o_P} we get
\bea
\int_{\theta}^{\theta_{\star}}\frac{\md \theta'}{P(\theta', \theta_{\star})}
&=&\frac{1}{a}
\int_{\theta}^{\theta_{\star}}\frac{\md \theta'}{\sqrt{\sin^2\theta_{\star}\Delta_{\theta'}-
		\sin^2\theta'}}\nonumber\\
&=&\frac{1}{a}\int_0^{\arcsec\left(\frac{\cos\theta}{\cos\theta_{\star}}\right)}
	\frac{\sec\tilde\theta\cos\theta_\star\tan\tilde\theta\md\tilde\theta}
	{\sin\theta'\sqrt{\Delta_{\theta'}-1+\cos^2\theta_{\star}\sec^2\tilde{\theta}-\Delta_{\theta'}\cos^2\theta_{\star}}}\nonumber\\
&=&\frac{1}{a}\int_0^{\arcsec\left(\frac{\cos\theta}{\cos\theta_{\star}}\right)}
	\frac{\sec\tilde\theta\tan\tilde\theta\md\tilde\theta}
	{\sin\theta'\sqrt{\frac{\Lambda}{3}a^2\sec^2\tilde\theta+\sec^2\tilde{\theta}-\Delta_{\theta'}}}\nonumber\\
&=&\frac{1}{a}\int_0^{\arcsec\left(\frac{\cos\theta}{\cos\theta_{\star}}\right)}
	\frac{\sec\tilde\theta\tan\tilde\theta\md\tilde\theta}
	{\sin\theta'\sqrt{\frac{\Lambda}{3}a^2\sec^2\tilde\theta\sin^2\theta_{\star}+\tan^2\tilde{\theta}}}\nonumber\\
&=&\frac{1}{a}\int_0^{\arcsec\left(\frac{\cos\theta}{\cos\theta_{\star}}\right)}
\frac{\sec\tilde\theta\md\tilde\theta}
{\sin\theta'\sqrt{1+\frac{\Lambda}{3}a^2\frac{\sin^2\theta_{\star}}{\sin^2\tilde\theta}}}.\nonumber
\eea
If $\frac{\cos\theta}{\cos\theta_{\star}}\leq\sqrt{2}$, we are done.
Otherwise, as $$\frac{1}{\sin^2(\arcsec x)}=\frac{1}{1-\,\frac{1}{x^2}}$$ and
$\sin^2\theta_{\star}\leq 1$, we have that
\bea
\int_{\theta}^{\theta_{\star}}\frac{\md \theta'}{P(\theta', \theta_{\star})}
&\geq&
\frac{1}{a}\int_{\frac{\pi}{4}}^{\arcsec\left(\frac{\cos\theta}{\cos\theta_{\star}}\right)}
\frac{\sec\tilde\theta\md\tilde\theta}
{\sqrt{1+2\frac{\Lambda}{3}a^2}}\nonumber\\
&\geq&\frac{1}{\sqrt{2}\Xi^\frac{1}{2} a}
\ln(\sec x+\tan x)\Big|_{\frac{\pi}{4}}^{\arcsec\left(\frac{\cos\theta}{\cos\theta_{\star}}\right)}\nonumber\\
&=&\frac{1}{\sqrt{2}\Xi^\frac{1}{2} a}\left[
\ln\left(\frac{\cos\theta}{\cos\theta_{\star}}+
\sqrt{\frac{\cos^2\theta}{\cos^2\theta_{\star}}-1}\right)-\ln(1+\sqrt{2})
\right].\nonumber
\eea
The result follows since~\eqref{i_o_Q} implies that the last expression is bounded above.
\end{proof}

\subsubsection{The derivative of the function defining $\theta_\star$}

\begin{rmk} Let us define
	$$
	D(\theta,\theta_{\star}):=\sin ^2\theta_{\star} \Delta_\theta-\sin ^2\theta,
	$$
so that $D(\theta,\theta_{\star})=P^2(\theta,\theta_{\star})/a^2$.
Clearly, 
$$
\partial_\theta D(\theta,\theta_{\star})=
\partial_\theta\left(\sin ^2\theta_{\star} \left(1+\frac{\Lambda}{3} a^2   \cos ^2\theta\right)-\sin ^2\theta\right)
=-\sin(2\theta)
\left(1+\frac{\Lambda}{3} a^2   \sin ^2\theta_{\star}\right)
$$
holds.
The trigonometric equality
\bea
\sin ^2(2 \theta_{\star}) \left(1+\frac{\Lambda}{3} a^2   \cos ^2\theta\right)&=&\sin ^2(2 \theta) \left(1+\frac{\Lambda}{3} a^2   \sin ^2\theta_{\star}\right)\nonumber\\
&&+2 (\cos (2 \theta)+\cos (2 \theta_{\star})) \left(\sin ^2\theta_{\star} \left(1+\frac{\Lambda}{3} a^2   \cos ^2\theta\right)-\sin ^2\theta\right)\nonumber
\eea
allows us to conclude that
\bea
\sin ^2(2 \theta_{\star}) \Delta_\theta&=&-\sin (2 \theta)
\left[\partial_\theta\left(\sin ^2\theta_{\star} \Delta_\theta-\sin ^2\theta\right)\right]\nonumber\\
&&+2 (\cos (2 \theta)+\cos (2 \theta_{\star})) \left(\sin ^2\theta_{\star} \Delta_\theta-\sin ^2\theta\right)\nonumber\\
&=&-\sin(2\theta)\,\partial_\theta D(\theta,\theta_{\star})+
2 (\cos (2 \theta)+\cos (2 \theta_{\star}))D(\theta,\theta_{\star}).
\label{trick}
\eea
We will have to use the identity~\eqref{trick} repeatedly in our 
calculations.
\end{rmk}

\begin{lem}
For $\theta_{\star}\in\left(0,\frac{\pi}{2}\right)$, we have that
\bea
G(r,\theta,\theta_{\star})&=&\partial_{\theta_{\star}}
F(r,\theta,\theta_{\star})
\nonumber\\
&=&
-\,\frac{\csc(2\theta_{\star})\sin(2\theta)}
{a\sqrt{\sin^2\theta_{\star}\Delta_{\theta}-\sin^2\theta}}
-\,\frac{2\csc(2\theta_{\star})}{a}
\int_{\theta}^{\theta_{\star}}
\frac{(\sin^2\theta_{\star}-\sin^2\theta')}
{\sqrt{\sin^2\theta_{\star}\Delta_{\theta'}-\sin^2\theta'}}\md\theta'
\nonumber\\
&&+\int_r^{r_+} \frac{a^2\sin(2\theta_{\star}) \Delta_{r'}}{2\left((r'^2+a^2)^2-a^2\sin^2\theta_{\star}\Delta_{r'}\right)^{\frac{3}{2}}}\md r'.\label{derivada}
\eea
\end{lem}
\begin{proof} 
	We start from the definition of $F$ in~\eqref{F_imp_def2}.
	Since $\theta_{\star}\in\left(0,\frac{\pi}{2}\right)$,
for $\theta'\in[\theta,\theta_{\star}]$, we get
\mbox{$\Delta_{\theta'}\geq 1+\frac{\Lambda}{3}a^2\cos^2\theta_{\star}$}, and so
$P(\theta',\theta_{\star})\geq \sqrt{\frac{\Lambda}{3}}\frac{a^2}{2}\sin{(2\theta_{\star})}>0$.
So, we do not have to worry about vanishing denominators.
The result follows from the calculation
\bea
&&a\left.\frac{\partial}{\partial\theta_{\star}}\right|_{\theta\mbox{\tiny\  fixed}}\int_{\theta}^{\theta_{\star}}\frac{\md \theta'}{P(\theta', \theta_{\star})}=\left.\frac{\partial}{\partial\theta_{\star}}\right|_{\theta\mbox{\tiny\  fixed}}\int_{\theta}^{\theta_{\star}}
\frac{\md \theta'}{\sqrt{\sin^2\theta_{\star}\Delta_{\theta'}-\sin^2\theta'}}\nonumber\\
&&\ \ =
\frac{2}{\sqrt{\frac{\Lambda}{3}}a\sin(2\theta_{\star})}-
\int_{\theta}^{\theta_{\star}}
\frac{\sin(2\theta_{\star})\Delta_{\theta'}}
{2\sqrt{\left(D(\theta',\theta_{\star})\right)^3}}\md\theta'
\nonumber\\
&&\ \ \stackrel{\eqref{trick}}{=}
\frac{2}{\sqrt{\frac{\Lambda}{3}}a\sin(2\theta_{\star})}+
\int_{\theta}^{\theta_{\star}}
\frac{\csc(2\theta_{\star})\sin(2\theta')
\partial_{\theta'}D(\theta',\theta_{\star})}
{2\sqrt{\left(D(\theta',\theta_{\star})\right)^3}}\md\theta'
\nonumber\\
&&\ \ \ \ \ \ -
\int_{\theta}^{\theta_{\star}}
\frac{\csc(2\theta_{\star})(\cos(2\theta')+\cos(2\theta_{\star}))}
{\sqrt{D(\theta',\theta_{\star})}}\md\theta'
\nonumber\\
&&\ \ =
\frac{2}{\sqrt{\frac{\Lambda}{3}}a\sin(2\theta_{\star})}+
\int_{\theta}^{\theta_{\star}}
\frac{\csc(2\theta_{\star})\sin(2\theta')
	\partial_{\theta'}D(\theta',\theta_{\star})
}
{2\sqrt{\left(D(\theta',\theta_{\star})\right)^3}}\md\theta'
\nonumber\\
&&\ \ \ \ \ \ -
\int_{\theta}^{\theta_{\star}}
\frac{2\csc(2\theta_{\star})\cos(2\theta')}
{\sqrt{D(\theta',\theta_{\star})}}\md\theta'+
\int_{\theta}^{\theta_{\star}}
\frac{\csc(2\theta_{\star})(\cos(2\theta')-\cos(2\theta_{\star}))}
{\sqrt{D(\theta',\theta_{\star})}}\md\theta'
\nonumber\\
&&\ \ =
\frac{2}{\sqrt{\frac{\Lambda}{3}}a\sin(2\theta_{\star})}-
\int_{\theta}^{\theta_{\star}}\frac{\partial}{\partial\theta'}
\frac{\csc(2\theta_{\star})\sin(2\theta')}
{\sqrt{D(\theta',\theta_{\star})}}\md\theta'
\nonumber\\
&&\ \ \ \ \ \ +
\int_{\theta}^{\theta_{\star}}
\frac{2\csc(2\theta_{\star})(\sin^2\theta_{\star}-\sin^2\theta')}
{\sqrt{D(\theta',\theta_{\star})}}\md\theta'
\nonumber\\
&&\ \ =
\frac{\csc(2\theta_{\star})\sin(2\theta)}
{\sqrt{D(\theta,\theta_{\star})}}
+2\csc(2\theta_{\star})
\int_{\theta}^{\theta_{\star}}
\frac{(\sin^2\theta_{\star}-\sin^2\theta')}
{\sqrt{D(\theta',\theta_{\star})}}\md\theta'.
\nonumber
\eea
For the second equality we used~\eqref{trick}.
\end{proof}

The expression for $G$ and the definition of $L$ in~\eqref{L} yield
\bea
L&=&\sqrt{(r^2+a^2)^2-a^2\sin^2\theta_{\star}\Delta_r}\Biggl(\frac{\sin(2\theta)}{\sin(2\theta_{\star})}\Biggr.\nonumber\\
&&
+2\sin^2(2\theta_{\star})\frac{\sqrt{\sin^2\theta_{\star}\Delta_{\theta}-\sin^2\theta}}{\sin(2\theta_{\star})}\left(\frac{1}{\sin^2(2\theta_{\star})}
\int_{\theta}^{\theta_{\star}}
\frac{(\sin^2\theta_{\star}-\sin^2\theta')}
{\sqrt{\sin^2\theta_{\star}\Delta_{\theta'}-\sin^2\theta'}}\md\theta'\right)
\nonumber\\
&&\Biggl. -\sin^2(2\theta_{\star})\frac{\sqrt{\sin^2\theta_{\star}\Delta_{\theta}-\sin^2\theta}}{\sin(2\theta_{\star})}\int_r^{r_+} \frac{a^3 \Delta_{r'}}{2\left((r'^2+a^2)^2-a^2\sin^2\theta_{\star}\Delta_{r'}\right)^{\frac{3}{2}}}\md r'\Biggr).\label{LLL}
\eea

\begin{lem}\label{G}
	We have the following estimate for $G$:
\bea
-\,\frac{C_{\Lambda,M,a,e}}{\sqrt{\sin^2\theta_{\star}\Delta_{\theta}-\sin^2\theta}}\leq G(r,\theta, \theta_{\star}(r,\theta))\leq -\,\frac{\csc(2\theta_{\star})\sin(2\theta)}{ a\sqrt{\sin^2\theta_{\star}\Delta_{\theta}-\sin^2\theta}}.\label{bounds_for_G}
\eea
for some constant $C_{\Lambda,M,a,e}>0$ depending on $\Lambda, M, a$ and $e$.
\end{lem}
\begin{proof} We use expression~\eqref{derivada}. Note that
$$
\frac{2\csc(2\theta_{\star})}{a}
\int_{\theta}^{\theta_{\star}}
\frac{(\sin^2\theta_{\star}-\sin^2\theta')}
{\sqrt{\sin^2\theta_{\star}\Delta_{\theta'}-\sin^2\theta'}}\md\theta'
\leq
\frac{2\csc(2\theta_{\star})}{a}
\int_{\theta}^{\theta_{\star}}
\sqrt{\sin^2\theta_{\star}-\sin^2\theta'}\md\theta'.
$$
For $\theta_{\star}$ close to zero, the right-hand side is bounded above by
$$
\frac{\pi/2}{a\cos\theta_{\star}},
$$
whereas for $\theta_{\star}$ close to $\frac{\pi}{2}$, the right-hand side is bounded above by
\bea
\frac{2\csc(2\theta_{\star})}{a}\int_{\theta}^{\theta_{\star}}\cos\theta'\md\theta'
&=&\frac{2\csc(2\theta_{\star})}{a}(\sin\theta_{\star}-\sin\theta)\nonumber\\
&\leq& \frac{2\csc(2\theta_{\star})}{a}(\theta_{\star}-\theta)\cos\theta\nonumber\\
&\leq&\frac{\pi/2}{a\sin\theta_{\star}}\frac{\cos\theta}{\cos\theta_\star}.\nonumber
\eea
Recalling~\eqref{bound}, we conclude that the second term on the right-hand side
of~\eqref{derivada} is uniformly bounded below by a negative constant. Clearly,
the same is true for the third  term on the right-hand side
of~\eqref{derivada}. Since both terms are negative,
$\csc(2\theta_{\star})\sin(2\theta)$ is bounded,
and $1/\sqrt{\sin^2\theta_{\star}\Delta_{\theta }-\sin^2\theta}$ is bounded below 
by a positive constant, we obtain~\eqref{bounds_for_G}.
\end{proof}

\subsubsection{The coordinates $r$ and $\theta$ as functions of $r_\star$ and $\theta_\star$}

\begin{lem}\label{sharon}
	The mapping $(r,\theta)\to(r_\star,\theta_\star)$,
	from $(r_-,r_+)\times\left[0,\frac{\pi}{2}\right]$ to
	$(-\infty,+\infty)\times\left[0,\frac{\pi}{2}\right]$,
	 is globally invertible.
\end{lem}
\begin{proof}
Fix $(r_\star)_0\in\bbR$.
For each fixed $\theta\in\bigl(0,\frac{\pi}{2}\bigr)$, let $r(\theta)$ be the
unique (because $\partial_rr_\star<0$) value of $r\in(r_-,r_+)$ such that
$$
r_\star(r(\theta),\theta)=(r_\star)_0.
$$
The Implicit Function Theorem guarantees that $(r(\,\cdot\,),\,\cdot\,)$ is a $C^1$ curve
and that 
$$
\dot{r}=-\,\frac{\partial_\theta r_\star}{\partial_rr_\star}=-\,\frac{P}{Q}
\frac{\Delta_r}{\Delta_\theta}.
$$
To see how $\theta_\star$ varies along this curve we calculate
$$
\frac{d}{d\theta}\theta_\star(r(\theta),\theta)=\dot{r}\,\partial_r\theta_\star+
\partial_\theta\theta_\star=
\frac{1}{GQ}\left(-\,\frac{P}{Q}
\frac{\Delta_r}{\Delta_\theta}\right)-\,\frac{1}{GP}=-\,\frac{\Upsilon}{GPQ^2}>0.
$$
So, $\theta_\star$ is strictly increasing along the curve $(r(\,\cdot\,),\,\cdot\,)$.
Moreover, \eqref{bound} shows that $\lim_{\theta\searrow 0}\theta_\star(r(\theta),\theta)=0$ and $\lim_{\theta\nearrow \frac{\pi}{2}}\theta_\star(r(\theta),\theta)=\frac{\pi}{2}$.
Thus, given $((r_\star)_0,(\theta_\star)_0)\in
(-\infty,+\infty)\times\left[0,\frac{\pi}{2}\right]$ there exists one and only one
$(r,\theta)$ such that $(r_\star(r,\theta),\theta_\star(r,\theta))=
((r_\star)_0,(\theta_\star)_0)$.
We conclude that the mapping 
$$(r_\star,\theta_\star):(r_-,r_+)\times\left[0,\frac{\pi}{2}\right]\longrightarrow
(-\infty,+\infty)\times\left[0,\frac{\pi}{2}\right]
$$
is one-to-one and onto. 

\end{proof}

\subsubsection{First partial derivatives}

\begin{lem}
	\label{A.5}
	The partial derivatives of $(r, {\theta})$ with respect to $(r_{\star}, \theta_{\star})$ are given by
	\begin{eqnarray}
	\frac{\partial r}{\partial r_{\star}}&=&\frac{\Delta _r
		\Delta_{\theta}Q}{\Upsilon }=\frac{\Delta _r \Delta _{\theta} \sqrt{(r^2+a^2)^2-a^2 \Delta _r \sin ^2{\theta_{\star}}}}{\left((r^2+a^2)^2 \Delta _{\theta}-a^2 \Delta _r \sin ^2{\theta}\right)},
	\label{r_r*}
	\end{eqnarray}
	\bea
	\nonumber
	\frac{\partial r}{\partial {\theta_{\star}}}&=&-\,\frac{\Delta_rL P }{\Upsilon }=
	\frac{\Delta _rG P^2 Q}{\Upsilon }=G \frac{a^2   \Delta _r \left(\sin ^2{\theta_{\star}}\Delta _{\theta} -\sin ^2{\theta}\right) \sqrt{(r^2+a^2)^2-a^2  \sin ^2{\theta_{\star}}\Delta _r}}{(r^2+a^2)^2 \Delta _{\theta}-a^2\sin ^2{\theta} \Delta _r },\\
	&=&\frac{a\Delta _r  \sqrt{(r^2+a^2)^2-a^2  \sin ^2{\theta_{\star}}\Delta _r}}{(r^2+a^2)^2 \Delta _{\theta}-a^2\sin ^2{\theta} \Delta _r }
	\Biggl(
	-\,\frac{\sin(2\theta)}{\sin(2\theta_{\star})}
	\frac{\sqrt{\sin^2\theta_{\star}\Delta_{\theta}-\sin^2\theta}}{\sin(2\theta_{\star})}\Biggr.
	\nonumber\\
	&&
	-2\frac{\left(\sin ^2{\theta_{\star}}\Delta _{\theta} -\sin ^2{\theta}\right)}{\sin^2(2\theta_{\star})}
	\int_{\theta}^{\theta_{\star}}
	\frac{(\sin^2\theta_{\star}-\sin^2\theta')}
	{\sqrt{\sin^2\theta_{\star}\Delta_{\theta'}-\sin^2\theta'}}\md\theta'
	\nonumber\\
	&&\Biggl.+\frac{\left(\sin ^2{\theta_{\star}}\Delta _{\theta} -\sin ^2{\theta}\right)}
	{\sin^2(2\theta_{\star})}
	\int_r^{r_+} \frac{a^3\sin^2(2\theta_{\star}) \Delta_{r'}}{2\left((r'^2+a^2)^2-a^2\sin^2\theta_{\star}\Delta_{r'}\right)^{\frac{3}{2}}}\md r'\Biggr)\sin(2\theta_{\star}),\label{r_t*}
	\eea
	\bea
	\label{t_r*}
	\frac{\partial {\theta}}{\partial r_{\star}}&=&
	\frac{\Delta_r \Delta_{\theta}P}{\Upsilon }=
	\frac{\sqrt{ \sin ^2{\theta_{\star}}\Delta _{\theta}-\sin ^2{\theta}}}
	{\sin(2\theta_{\star})}
	\frac{a \Delta _r \Delta _{\theta}}
	{\left((r^2+a^2)^2 \Delta _{\theta}-a^2  \sin ^2{\theta}\Delta _r\right)}\sin(2\theta_{\star}),
	\eea
	\bea
	\frac{\partial {\theta}}{\partial {\theta_{\star}}}&=&\frac{\Delta_{\theta}L Q }{\Upsilon }=-\,\frac{\Delta _{\theta}G P Q^2 }{\Upsilon }=-G\frac{a \Delta _{\theta} \sqrt{ \sin ^2{\theta_{\star}}\Delta _{\theta}-\sin ^2{\theta}} \left((r^2+a^2)^2-a^2  \sin ^2{\theta_{\star}}\Delta_r \right)}{(r^2+a^2)^2 \Delta _{\theta}-a^2  \sin ^2{\theta}\Delta _r}\nonumber\\
	&=&
	\Delta_{\theta}\frac{(r^2+a^2)^2-a^2\sin^2\theta_{\star}\Delta_r}{ (r^2+a^2)^2\Delta_{\theta}-a^2\sin^2\theta\Delta_{r}}\Biggl({\frac{\sin(2\theta)}{\sin(2\theta_{\star})}}
	\Biggr.\nonumber\\
	&&\ \ \ \ +\,{2\frac{\sqrt{\sin^2\theta_{\star}\Delta_{\theta}-\sin^2\theta}}{\sin(2\theta_{\star})}}
	\int_{\theta}^{\theta_{\star}}
	\frac{(\sin^2\theta_{\star}-\sin^2\theta')}
	{\sqrt{\sin^2\theta_{\star}\Delta_{\theta'}-\sin^2\theta'}}\md\theta'
	\nonumber\\
	&&\Biggl.\ \ \ \ -\,\frac{\sqrt{\sin^2\theta_{\star}\Delta_{\theta}-\sin^2\theta}}{\sin(2\theta_{\star})}\int_r^{r_+} \frac{a^3\sin^2(2\theta_{\star}) \Delta_{r'}}{2\left((r'^2+a^2)^2-a^2\sin^2\theta_{\star}\Delta_{r'}\right)^{\frac{3}{2}}}\md r'\Biggr).\label{t_t*}
	\eea
In the region $r_-<r<r_+$, we have 
\begin{equation}
\begin{array}{ll}
\displaystyle
\left|\frac{\partial r}{\partial r_{\star}}\right|\lesssim |\Delta _r|, 
&\displaystyle
 \left|\frac{\partial r}{\partial {\theta_{\star}}}\right|\lesssim |\Delta _r|\sin (2 {\theta_{\star}})\\
\vspace{-2mm}
\\ 
 \displaystyle
\left|\frac{\partial {\theta}}{\partial r_{\star}}\right|\lesssim |\Delta _r|\sin (2 {\theta_{\star}}),
&\displaystyle \left|\frac{\partial {\theta}}{\partial {\theta_{\star}}}\right|\lesssim 1.
\end{array}
\label{deriv_bounds}
\end{equation}
\end{lem}
\begin{proof}
	We start from \eqref{differentials} and~\eqref{differentials2}.
	We use~\eqref{derivada} to obtain~\eqref{r_t*} and~\eqref{t_t*}.
	The estimates for the derivatives of $r$ and $\theta$ follow
	from the estimates
$$
\sqrt{\sin^2\theta_{\star}\Delta_{\theta }-\sin^2\theta}\leq
\sin\theta_{\star}\sqrt{\Delta_{\theta }}\lesssim\sin\theta_{\star}
\lesssim\sin(2\theta_{\star}),\ \mbox{for}\ \theta_{\star}\ \mbox{close to}\ 0,
$$
$$
\sqrt{\sin^2\theta_{\star}\Delta_{\theta }-\sin^2\theta}=
\cos\theta_{\star}\sqrt{\Xi\frac{\cos^2\theta}{\cos^2\theta_{\star}}-\Delta_\theta}
\stackrel{\eqref{bound}}{\lesssim}\cos\theta_{\star}\lesssim\sin(2\theta_{\star}),
\ \mbox{for}\ \theta_{\star}\ \mbox{close to}\ \frac{\pi}{2},
$$
which together imply that
$$
\frac{\sqrt{\sin^2\theta_{\star}\Delta_{\theta }-\sin^2\theta}}{\sin(2\theta_{\star})}
\lesssim 1.
$$
\end{proof}

\subsubsection{Higher order derivatives}

\begin{lem}
	\label{A.6}
	The functions $r$ and $\frac{\partial \theta}{\partial \theta_{\star}}$ are $C^\infty$. For every 
$k\geq 2$, we have that
\bea
\sum_{1\leq k_1+k_2\leq k}\left|\left(\partial_{r_{\star}}\right)^{k_1}\left(\frac{1}{\sin(2\theta_{\star})}\partial_{\theta_{\star}}\right)^{k_2}r\right|&\lesssim& |\Delta_r|,\label{AA}\\
\sum_{1\leq k_1\leq k}\left|\left(\partial_{r_{\star}}\right)^{k_1}\theta\right|&\lesssim& |\Delta_r| \sin(2\theta_{\star}),\label{BB}\\
\sum_{0\leq k_1\leq k-1}\left|\left(\frac{1}{\sin(2\theta_{\star})}\partial_{\theta_{\star}}\right)^{k_1}\left(\frac{\partial \theta}{\partial \theta_{\star}}\right)\right|&\lesssim& 1,\label{C}\\
\sum_{\stackrel{{1\leq k_1+k_2\leq k-1}}{
		{\mbox{\tiny $k_1 \geq 1$}}}}
	\left|
	\left(\partial_{r_{\star}}\right)^{k_1}
	\left(\frac{1}{\sin(2\theta_{\star})}\partial_{\theta_{\star}}\right)^{k_2}
	\left(\frac{\partial \theta}{\partial \theta_{\star}}\right)\right|&\lesssim& |\Delta_r|.\label{D}
\eea
Moreover, the derivatives of the function $L$ given in~\eqref{L} are bounded as
follows:
\bea
\label{bound1:L}
\sum_{0\leq k_1\leq k-1}\left|\left(\frac{1}{\sin(2\theta_{\star})}\partial_{\theta_{\star}}\right)^{k_1}L\right|&\lesssim& 1,\\
\label{bound2:L}
\sum_{\stackrel{{1\leq k_1+k_2\leq k-1}}{
		{\mbox{\tiny $k_1 \geq 1$}}}}
\left|
\left(\partial_{r_{\star}}\right)^{k_1}
\left(\frac{1}{\sin(2\theta_{\star})}\partial_{\theta_{\star}}\right)^{k_2}
L\right|&\lesssim& |\Delta_r|.
\eea
\end{lem}
\begin{rmk}\label{sense}
	We say that a function depending on $r_\star$ and
	$\theta_\star$ 
	is $C^\infty$ if it has derivatives of all
	orders with respect of $\partial_{r_\star}$ and
	$\frac{1}{\sin(2\theta_\star)}\partial_{\theta_\star}$.
	Refer to\/ {\rm Remarks~\eqref{polo1}, \eqref{polo3}}
	and\/~{\rm \eqref{polo2}} for an explanation about
	the reason for introducing the factor $\frac{1}{\sin(2\theta_\star)}$
	behind $\partial_{\theta_\star}$. The function $\sin^2\theta$ is $C^\infty$.
\end{rmk}
To prove Lemma~\ref{A.6} we need to know the derivatives of
\bea
\nonumber 
R_1&=&\frac{1}{\sin ^{2 n}(2 {\theta_{\star}})}\int_{\theta}^{\theta_{\star}} \frac{\left(\sin^2\theta_{\star}-\sin^2\theta'\right)^n}{\sqrt{\sin^2\theta_{\star} \Delta_{\theta'}-\sin^2\theta'}} \md {\theta'},\\
\nonumber 
R_2&=&\frac{\sqrt{\sin^2\theta_{\star} \Delta_{\theta}-\sin^2\theta}}{\sin (2 {\theta_{\star}})},\\
\nonumber 
R_3&=&\frac{\sin^2\theta_{\star}-\sin^2\theta}{\sin^2 (2 {\theta_{\star}})},\\
\nonumber 
R_4&=&\frac{\sin (2 {\theta})}{\sin (2 {\theta_{\star}})}.
\eea
Note that the dependence of 
$\frac{\partial r}{\partial r_{\star}}$,
$\frac{1}{\sin(2\theta_{\star})}\frac{\partial r}{\partial\theta_{\star}}$,
$\frac{\partial\theta}{\partial\theta_{\star}}$ and
$L$ on $\theta$ and $\theta_{\star}$ is done through $\sin^2\theta$, 
$\sin^2\theta_{\star}$,
$\Delta_\theta$ (which is a function of $\sin^2\theta$, and this is crucial 
for our argument to work),
$\sin^2(2\theta_{\star})$, and $R_1$ to $R_4$. 
The same is true for $\frac{\partial\theta}{\partial r_{\star}}$, which however
also has a factor
$\sin(2\theta_{\star})$.
Moreover, note that
\be
\begin{array}{rcl}
	\partial_{r_{\star}}(\sin^2\theta)&=&\displaystyle\sin(2\theta)\frac{\partial\theta}{\partial r_{\star}},\quad \mbox{this has a factor}\ \sin(2\theta)\sin(2\theta_{\star}),\\
	\displaystyle
	\partial_{r_{\star}}(\sin(2\theta)\sin(2\theta_{\star}))&=&\displaystyle\cos(2\theta)\frac{\partial\theta}{\partial r_{\star}}\sin(2\theta_{\star}),\quad \mbox{this has factors}\ \cos(2\theta)\ \mbox{and}\ \sin^2(2\theta_{\star}),\\
	\displaystyle
	\frac{1}{\sin(2\theta_{\star})}\partial_{\theta_{\star}}(\sin(2\theta)\sin(2\theta_{\star}))&=&\displaystyle2\cos(2\theta)\frac{\partial\theta}{\partial\theta_{\star}}+2\frac{\sin(2\theta)}{\sin(2\theta_{\star})}\cos(2\theta_{\star})\\
	\displaystyle
	\frac{1}{\sin(2\theta_{\star})}\partial_{\theta_{\star}}(\sin^2\theta)&=&\displaystyle
	\frac{\sin(2\theta)}{\sin(2\theta_{\star})}\frac{\partial\theta}{\partial\theta_{\star}},\\
	\displaystyle
		\frac{1}{\sin(2\theta_{\star})}\partial_{\theta_{\star}}(\sin^2\theta_{\star})
	&=&1,\\
	\displaystyle
	\frac{1}{\sin(2\theta_{\star})}\partial_{\theta_{\star}}(\sin^2(2\theta_{\star}))&=&
	4\cos(2\theta_{\star}),\\
	\displaystyle
	\frac{1}{\sin(2\theta_{\star})}\partial_{\theta_{\star}}\cos(2\theta)&=&\displaystyle-2\frac{\sin(2\theta)}{\sin(2\theta_{\star})}\frac{\partial\theta}{\partial\theta_{\star}},\\
	\displaystyle
	\frac{1}{\sin(2\theta_{\star})}\partial_{\theta_{\star}}\cos(2\theta_{\star})
	&=&-2,
\end{array}
\label{good1}
\end{equation}
or, alternatively, $\cos(2\theta)=1-2\sin^2\theta$.
If the factor $\sin(2\theta)$ were to appear by itself in our formulas,
or if the factor $\sin(2\theta_{\star})$ were to appear by itself,
then our argument would not go through because the derivatives 
$\frac{1}{\sin(2\theta_{\star})}\partial_{\theta_{\star}}$ of each of these
are not bounded. The structure of our problem is such that when one of these factors
is present, then the other one is also present, and their product has derivatives of all orders with
respect to
$\partial_{r_{\star}}$ and $\frac{1}{\sin(2\theta_{\star})}\partial_{\theta_{\star}}$ which are bounded.
Furthermore, as we will see below,
 the derivatives of $R_1$ to $R_4$ are sums whose summands
are products of factors that are either $R_1$, $R_2$, $R_3$, $R_4$, or, when
this is not the case, others that are clearly smooth (since the denominators that
will appear do not  vanish for $r_-\leq r\leq r_+$). Hence, it turns out 
that to prove that $r$ and $\theta$ are smooth, we just have to 
check that $R_1$ to $R_4$ are $C^1$ and that their first derivatives have the
aforementioned property. Next we calculate the first derivatives of
$R_1$ to $R_4$.

As a small  note, let $\theta_{\star}\in
\bigl(0,\frac{\pi}{2}\bigr)$. Consider the function
$$
\frac{1}{\sin (2 {\theta_{\star}})}\int_\theta^{\theta_{\star}}\md\theta'
=
\frac{{\theta_{\star}}-{\theta(r_{\star},\theta_{\star})}}{\sin (2 {\theta_{\star}})}.
$$
This is clearly bounded for $\theta_{\star}$ close to zero.
It is also bouded for ${\theta_{\star}}$ close to $\frac{\pi}{2}$.
Indeed, in this situation, we have
$$
\frac{{\theta_{\star}}-{\theta}}{\sin (2 {\theta_{\star}})}\leq \frac{\frac{\pi }{2}-{\theta}}{\sin (2 {\theta_{\star}})}=\frac{\frac{\pi }{2}-{\theta}}{2 \sin ({\theta_{\star}}) \cos ({\theta_{\star}})}\leq \frac{C \left(\frac{\pi }{2}-{\theta}\right)}{\cos ({\theta})}=\frac{C \left(\frac{\pi }{2}-{\theta}\right)}{\sin \left(\frac{\pi }{2}-{\theta}\right)}\leq C,
$$	
because $\frac{\cos\theta}{\cos\theta_{\star}}$ is bounded (\eqref{bound}). So the function
is bounded for $\theta_{\star}\in\bigl(0,\frac{\pi}{2}\bigr)$.

\begin{rmk}
	The reader will notice, using the expressions below, that 
	\bea
	|\partial_{r_{\star}}R_i|&\lesssim&|\Delta_r|\label{dr},\\
	\left|\frac{1}{\sin (2 {\theta_{\star}})}
	\partial_{\theta_{\star}}R_i\right|&\lesssim&1,\label{dt}
	\eea
	for each $i\in\{1,2,3,4\}$.
\end{rmk}

\vspace{\baselineskip}

\noindent{\bf Derivatives of $R_1$.} For any integer $n\geq 1$, the following identities hold:
	\bea
	&&\partial_{r_{\star}}\left( \frac{1}{\sin ^{2 n}(2 {\theta_{\star}})}\int_{\theta}^{\theta_{\star}} \frac{\left(\sin^2\theta_{\star}-\sin^2\theta'\right)^n}{\sqrt{\sin^2\theta_{\star} \Delta_{\theta'}-\sin^2\theta'}} \md {\theta'}\right)\nonumber\\
	&&\qquad=-\,\frac{\left(\sin^2\theta_{\star}-\sin^2\theta\right)^n}{\sin ^{2 n}(2 {\theta_{\star}})}\frac{a \Delta _r \Delta _{\theta}}{\left((r^2+a^2)^2 \Delta _{\theta}-a^2 \sin^2\theta \Delta_r\right)},\label{A.8}
	\eea
\bea
&&\left(\frac{1}{\sin (2 {\theta_{\star}})}\partial_{\theta_{\star}}\right)\left(\frac{1}{\sin ^{2 n}(2 {\theta_{\star}})}\int_{\theta}^{\theta_{\star}} \frac{\left(\sin^2\theta_{\star}-\sin^2\theta'\right)^n}{\sqrt{\sin^2\theta_{\star} \Delta_{\theta'}-\sin^2\theta'}} \md {\theta'}\right)\nonumber\\
&&\ \ \ \ =\frac{2 (2 n+1)}{\sin ^{2 (n+1)}(2 {\theta_{\star}})} \int_{\theta}^{\theta_{\star}} \frac{\left(\sin^2\theta_{\star}-\sin^2\theta'\right)^{n+1}}{\sqrt{\sin^2\theta_{\star} \Delta_{\theta'}-\sin^2\theta'}} \md {\theta'}\nonumber\\
&&\ \ \ \qquad-\Biggl(\frac{2}{\sin^2 (2 {\theta_{\star}})}  \int_{\theta}^{\theta_{\star}} \frac{\sin^2\theta_{\star}-\sin^2\theta'}{\sqrt{\sin^2\theta_{\star} \Delta_{\theta'}-\sin^2\theta'}} \md {\theta'}\Biggr.\nonumber\\
&&\qquad\ \  \qquad \Biggl.
- \int_r^{r_+} \frac{a^3 \Delta _{r'}}{2 {\left(\left((r')^2+a^2\right)^2-a^2 \sin^2\theta_{\star} \Delta _{r'}\right)^{3/2}}} \md r'\Biggr)
\nonumber\\
&&\qquad\qquad\qquad\qquad \times
\Delta _{\theta}\frac{(r^2+a^2)^2-a^2 \sin^2\theta_{\star} \Delta_r}{(r^2+a^2)^2 \Delta _{\theta}-a^2 \sin^2\theta \Delta_r}
\frac{ \left(\sin^2\theta_{\star}-\sin^2\theta\right)^n } {\sin ^{2n}(2 {\theta_{\star}}) }\nonumber
\\
&&\ \ \ \qquad + \frac{a^2 \Delta _r}{(r^2+a^2)^2 \Delta _{\theta}-a^2 \sin^2\theta \Delta_r} 
\frac{\sqrt{\sin^2\theta_{\star} \Delta_{\theta}-\sin^2\theta}}{\sin (2 {\theta_{\star}})}
\frac{\sin (2 {\theta})}{ \sin (2 {\theta_{\star}})} 
\frac{ \left(\sin^2\theta_{\star}-\sin^2\theta\right)^n}{\sin ^{2 n}(2 {\theta_{\star}})}.\label{ddtheta}
\eea

\begin{proof}
The proof of~\eqref{A.8} is immediate using~\eqref{t_r*}, as
	\bea
&&\partial_{r_{\star}}\left( \frac{1}{\sin ^{2 n}(2 {\theta_{\star}})}\int_{\theta}^{\theta_{\star}} \frac{\left(\sin^2\theta_{\star}-\sin^2\theta'\right)^n}{\sqrt{\sin^2\theta_{\star} \Delta_{\theta'}-\sin^2\theta'}} \md {\theta'}\right)\nonumber\\
&&\qquad=-\,\frac{1}{\sin ^{2 n}(2 {\theta_{\star}})}\frac{\left(\sin^2\theta_{\star}-\sin^2\theta\right)^n}{\sqrt{\sin^2\theta_{\star} \Delta_{\theta}-\sin^2\theta}}\frac{\partial {\theta}}{\partial r_{\star}}.
\nonumber
\eea
By differentiation we obtain
\bea
&&\partial_{\theta_{\star}}\left(\frac{1}{\sin ^{2 n}(2 {\theta_{\star}})}\int_{\theta}^{\theta_{\star}} \frac{\left(\sin^2\theta_{\star}-\sin^2\theta'\right)^n}{\sqrt{\sin^2\theta_{\star} \Delta_{\theta'}-\sin^2\theta'}} \md {\theta'}\right)\nonumber\\
&&\qquad=-\,\frac{4 n \cos (2 {\theta_{\star}})}{\sin ^{2 n+1}(2 {\theta_{\star}})} \int_{\theta}^{\theta_{\star}} \frac{\left(\sin^2\theta_{\star}-\sin^2\theta'\right)^n}{\sqrt{D(\theta',\theta_{\star})}} \md {\theta'}\nonumber\\
&&\qquad\ \ \ +\frac{n}{\sin ^{2 n-1}(2 {\theta_{\star}})} \int_{\theta}^{\theta_{\star}} \frac{\left(\sin^2\theta_{\star}-\sin^2\theta'\right)^{n-1}}{\sqrt{D(\theta',\theta_{\star})}} \md {\theta'}
\nonumber\\
&&\qquad\ \ \ 
-\,\frac{1}{2 \sin ^{2 n-1}(2 {\theta_{\star}})}\int_{\theta}^{\theta_{\star}} \frac{\Delta_{\theta'} \left(\sin^2\theta_{\star}-\sin^2\theta'\right)^n}{\sqrt{\left(D(\theta',\theta_{\star})\right)^3}} \md {\theta'}\label{here}\\
&&\qquad\ \ \
-\,\frac{\partial {\theta}}{\partial {\theta_{\star}}}\frac{\left(\sin^2\theta_{\star}-\sin^2\theta\right)^n}{\sin ^{2 n}(2 {\theta_{\star}}) \sqrt{D(\theta,\theta_{\star})}}\nonumber.
\eea
We keep the two first integrals unchanged and use~\eqref{trick}
on the integral marked \eqref{here} to obtain
\begin{eqnarray*}
&&-\,\frac{4 n \cos (2 \theta_{\star})}{\sin ^{2 n+1}(2 \theta_{\star})}\int_\theta^{\theta_{\star}} \frac{\left(\sin ^2\theta_{\star}-\sin ^2\theta'\right)^n}{\sqrt{D(\theta',\theta_{\star})}} \md \theta'\\
&&+\frac{n}{\sin ^{2 n-1}(2 \theta_{\star})} \int_\theta^{\theta_{\star}} \frac{\left(\sin ^2\theta_{\star}-\sin ^2\theta'\right)^{n-1}}{\sqrt{D(\theta',\theta_{\star})}} \md \theta'\\
&&-\,\frac{1}{\sin ^{2 n+1}(2 \theta_{\star})}\int_\theta^{\theta_{\star}} \frac{(\cos (2 \theta')+\cos (2 \theta_{\star})) \left(\sin ^2\theta_{\star}-\sin ^2\theta'\right)^n}{\sqrt{D(\theta',\theta_{\star})}}\md \theta'\\
&&-\,\frac{1}{\sin ^{2 n+1}(2 \theta_{\star})}\int_\theta^{\theta_{\star}}\left(\sin ^2\theta_{\star}-\sin ^2\theta'\right)^n\sin (2 \theta')\,\partial_{\theta'}
\left(\frac{1}{\sqrt{D(\theta',\theta_{\star})}}\right)\md \theta'\\
&&
-\,\frac{\partial {\theta}}{\partial {\theta_{\star}}}\frac{\left(\sin^2\theta_{\star}-\sin^2\theta\right)^n}{\sin ^{2 n}(2 {\theta_{\star}}) \sqrt{D(\theta,\theta_{\star})}}\nonumber.
\end{eqnarray*}
Integrating the last integral by parts and using the fact that
\begin{eqnarray*}
&&-4 n\cos(2 \theta_{\star})\left(\sin^2\theta_{\star}-\sin^2\theta\right)
+n\sin^2(2\theta_{\star})
 - (\cos(2 \theta')+\cos(2 \theta_{\star})) \left(\sin^2\theta_{\star}-\sin^2\theta\right)\\
 &&
 \ \ \ +2 \cos(2 \theta')\left(\sin^2\theta_{\star}-\sin^2\theta\right)
 -n\sin^2(2\theta') \ =\ 
2 (2 n + 1) (\sin^2\theta_{\star}- \sin^2\theta')^2,
\end{eqnarray*}
we obtain
\bea
&&\partial_{\theta_{\star}}\left(\frac{1}{\sin ^{2 n}(2 {\theta_{\star}})}\int_{\theta}^{\theta_{\star}} \frac{\left(\sin^2\theta_{\star}-\sin^2\theta'\right)^n}{\sqrt{D(\theta',\theta_{\star})}} \md {\theta'}\right)\nonumber\\
&&\qquad=
\frac{2 (2 n+1) }{\sin ^{2 n+1}(2 {\theta_{\star}})}\int_{\theta'}^{\theta_{\star}} \frac{\left(\sin^2\theta_{\star}-\sin^2\theta'\right)^{n+1}}{\sqrt{D(\theta',\theta_{\star})}} \md {\theta'}\nonumber\\
&&\qquad\ \ \ +\left(\frac{\sin (2 {\theta})}{\sin (2 {\theta_{\star}})}-\,\frac{\partial {\theta}}{\partial {\theta_{\star}}}\right)\frac{\left(\sin^2\theta_{\star}-\sin^2\theta\right)^n}{\sin ^{2 n}(2 {\theta_{\star}}) \sqrt{D(\theta,\theta_{\star})}}.\label{A.33}
\eea
From equality \eqref{t_t*}, it follows that
\bea
&&\frac{\partial {\theta}}{\partial {\theta_{\star}}}-\,\frac{\sin (2 {\theta})}{\sin (2 {\theta_{\star}})}\nonumber\\
&&\ \ =\Delta _{\theta}\frac{(r^2+a^2)^2-a^2 \sin^2\theta_{\star} \Delta_r}{(r^2+a^2)^2 \Delta _{\theta}-a^2 \sin^2\theta' \Delta_r} 
\frac{\sqrt{\sin^2\theta_{\star}\Delta_{\theta}-\sin^2\theta}}{\sin (2 {\theta_{\star}})}\nonumber\\
&&\ \ \ \ \times \Biggl(2  \int_{\theta}^{\theta_{\star}} \frac{\sin^2\theta_{\star}-\sin^2\theta'}{\sqrt{\sin^2\theta_{\star} \Delta_{\theta'}-\sin^2\theta'}} \md {\theta'}- \int_r^{r_+} \frac{a^3 \sin^2 (2 {\theta_{\star}}) \Delta _{r'}}{2 \left(\left((r')^2+a^2\right)^2-a^2 \sin^2\theta_{\star} \Delta _{r'}\right)^{3/2}} \md r'\Biggr)\nonumber\\
&&\ \ \ \ 
-a^2 \Delta _r \frac{\sin^2\theta_{\star}\Delta_{\theta} -\sin^2\theta}{ (r^2+a^2)^2 \Delta _{\theta}-a^2 \sin^2\theta \Delta_r}\frac{ \sin (2 {\theta})}{\sin (2 {\theta_{\star}})}.\label{menos-um}
\eea
Using the last equality in \eqref{A.33}
and dividing by $\sin(2\theta_{\star})$, we obtain~\eqref{ddtheta}.
\end{proof}

\vspace{\baselineskip}

\noindent{\bf Derivatives of $R_2$.} The following identities hold:
	\bea
	&&\partial_{r_{\star}} \frac{\sqrt{\sin^2\theta_{\star} \Delta_{\theta}-\sin^2\theta}}{\sin (2 {\theta_{\star}})}
	=
	-\,\frac{a}{2}\frac{\sin (2 {\theta})}{\sin (2 {\theta_{\star}})}\frac{\Delta _r \Delta _{\theta} \left(1+\frac{\Lambda}{3} a^2  \sin^2\theta_{\star}\right)}{\left((r^2+a^2)^2 \Delta _{\theta}-a^2 \sin^2\theta \Delta_r\right)},\label{d_sqrt_r}
	\eea
\bea
&&\frac{1}{\sin (2 {\theta_{\star}})}\partial_{\theta_{\star}}\frac{\sqrt{\sin^2\theta_{\star} \Delta_{\theta}-\sin^2\theta}}{\sin (2 {\theta_{\star}})}\nonumber\\
&&\qquad=2\frac{\left(\sin^2\theta_{\star}-\sin^2\theta\right)}{\sin^2(2 {\theta_{\star}})}\frac{\sqrt{\sin^2\theta_{\star} \Delta_{\theta}-\sin^2\theta}}{\sin(2 {\theta_{\star}})}\nonumber\\
&&\qquad\ \ \ +\frac{1}{2}
\frac{\sin ^2(2 {\theta})}{\sin^2(2 {\theta_{\star}})}
\frac{ \sqrt{\sin^2\theta_{\star} \Delta_{\theta}-\sin^2\theta}}
{\sin(2 {\theta_{\star}})}\frac{a^2 \Delta _r\left(1+\frac{\Lambda}{3} a^2  \sin^2\theta_{\star}\right)}{(r^2+a^2)^2 \Delta _{\theta}-a^2 \sin^2\theta \Delta_r} \nonumber\\
&&\qquad\ \ \ -\,\frac{1}{2}\frac{\sin (2 {\theta})}{\sin(2 {\theta_{\star}}) }\Delta _{\theta}\frac{((r^2+a^2)^2 -a^2 \Delta _r  \sin^2\theta_{\star})\left(1+\frac{\Lambda}{3} a^2  \sin^2\theta_{\star}\right)}{(r^2+a^2)^2 \Delta _{\theta}-a^2 \sin^2\theta \Delta_r}\nonumber\\
&&\qquad\qquad\ \ \  \times\Biggl(\frac{2}
{\sin^2 (2 {\theta_{\star}})}\int_\theta^{\theta_{\star}} \frac{\sin^2\theta_{\star}-\sin^2\theta'}{\sqrt{\sin^2\theta_{\star} \Delta_{\theta'}-\sin^2\theta'}} \md {\theta'}
\Biggr.\nonumber\\ &&\qquad\qquad\ \ \ \ \ \ \ \ \Biggl.
- \int_r^{r_+} \frac{a^3 \Delta _{r'}}{2 \left(\left((r')^2+a^2\right)^2-a^2 \sin^2\theta_{\star} \Delta _{r'}\right)^{3/2}}\md r'\Biggr).\label{d_sqrt}
\eea
\begin{proof}
	From~\eqref{t_r*}, we get~\eqref{d_sqrt_r}.
	To start the proof of~\eqref{d_sqrt}, note that
\begin{eqnarray*}
	\frac{1}{\sin (2 {\theta_{\star}})}\partial_{\theta_{\star}}\frac{\sqrt{\sin^2\theta_{\star} \Delta_{\theta}-\sin^2\theta}}{\sin (2 {\theta_{\star}})}
&=&-\,\frac{2 \cos(2\theta_{\star}) \sqrt{D(\theta,\theta_{\star})}}{\sin^3(2\theta_{\star})}\\ 
&&+\frac{1}{2} \frac{\Delta_{\theta}}{\sin(2\theta_{\star}) \sqrt{D(\theta,\theta_{\star})}}\\
&&-\,\frac{1}{2}\frac{1}{\sin^2(2\theta_{\star})}\frac{\sin(2\theta) \left(1+\frac{\Lambda}{3} a^2 \sin^2\theta_{\star}\right)}{\sqrt{D(\theta,\theta_{\star})}}\frac{\partial\theta}{\partial\theta_{\star}}.
\end{eqnarray*}
For any $A$, the last expression is equal to
\begin{eqnarray*}
&&=\frac{1}{2 \sin^3(2\theta_{\star}) \sqrt{D(\theta,\theta_{\star})}}\Biggl(-4 \cos(2\theta_{\star}) D(\theta,\theta_{\star})\Biggr.\\ 
&&\ \ \ +\sin^2(2\theta_{\star}) \Delta_{\theta}-\sin^2(2\theta) \left(1+\frac{\Lambda}{3} a^2 \sin^2\theta_{\star}\right)\\ 
&&\ \ \ -(A-1) \sin^2(2\theta) \left(1+\frac{\Lambda}{3} a^2 \sin^2\theta_{\star}\right)\\ 
&&\ \ \ \Biggl.
+\sin(2\theta)\left(1+\frac{\Lambda}{3} a^2 						
						\sin^2\theta_{\star}\right)\sin(2\theta_{\star})
\left(
	A \frac{\sin(2\theta)}{\sin(2\theta_{\star})}	-\,\frac{\partial\theta}{\partial\theta_{\star}}
\right)\Biggr),
\end{eqnarray*}
because the terms with $A$ cancel out and the terms with $\sin^2(2\theta)$
cancel out. The value of $A$ will be chosen taking~\eqref{t_t*}
into account, that is
we choose $A=\Delta _{\theta}\frac{(r^2+a^2)^2 -a^2 \sin^2\theta_{\star} \Delta_r}{(r^2+a^2)^2 \Delta _{\theta}-a^2 \sin^2\theta \Delta_r}$.
The reason we made a term with $A-1$ appear is that such a term
is proportional to $\Delta_r$. In fact,
$$
A-1=\Delta _{\theta}\frac{(r^2+a^2)^2 -a^2 \sin^2\theta_{\star} \Delta_r}{(r^2+a^2)^2 \Delta _{\theta}-a^2 \sin^2\theta \Delta_r}-1=-\,\frac{a^2 \Delta _r 
	\left(\sin^2\theta_{\star}\Delta_{\theta}-\sin^2\theta\right)}
{(r^2+a^2)^2 \Delta _{\theta}-a^2 \sin^2\theta \Delta_r}.
$$
According to~\eqref{trick}, we have that
\begin{eqnarray*}
&&-4 \cos(2\theta_{\star}) D(\theta,\theta_{\star})
+\sin^2(2\theta_{\star}) \Delta_{\theta}-\sin^2(2\theta) \left(1+\frac{\Lambda}{3} a^2 \sin^2\theta_{\star}\right)\\
&&\qquad =-4 \cos(2\theta_{\star}) D(\theta,\theta_{\star})
+\sin^2(2\theta_{\star}) \Delta_{\theta}+\sin(2\theta) \partial_\theta 
D(\theta,\theta_{\star})\\
&&\qquad =-4 \cos(2\theta_{\star}) D(\theta,\theta_{\star})
+2(\cos(2\theta)+\cos(2\theta_{\star}))D(\theta,\theta_{\star})\\
&&\qquad =
2(\cos(2\theta)-\cos(2\theta_{\star}))D(\theta,\theta_{\star})\\
&&\qquad =4 \left(\sin^2\theta_{\star}-\sin^2\theta\right) D(\theta,\theta_{\star}).
\end{eqnarray*}
Hence, the expression above is equal to
\begin{eqnarray*}
&&\frac{1}{2 \sin^3(2\theta_{\star}) \sqrt{D(\theta,\theta_{\star})}}\Biggl(4 \left(\sin^2\theta_{\star}-\sin^2\theta\right) D(\theta,\theta_{\star})\Biggr.\\ 
&&\ \ \ -(A-1) \sin^2(2\theta) \left(1+\frac{\Lambda}{3} a^2 \sin^2\theta_{\star}\right)\\ 
&&\ \ \ -\sin(2\theta)\left(1+\frac{\Lambda}{3} a^2 \sin^2\theta_{\star}\right)\Biggl.\sin(2\theta_{\star})\left(\frac{\partial\theta}{\partial\theta_{\star}}-A\frac{\sin(2\theta)}{\sin(2\theta_{\star})}\right)\Biggr)\\
&&=2\frac{\left(\sin^2\theta_{\star}-\sin^2\theta\right)}{\sin^2(2\theta_{\star})}\frac{\sqrt{D(\theta,\theta_{\star})}}{\sin(2\theta_{\star})}\\ 
&&\ \ \ +\frac{1}{2}\frac{\sin^2(2\theta)}{\sin^2(2\theta_{\star})}\frac{\sqrt{D(\theta,\theta_{\star})}}{\sin(2\theta_{\star})}
\frac{a^2 \Delta_r\left(1+\frac{\Lambda}{3} a^2 \sin^2\theta_{\star}\right)}{\left((r^2+a^2)^2 \Delta_\theta-a^2  \sin^2\theta\Delta_r\right)}\\ 
&&\ \ \ -\,\frac{1}{2}\frac{\sin(2\theta)}{\sin^2(2\theta_{\star})}\frac{ \left(1+\frac{\Lambda}{3} a^2 \sin^2\theta_{\star}\right)}{ \sqrt{D(\theta,\theta_{\star})}}\left(\frac{\partial\theta}{\partial\theta_{\star}}-A\frac{\sin(2\theta)}{\sin(2\theta_{\star})}\right).
\end{eqnarray*}
The final expression~\eqref{d_sqrt} is obtained using~\eqref{t_t*}.
\end{proof}

\vspace{\baselineskip}

\noindent{\bf Derivatives of $R_3$.} The following identities hold:
	\begin{eqnarray}
		&&	\frac{\partial }{\partial {r_{\star}}}\frac{(\sin^2\theta_{\star} -\sin^2\theta)}{\sin^2(2 {\theta_{\star}})}\nonumber\\	
		&&\qquad =-\,\frac{\sin(2\theta)}{\sin(2 {\theta_{\star}})}
		\frac{\sqrt{ \sin ^2{\theta_{\star}}\Delta _{\theta}-\sin ^2{\theta}}}{\sin(2 {\theta_{\star}})}
		\frac{a \Delta _r \Delta _{\theta}}{\left((r^2+a^2)^2 \Delta _{\theta}-a^2  \sin ^2{\theta}\Delta _r\right)},\label{q_r}
	\end{eqnarray}	
\begin{eqnarray*}
	&&	\frac{1}{\sin (2 {\theta_{\star}})}\partial_{\theta_{\star}}\frac{(\sin^2\theta_{\star} -\sin^2\theta)}{\sin^2(2 {\theta_{\star}})}\\	
	&&\qquad\ \ =
	4\frac{(\sin^2\theta_{\star}-\sin^2\theta)^2}{\sin^4(2\theta_{\star})}\\
	&&\qquad\ \ \ \ \ +\frac{\left(\sin^2\theta_{\star}\Delta_{\theta}-\sin^2\theta\right)}{\sin^2(2\theta_{\star})}
	\frac{\sin^2(2\theta)}{\sin^2(2\theta_{\star})}
	\frac{a^2 \Delta _r}
	{(r^2+a^2)^2 \Delta _{\theta}-a^2 \sin^2\theta \Delta_r}\\
	&&\qquad\ \ \ \ \ -\,\frac{\sin(2\theta)}{\sin(2\theta_{\star})}\frac{\sqrt{\sin^2\theta_{\star}\Delta_{\theta}-\sin^2\theta}}{\sin(2\theta_{\star})}
		\Delta_{\theta}\frac{(r^2+a^2)^2-a^2\sin^2\theta_{\star}\Delta_r}{ (r^2+a^2)^2\Delta_{\theta}-a^2\sin^2\theta\Delta_{r}}\\
	&&\qquad\qquad\qquad \times\Biggl(\frac{2}{\sin^2(2\theta_{\star})}
	\int_{\theta}^{\theta_{\star}}
	\frac{(\sin^2\theta_{\star}-\sin^2\theta')}
	{\sqrt{\sin^2\theta_{\star}\Delta_{\theta'}-\sin^2\theta'}}\md\theta'\Biggr.
	\nonumber\\
	&&\qquad\qquad\qquad\ \ \ \ \Biggl.-\int_r^{r_+} \frac{a^3 \Delta_{r'}}{2\left((r'^2+a^2)^2-a^2\sin^2\theta_{\star}\Delta_{r'}\right)^{\frac{3}{2}}}\md r'\Biggr).
\end{eqnarray*}
\begin{proof}
	From~\eqref{t_r*}, we get~\eqref{q_r}.
	Arguing as in the proof of~\eqref{d_sqrt}, we obtain
\begin{eqnarray*}
	\frac{1}{\sin (2 {\theta_{\star}})}\partial_{\theta_{\star}}\frac{(\sin^2\theta_{\star} -\sin^2\theta)}{\sin^2(2 {\theta_{\star}})}
	&=&-\,\frac{4 \cos(2\theta_{\star}) (\sin^2\theta_{\star} -\sin^2\theta)}{\sin^4(2\theta_{\star})}
	+ \frac{1}{\sin^2(2\theta_{\star})}
	-\,\frac{\sin(2\theta)}{\sin^3(2\theta_{\star})}\frac{\partial\theta}{\partial\theta_{\star}}\\
	&=&-\,\frac{4 \cos(2\theta_{\star}) (\sin^2\theta_{\star} -\sin^2\theta)}{\sin^4(2\theta_{\star})}
	+ \frac{1}{\sin^2(2\theta_{\star})}
	-\,\frac{\sin^2(2\theta)}{\sin^4(2\theta_{\star})}\\
	&&-(A-1)\frac{\sin^2(2\theta)}{\sin^4(2\theta_{\star})}
	-\,\frac{\sin(2\theta)}{\sin^3(2\theta_{\star})}
	\Biggl(\frac{\partial\theta}{\partial\theta_{\star}}-A\frac{\sin(2\theta)}{\sin(2\theta_{\star})}\Biggr)\\
		&=&
	4\frac{(\sin^2\theta_{\star}-\sin^2\theta)^2}{\sin^4(2\theta_{\star})}\\
	&&+\frac{\left(\sin^2\theta_{\star}\Delta_{\theta}-\sin^2\theta\right)}{\sin^2(2\theta_{\star})}
	\frac{\sin^2(2\theta)}{\sin^2(2\theta_{\star})}
	\frac{a^2 \Delta _r}
	{(r^2+a^2)^2 \Delta _{\theta}-a^2 \sin^2\theta \Delta_r}\\
	&&-\,\frac{\sin(2\theta)}{\sin^3(2\theta_{\star})}
	\Biggl(\frac{\partial\theta}{\partial\theta_{\star}}-A\frac{\sin(2\theta)}{\sin(2\theta_{\star})}\Biggr).
\end{eqnarray*}
One concludes by once again applying~\eqref{t_t*}.
\end{proof}

\vspace{\baselineskip}

\noindent{\bf Derivatives of $R_4$.} The following identities hold:
	\bea
	\label{r_der_A.10}
	\partial_{r_{\star}} \left(\frac{\sin (2 {\theta})}{\sin (2 {\theta_{\star}})}\right)
	&=&\frac{ \sqrt{\sin^2\theta_{\star} \Delta_{\theta}-\sin^2\theta}}{\sin (2 {\theta_{\star}})}\frac{2 a \Delta _r \Delta _{\theta} \cos (2 {\theta})}
	{\left((r^2+a^2)^2 \Delta _{\theta}-a^2 \sin^2\theta \Delta_r\right)},\label{ss_r}
	\eea
\bea
&&\left(\frac{1}{\sin (2 {\theta_{\star}})}\partial_{\theta_{\star}}\right)\left(\frac{\sin (2 {\theta})}{\sin (2 {\theta_{\star}})}\right)\nonumber\\
&&\ \ \ =
4\frac{\sin (2 {\theta})}{\sin(2 {\theta_{\star}})}
\frac{\left(\sin^2\theta_{\star}-\sin^2\theta\right)}{\sin^2(2 {\theta_{\star}})}
-2\frac{\sin (2 {\theta})}{\sin(2 {\theta_{\star}})}
\frac{\left(\sin^2\theta_{\star}\Delta_{\theta}-\sin^2\theta\right)}{\sin^2(2 {\theta_{\star}})}\frac{ a^2 \Delta_r\cos (2 {\theta})}{(r^2+a^2)^2 \Delta_{\theta}
	-a^2 \sin^2\theta\Delta_r}
\nonumber\\
&&\qquad +2 \frac{\sqrt{\sin^2\theta_{\star}\Delta_{\theta}-\sin^2\theta}}{\sin(2 {\theta_{\star}})}\Delta _{\theta}\frac{(r^2+a^2)^2-a^2  \sin^2\theta_{\star}\Delta_r}{(r^2+a^2)^2 \Delta _{\theta}-a^2 \sin^2\theta \Delta_r}\cos (2 {\theta})
\nonumber\\
&&\qquad  \times
\Biggl(\frac{2}{\sin^2(2 {\theta_{\star}})}\int_{\theta}^{\theta_{\star}} \frac{\sin^2\theta_{\star}-\sin^2\theta'}{\sqrt{\sin^2\theta_{\star} \Delta_{\theta'}-\sin^2\theta'}} \md {\theta'}
- \int_r^{r_+}{ \frac{a^3  \Delta _{r'}}{2 \left(\left((r')^2+a^2\right)^2-a^2 \sin^2\theta_{\star} \Delta _{r'}\right)^{3/2}} \md r'}\Biggr).\label{t_der_A.10}
\eea
\begin{proof}
	From~\eqref{t_r*}, we immediately get~\eqref{ss_r}.
Moreover
\begin{eqnarray*}
\left(\frac{1}{\sin (2 {\theta_{\star}})}\partial_{\theta_{\star}}\right)\left(\frac{\sin (2 {\theta})}{\sin (2 {\theta_{\star}})}\right)
&=&-\,\frac{2 \sin (2 {\theta}) \cos (2 {\theta_{\star}})}{\sin ^3(2 {\theta_{\star}})}+\frac{2 \cos (2 {\theta})}{\sin ^2(2 {\theta_{\star}})}\frac{\partial {\theta}}{\partial {\theta_{\star}}}\\
&=&-\,\frac{\sin(2\theta)}{\sin(2\theta_{\star})}
\Biggl(\frac{2\cos(2\theta_{\star})}{\sin^2(2\theta_{\star})}
-\,\frac{2\cos(2\theta)}{\sin^2(2\theta_{\star})}\Biggr)\\
&&+\frac{2 \cos (2 {\theta})}{\sin ^2(2 {\theta_{\star}})}\Biggl(\frac{\partial {\theta}}{\partial {\theta_{\star}}}-
\,\frac{\sin(2\theta)}{\sin(2\theta_{\star})}
\Biggr).
\end{eqnarray*}
Equality~\eqref{t_der_A.10} is obtained using \eqref{menos-um}.
\end{proof}

\vspace{\baselineskip}

\begin{altproof}{{\,\rm Lemma \ref{A.6}}}\\
\begin{altproof}{\eqref{AA}}
We have seen in~\eqref{deriv_bounds} that
$\left|\frac{\partial r}{\partial{r_{\star}}}\right|\lesssim |\Delta_r|$
(this is \eqref{AA} with $k_1=1$ and $k_2=0$)
and that
$\left|\frac{1}{\sin(2\theta_{\star})}\frac{\partial r}{\partial{\theta_{\star}}}\right|\lesssim |\Delta_r|$
(this is \eqref{AA} with $k_1=0$ and $k_2=1$).
Thus, inequalities~\eqref{AA} follow from 
\begin{enumerate}[(i)]
	\item \eqref{dr} holds, 	
	\item \eqref{dt} holds, 
	\item $\left|\frac{\partial r}{\partial{r_{\star}}}\right|\lesssim |\Delta_r|$,
	which in particular implies that 
	$\left|\frac{\partial \Delta_r}{\partial{r_{\star}}}\right|\lesssim |\Delta_r|$,
	\item $\left|\frac{1}{\sin(2\theta_{\star})}\frac{\partial r}{\partial{\theta_{\star}}}\right|\lesssim |\Delta_r|$, which in particular implies that 
	$\left|\frac{1}{\sin(2\theta_{\star})}\frac{\partial \Delta_r}{\partial{\theta_{\star}}}\right|\lesssim |\Delta_r|$,
	\item $\left|\frac{\partial\theta}{\partial{r_{\star}}}\right|\lesssim |\Delta_r|\sin(2\theta_{\star})\lesssim |\Delta_r|\lesssim 1$,
	\item $\left|\frac{1}{\sin(2\theta_{\star})}\frac{\partial(\sin^2\theta)}{\partial\theta_{\star}}\right|\lesssim 1$,
	$\left|\frac{1}{\sin(2\theta_{\star})}\frac{\partial(\sin(2\theta)\sin(2\theta_{\star}))}{\partial\theta_{\star}}\right|\lesssim 1$,
	$\left|\frac{1}{\sin(2\theta_{\star})}\frac{\partial(\cos(2\theta))}{\partial\theta_{\star}}\right|\lesssim 1$.
	
	$\left|\frac{1}{\sin(2\theta_{\star})}\frac{\partial(\sin^2\theta_{\star})}{\partial\theta_{\star}}\right|=1$,
	$\left|\frac{1}{\sin(2\theta_{\star})}\frac{\partial(\sin^2(2\theta_{\star}))}{\partial\theta_{\star}}\right|\lesssim 1$,
	$\left|\frac{1}{\sin(2\theta_{\star})}\frac{\partial(\cos(2\theta_{\star}))}{\partial\theta_{\star}}\right|=2$.
\end{enumerate}
\end{altproof}
\begin{altproof}{\eqref{BB}}
It was shown in~\eqref{deriv_bounds} that
 $\left|\frac{\partial\theta}{\partial{r_{\star}}}\right|\lesssim 
 |\Delta_r|\sin(2\theta_{\star})$
(this is \eqref{BB} with $k_1=1$).
Hence, inequalities~\eqref{BB} follow from (i), (iii) and (v) (with the bound $1$).
\end{altproof}
\begin{altproof}{\eqref{C}}
We have seen in~\eqref{deriv_bounds} that
$\left|\frac{\partial\theta}{\partial{\theta_{\star}}}\right|\lesssim 1$
(this is \eqref{C} with $k_1=0$). 
Thus, inequalities~\eqref{C} are a consequence of (ii), (iv) and (vi).
\end{altproof}
\begin{altproof}{\eqref{D}} 
Inequalities~\eqref{D} result from~\eqref{C}, (i), (iii) and (v) (with the bound $|\Delta_r|$).
\end{altproof}
\begin{altproof}{\eqref{bound1:L} and \eqref{bound2:L}}
These inequalities ensue from~\eqref{LLL}, \eqref{AA} to~\eqref{D} and~\eqref{good1}.
\end{altproof}
\noindent This completes the proof of Lemma \ref{A.6}.
\end{altproof}
\begin{rmk}
	Note that \eqref{bound1:L} and \eqref{bound2:L} hold for any smooth function
	of $r_\star$ and $\sin^2\theta_{\star}$.
\end{rmk}

\begin{lem}\label{smooth-functions}
	The functions
	$\frac{1}{\sin(2\theta_{\star})}\partial_{\theta_{\star}}h$,
	 $b^{\phi_{\star}}$, $\gamma_{\theta_{\star},\theta_{\star}}$,
	 $\frac{\gamma_{\theta_{\star}\phi_{\star}}}{\sin^2\theta_{\star}\sin(2\theta_\star)}$,
	 $\frac{\gamma_{\phi_{\star}\phi_{\star}}}{\sin^2\theta_{\star}}$
	 and
	 $\frac{1}{\sin^2\theta_{\star}}\left(\gamma_{\theta_{\star}\theta_{\star}}
	 -\,\frac{\gamma_{\phi_{\star}\phi_{\star}}}{\sin^2(\theta_{\star})}\right)$
	are smooth.
\end{lem}
\begin{proof}
Since $h(0,\theta_{\star})=0$, we have that $\partial_{\theta_{\star}}h(0,\theta_{\star})=0$. So,
as $\partial_{r_{\star}}h=-B$, it follows that
\begin{eqnarray*}
\partial_{\theta_{\star}}h(r_{\star},\theta_{\star})&=&\int_0^{r_\star}
(\partial_{r_{\star}}\partial_{\theta_{\star}}h)(r',\theta_{\star})\md r'\ =\ 
\int_0^{r_\star}
(\partial_{\theta_{\star}}\partial_{r_{\star}}h)(r',\theta_{\star})\md r'\\
&=&-\int_0^{r_\star}(
(\partial_{\theta_{\star}}r)(\partial_{r}B)+
	(\partial_{\theta_{\star}}\theta)(\partial_\theta B))(r',\theta_{\star})
\md r'\\
&=&-\sin(2\theta_{\star})\int_0^{r_\star}\Biggl(\left(\frac{1}{\sin(2\theta_{\star})}\partial_{\theta_{\star}}r\right)(\partial_{r}B)\Biggr.\\
&&\qquad\qquad\qquad+\Biggl.
(\partial_{\theta_{\star}}\theta)\left(\frac{1}{\sin(2\theta)}\partial_\theta B\right)\frac{\sin(2\theta)}{\sin(2\theta_{\star})}\Biggr)(r',\theta_{\star})
\md r'.
\end{eqnarray*}
Recall from~\eqref{r_t*} that $\partial_{\theta_{\star}}r$ has a factor
$\Delta_r\sin(2\theta_{\star})$ and observe that
$$
\frac{1}{\sin(2\theta)}
\partial_\theta B=\frac{\Xi a^3\Delta_r((r^2+a^2)\left(1-\,\frac{\Lambda}{3}r^2\right)-\Delta_r)}{\Upsilon^2}.
$$
This shows that $\frac{1}{\sin(2\theta_{\star})}\partial_{\theta_{\star}}h$
has derivatives of all orders with respect to $\partial_{r_{\star}}$ and
$\frac{1}{\sin(2\theta_{\star})}\partial_{\theta_{\star}}$.
Moreover, as $b^{\phi_{\star}}=2B$, we obtain
$$
\left|
\frac{1}{\sin(2\theta_{\star})}\partial_{\theta_{\star}}b^{\phi_{\star}}
\right|\lesssim|\Delta_r|
$$
and $\gamma_{\theta_{\star}\theta_{\star}}$ is smooth.

Note that $\frac{\sin^2\theta}{\sin^2\theta_{\star}}$ has derivatives of all
orders with respect to $r_\star$ and $\sin^2\theta_{\star}$. Indeed, we have
\begin{eqnarray*}
\partial_{r_{\star}}
\left(\frac{\sin^2\theta}{\sin^2\theta_{\star}}\right)&=&
\frac{\sin(2\theta)}{\sin^2\theta_{\star}}\frac{\partial\theta}{\partial r_\star}.
\end{eqnarray*}
As $\frac{\partial\theta}{\partial r_\star}$ has a factor $\sin(2\theta_\star)$,
the right-hand side has a factor $\sin(2\theta)\sin(2\theta_\star)$.
Taking~\eqref{good1} into account, differentiability at $\theta_{\star}$
equal to $\frac{\pi}{2}$ is not a problem. Neither is there a problem
at $\theta_{\star}$ equal to zero because
$$
\frac{\sin(2\theta)\sin(2\theta_{\star})}{\sin^2\theta_{\star}}=
4\frac{\sin(2\theta)}{\sin(2\theta_{\star})}(1-\sin^2\theta_\star).
$$
Differentiability with respect to $\sin^2\theta_{\star}$
holds at $\theta_{\star}$ equal to zero because
$$
\frac{\sin^2\theta}{\sin^2\theta_{\star}}=
\left(\frac{\sin(2\theta)}{\sin(2\theta_{\star})}\right)^2\frac{1-\sin^2\theta_{\star}}{1-\sin^2\theta},
$$
and differentiability holds at $\theta_{\star}$ equal to $\frac{\pi}{2}$ because
\begin{eqnarray*}
\frac{1}{\sin(2\theta_{\star})}\partial_{\theta_{\star}}
\left(\frac{\sin^2\theta}{\sin^2\theta_{\star}}\right)&=&
\frac{\sin(2\theta)}{\sin(2\theta_{\star})}
\frac{\partial\theta}{\partial\theta_\star}
\frac{1}{\sin^2\theta_{\star}}-\,\frac{\sin^2\theta}{\sin^2\theta_{\star}}
\frac{1}{\sin^2\theta_{\star}}.
\end{eqnarray*}
So, we see that
$$
\frac{\gamma_{\theta_{\star}\phi_{\star}}}{\sin^2\theta_{\star}\sin(2\theta_\star)}\qquad
\mbox{and}\qquad\frac{\gamma_{\phi_{\star}\phi_{\star}}}{\sin^2\theta_{\star}}
$$
are smooth.
Using~\eqref{Upsilon} and~\eqref{LLL}, we get
\begin{eqnarray}
&&\frac{1}{\sin^2\theta_{\star}}\left(\gamma_{\theta_{\star}\theta_{\star}}
-\,\frac{\gamma_{\phi_{\star}\phi_{\star}}}{\sin^2(\theta_{\star})}\right)
\label{telheiro}\\
&&=
\frac{1}{\sin^2\theta_{\star}}\left(\frac{\rho^2L^2}{\Upsilon}
-\,\frac{\sin^2\theta}{\sin^2\theta_{\star}}\frac{\Upsilon}{\Xi^2\rho^2}\right)
+\sin^2(2\theta_{\star})\frac{\sin^2\theta}{\sin^2\theta_{\star}}
\frac{\Upsilon}{\Xi^2\rho^2}
\left(\frac{\partial_{\theta_{\star}}h}{\sin(2\theta_{\star})}\right)^2\nonumber\\
&&=
\frac{1}{\Upsilon\rho^2\sin^2\theta_{\star}}\left({\rho^4L^2}
-\,\frac{\sin^2\theta}{\sin^2\theta_{\star}}\frac{\Upsilon^2}{\Xi^2}\right)
+\mbox{smooth}\nonumber\\
&&=
\frac{1}{\Xi^2\Upsilon\rho^2\sin^2\theta_{\star}}
\left((r^2+a^2\cos^2\theta)^2\Xi^2((r^2+a^2)^2-a^2\sin^2\theta_{\star}\Delta_r)
\frac{\sin^2(2\theta)}{\sin^2(2\theta_{\star})}\right.\nonumber\\
&&\qquad\qquad\qquad\ \ \left.
-\,\frac{\sin^2\theta}{\sin^2\theta_{\star}}((r^2+a^2)^2\Delta_\theta-a^2\sin^2\theta\Delta_r)^2
\right)
+\mbox{smooth}\nonumber\\
&&=
\frac{1}{\Xi^2\Upsilon\rho^2\sin^2\theta_{\star}}\frac{\sin^2\theta}{\sin^2\theta_{\star}}
\frac{1}{\cos^2\theta_{\star}}
\left[(r^2+a^2\cos^2\theta)^2\Xi^2((r^2+a^2)^2-a^2\sin^2\theta_{\star}\Delta_r)
\cos^2\theta\right.\nonumber\\
&&\qquad\qquad\qquad\qquad\qquad\qquad\ \ \left.
-((r^2+a^2)^2\Delta_\theta-a^2\sin^2\theta\Delta_r)^2\cos^2\theta_{\star}
\right]
+\mbox{smooth}.\nonumber\\
&&=
\frac{1}{\Xi^2\Upsilon\rho^2}\frac{\sin^2\theta}{\sin^2\theta_{\star}}
\frac{1}{1-\sin^2\theta_{\star}}\frac{p(\sin^2\theta,\sin^2\theta_{\star})}{\sin^2\theta_{\star}}+\mbox{smooth}.
\nonumber
\end{eqnarray}
The expression inside the square parenthesis is a polynomial $p$ in $\alpha:=\sin^2\theta$
and $\beta:=\sin^2\theta_{\star}$ satisfying $p(0,0)=0$. So,
$p(\alpha,\beta)=D_1p(0,0)\alpha+D_2p(0,0)\beta+$ higher order terms.
Thus, it is clear that $\frac{p(\sin^2\theta,\sin^2\theta_{\star})}{\sin^2\theta_{\star}}$
is smooth, and so it is clear that~\eqref{telheiro} is smooth at zero.
Differentiability at $\frac{\pi}{2}$ is not an issue since
$$
\frac{1}{\sin(2\theta_{\star})}\partial_{\theta_{\star}}\frac{1}{\sin^2\theta_{\star}}
=-\,\frac{1}{\sin^4\theta_{\star}}.
$$
\end{proof}

\begin{rmk}
	The Christoffel symbols
	of the metric $\gamma$ are bounded with the exception of 
	$\GGamma^{\phi_{\star}}_{\theta_{\star}\phi_{\star}}$
	which blows up precisely like $\frac{1}{\sin\theta_{\star}}$.
\end{rmk}
\begin{proof}
Because $\gamma$ behaves like the round metric on $\bbS^2$,
$\gamma^{-1}$ behaves like the inverse of the round metric on $\bbS^2$, in that $\gamma^{\phi_{\star}\phi_{\star}}$ blows up at $\theta_{\star}=0$ 
like $\frac{1}{\sin^2\theta_{\star}}$. Hence, the Christoffel symbols
of the metric $\gamma$ also behave like the ones of the round metric on $\bbS^2$,
namely, they are all smooth with the exception of 
$\GGamma^{\phi_{\star}}_{\theta_{\star}\phi_{\star}}$,
\begin{eqnarray*}
\GGamma^{\phi_{\star}}_{\theta_{\star}\phi_{\star}}&=&\frac{1}{2}
\gamma^{\phi_{\star}\alpha}(\partial_{\theta_{\star}}\gamma_{\alpha\phi_{\star}}
+\partial_{\phi_{\star}}\gamma_{\theta_{\star}\alpha}-\partial_\alpha\gamma_{\theta_{\star}\phi_{\star}})\\
&=&\mbox{smooth}+\frac{1}{2}
\gamma^{\phi_{\star}{\phi_{\star}}}(\partial_{\theta_{\star}}\gamma_{{\phi_{\star}}\phi_{\star}}
+\partial_{\phi_{\star}}\gamma_{\theta_{\star}\phi_{\star}}-\partial_{\phi_{\star}}\gamma_{\theta_{\star}\phi_{\star}})\\
&=&\frac{1}{2}
\gamma^{\phi_{\star}{\phi_{\star}}}\partial_{\theta_{\star}}\gamma_{{\phi_{\star}}\phi_{\star}}+\mbox{smooth}\\
&=&\frac{1}{2}\frac{\Xi^2 \rho ^2}{\Upsilon  \sin ^2\theta} \partial_{\theta_{\star}}
\left(\frac{\sin^2\theta}{\Xi^2\rho^2}\Upsilon\right)+\mbox{smooth}\\
&=&\frac{1}{2}\frac{1}{\sin^2\theta}\partial_{\theta_{\star}}(\sin^2\theta)+\mbox{smooth}\\
&=&\frac{\cos\theta}{\sin\theta}\frac{\partial\theta}{\partial\theta_{\star}}+\mbox{smooth}\\
&=&\Delta_{\theta}\frac{(r^2+a^2)^2-a^2\sin^2\theta_{\star}\Delta_r}{ (r^2+a^2)^2\Delta_{\theta}-a^2\sin^2\theta\Delta_{r}}{\frac{\sin(2\theta)}{\sin(2\theta_{\star})}}\frac{\cos\theta}{\sin\theta}+\mbox{smooth}\\
&=&\Delta_{\theta}\frac{(r^2+a^2)^2-a^2\sin^2\theta_{\star}\Delta_r}{ (r^2+a^2)^2\Delta_{\theta}-a^2\sin^2\theta\Delta_{r}}{\frac{\cos^2\theta}{\cos\theta_{\star}}}\frac{1}{\sin\theta_{\star}}+\mbox{smooth},
\end{eqnarray*}
which blows up precisely like $\frac{1}{\sin\theta_{\star}}$.
\end{proof}

\subsubsection{Regularity at $\theta_\star=0$ and $\theta_\star=\pi$}
\label{polo}

Note that the coordinates that are constructed are not just locally defined 
but are well-defined globally on a manifold diffeomorphic to $\mathbb{R}^2\times
\mathbb{S}^2$.

\paragraph{Regularity of the metric.}

We check the regularity of the metric at $\theta_\star=0$ using coordinates 
\begin{equation}
\begin{array}{rcl}
x&=&\sin\theta_\star\cos\phi_\star,\\
y&=&\sin\theta_\star\sin\phi_\star.
\end{array}\label{coordinates}
\end{equation}
\begin{lem}
	The metric is smooth at $\theta_\star=0$ and 
	\be\label{paiva}
	\lim_{\theta_\star\to 0}g=
	-2\Omega^2 \left(\mmd u\otimes \mmd v+\md v\otimes \mmd u\right)+
	\frac{r^2+a^2}{\Xi}\left(\frac{\partial\theta}{\partial\theta_\star}\right)^2(r_\star,0)\left(\md x\otimes\md x+\md y\otimes\md y\right).
	\ee
\end{lem}
\begin{proof}
	The relations
	\begin{eqnarray*}
		\md\theta_\star&=&\ \ \, \frac{\cos\phi_\star}{\cos\theta_\star}\md x
		+\frac{\sin\phi_\star}{\cos\theta_\star}\md y,\\
		\md\phi_\star&=&-\,\frac{\sin\phi_\star}{\sin\theta_\star}\md x
		+\frac{\cos\phi_\star}{\sin\theta_\star}\md y
	\end{eqnarray*}
	imply that
	\begin{eqnarray*}
		&&\,\gamma_{\theta_{\star}\theta_{\star}}
		\md\theta_{\star}\otimes \mmd\theta_{\star}
		+\gamma_{\theta_{\star}\phi_{\star}}\md\theta_{\star}\otimes\mmd\phi_{\star}
		+\gamma_{\theta_{\star}\phi_{\star}}\mmd\phi_{\star}\otimes
		\mmd\theta_{\star}\nonumber+\gamma_{\phi_{\star}\phi_{\star}}
		\mmd\phi_{\star}\otimes\mmd\phi_{\star}\\
		&&\ \ \ =\gamma_{\theta_{\star}\theta_{\star}}\left(
		\frac{\cos^2\phi_\star}{\cos^2\theta_\star}\md x\otimes\md x
		+\frac{\sin(2\phi_\star)}{2\cos^2\theta_\star}\md x\otimes\md y
		+\frac{\sin(2\phi_\star)}{2\cos^2\theta_\star}\md y\otimes\md x
		+\frac{\sin^2\phi_\star}{\cos^2\theta_\star}\md y\otimes\md y
		\right)\\
		&&\qquad+\gamma_{\theta_{\star}\phi_{\star}}\left(
		-2\frac{\sin(2\phi_\star)}{\sin(2\theta_\star)}\md x\otimes\md x
		+\frac{\cos(2\phi_\star)}{\sin(2\theta_\star)}\md x\otimes\md y
		+\frac{\cos(2\phi_\star)}{\sin(2\theta_\star)}\md y\otimes\md x
		+2\frac{\sin^2\phi_\star}{\sin(2\theta_\star)}\md y\otimes\md y
		\right)\\
		&&\qquad+\gamma_{\phi_{\star}\phi_{\star}}\left(
		\frac{\sin^2\phi_\star}{\sin^2\theta_\star}\md x\otimes\md x
		-\,\frac{\sin(2\phi_\star)}{2\sin^2\theta_\star}\md x\otimes\md y
		-\,\frac{\sin(2\phi_\star)}{2\sin^2\theta_\star}\md y\otimes\md x
		+\frac{\cos^2\phi_\star}{\sin^2\theta_\star}\md y\otimes\md y
		\right).
	\end{eqnarray*}
	Using~\eqref{Upsilon}, \eqref{gamma-tt}, \eqref{gamma-ff} and~\eqref{LLL}, we have
	\begin{eqnarray*}
		\lim_{\theta_\star\to 0}\gamma_{\theta_{\star}\theta_{\star}}&=&
		\frac{r^2+a^2}{\Xi}\left(\frac{\partial\theta}{\partial\theta_\star}\right)^2(r_\star,0),\\
		\lim_{\theta_\star\to 0}\frac{\gamma_{\phi_{\star}\phi_{\star}}}{\sin^2\theta_\star}&=&
		\frac{r^2+a^2}{\Xi}\left(\frac{\partial\theta}{\partial\theta_\star}\right)^2(r_\star,0).
	\end{eqnarray*}
	As, by~\eqref{gamma-tf},
	$$
	\lim_{\theta_\star\to 0}\frac{\gamma_{\theta_\star\phi_\star}}{\sin(2\theta_\star)}=0,
	$$
	we obtain
	$$
	\lim_{\theta_\star\to 0}\gamma=\frac{r^2+a^2}{\Xi}\left(\frac{\partial\theta}{\partial\theta_\star}\right)^2(r_\star,0)\left(\md x\otimes\md x+\md y\otimes\md y\right).
	$$
	The trigonometric functions of $\phi_\star$, which would make the metric discontinuous
	at $\theta_\star=0$, have disappeared.
	Moreover, since
	$$
	\lim_{\theta_\star\to 0}b^{\phi_{\star}}\gamma_{\theta_{\star}\phi_{\star}}\md\theta_{\star}\otimes\md v=0,
	\qquad
	\lim_{\theta_\star\to 0}b^{\phi_{\star}}\gamma_{\theta_{\star}\phi_{\star}}\md v\otimes\md\theta_{\star}=0
	$$
	and
	$$
	\lim_{\theta_\star\to 0}\gamma_{\phi_\star\phi_\star}=0,
	$$
	we conclude that~\eqref{paiva} holds.
	This shows that the metric is continuous at $\theta_\star=0$.
	Lemma~\ref{smooth-functions} implies that the extension of
	the metric is smooth.
\end{proof}

\paragraph{Calculation of $\partial_{\theta_\star}\theta(r_\star,0)$.}
\begin{lem} We have that
	\be\label{braga}
	\frac{\partial\theta}{\partial\theta_\star}(r_\star,0)=
	\frac{1}{\frac{\partial\theta_\star}{\partial\theta}(r,0)}=
	\sqrt{\Xi}\sin\left(
	\arcsin\left(\frac{1}{\sqrt{\Xi}}\right)+
	\arctan\left(\frac{r}{a}\right)
	-\,\arctan\left(\frac{r_+}{a}\right)
	\right).
	\ee
\end{lem}
\begin{proof}
	Using the definition of $F$ in~\eqref{F_imp_def2}, the equation
	$
	F(r,\theta,\theta_\star(r,\theta))=0 
	$
	can be written as 
	\be\label{fim}
	\int_\theta^{\theta_\star}\frac{\md\theta'}{P(\theta',\theta_\star)}
	=\int_r^{r_+}\frac{\md r'}{Q(r',\theta_\star)}.
	\ee
	We will take the limit of both sides of~\eqref{fim} as $\theta$ goes to $0$.
	As $\lim_{\theta\to 0}\theta_\star(r,\theta)=0$, uniformly in $r$,
	$$
	\lim_{\theta\to 0}\int_r^{r_+}\frac{\md r'}{Q(r',\theta_\star)}=
	\int_r^{r_+}\frac{\md r'}{(r')^2+a^2}=\frac{1}{a}\arctan\left(\frac{r_+}{a}\right)-
	\frac{1}{a}\arctan\left(\frac{r}{a}\right).
	$$
	Using the substitution $s=\frac{\sin\theta'}{\sin\theta_\star}$,
	\begin{eqnarray*}
		\int_\theta^{\theta_\star}\frac{\md\theta'}{P(\theta',\theta_\star)}&=&
		\frac{1}{a}
		\int_{\frac{\sin\theta}{\sin\theta_\star}}^1
		\frac{1}{\cos\theta'}\frac{\md s}{\sqrt{\left(1+\frac{\Lambda}{3}a^2\cos^2\theta'\right)-s^2}},
	\end{eqnarray*}
	with $\cos\theta'=\sqrt{1-s^2\sin^2\theta_\star}$.
	The last integrand converges uniformly to
	$$
	\frac{1}{\sqrt{\Xi-s^2}}
	$$
	as $\theta$ goes to zero. Thus, we get
	\begin{eqnarray*}
		\lim_{\theta\to 0}
		\int_\theta^{\theta_\star}\frac{\md\theta'}{P(\theta',\theta_\star)}
		&=&\frac{1}{a}\int_{1\left/\left(\frac{\partial\theta_\star}{\partial\theta}(r,0)\right)\right.}^1
		\frac{\md s}{\sqrt{\Xi-s^2}}=\frac{1}{a}\left.
		\arcsin\left(\frac{s}{\sqrt{\Xi}}\right)
		\right|_{1\left/\left(\frac{\partial\theta_\star}{\partial\theta}(r,0)\right)\right.}^1\\
		&=&\frac{1}{a}\arcsin\left(\frac{1}{\sqrt{\Xi}}\right)-\frac{1}{a}
		\arcsin\left(\frac{1}{\sqrt{\Xi}\frac{\partial\theta_\star}{\partial\theta}(r,0)}\right).
	\end{eqnarray*}
	So, from~\eqref{fim} we conclude that
	$$
	\arcsin\left(\frac{1}{\sqrt{\Xi}\frac{\partial\theta_\star}{\partial\theta}(r,0)}\right)=\arcsin\left(\frac{1}{\sqrt{\Xi}}\right)+
	\arctan\left(\frac{r}{a}\right)
	-\,\arctan\left(\frac{r_+}{a}\right).
	$$
	We know that the right-hand side of this equality is positive 
	because we guaranteed that~\eqref{lira} holds. 
	And the right-hand side is obviously smaller than $\frac{\pi}{2}$.
	Therefore, we have
	$$
	\frac{\partial\theta_\star}{\partial\theta}(r,0)=
	\frac{1}{\sqrt{\Xi}}\csc\left(
	\arcsin\left(\frac{1}{\sqrt{\Xi}}\right)+
	\arctan\left(\frac{r}{a}\right)
	-\,\arctan\left(\frac{r_+}{a}\right)
	\right).
	$$
	This is strictly greater than $1$ (and goes to $1$ as $r\nearrow r_+$).
	According to~\eqref{QP} and~\eqref{UQP},
	the derivative of the map $(r,\theta)\mapsto(r_\star,\theta_\star)$ at 
	$(r,0)$ is represented by the matrix
	$$
	\left[
	\begin{array}{cc}
	\frac{r^2+a^2}{\Delta_r}&0\\
	0&\frac{\partial\theta_\star}{\partial\theta}(r,0)
	\end{array}
	\right].
	$$
	Equality~\eqref{braga} follows.
\end{proof}

\paragraph{Regularity of functions at the poles.}
Using the change of coordinates~\eqref{coordinates},
the variable $\theta_\star$ is written in terms of $x$
and $y$ as
$$
\theta_\star=\arcsin\sqrt{x^2+y^2}.
$$
Given a function $f$ that transforms pairs $(r_\star,\theta_\star)$,
we want to study the differentiability of
$$
(r_\star,x,y)
\stackrel{\hat{f}}{
	\longmapsto} f\left(r_\star,\arcsin\sqrt{x^2+y^2}\right).
$$
\begin{rmk}\label{polo1}
	Let $f:(-\infty,+\infty)\times\left[0,\pi\right]\to\bbR$ be $C^1$ such that
	$
	\frac{1}{\sin(2\theta_\star)}\left(
	\partial_{\theta_\star}f\right)$ has a limit when $\theta_\star=0$.
	Then $\hat{f}$ is $C^1$.
\end{rmk}
\begin{proof}
	The derivative of $\hat{f}$ with respect to $x$ is
	\begin{eqnarray}
	\partial_x\hat{f}&=&
	\left(
	\partial_{\theta_\star}f\right)\frac{1}{\sqrt{1-(x^2+y^2)}}\frac{x}
	{\sqrt{x^2+y^2}}.\label{direita}
	\end{eqnarray}
	We observe that, although the map $(x,y)\mapsto\sqrt{x^2+y^2}$
	is not differentiable at the origin, when $x=y=0$
	the quotient $\frac{x}
	{\sqrt{x^2+y^2}}$ in~\eqref{direita} is
	$+1$ or $-1$, according to the calculation of a right
	or a left derivative.
	If we assume 
	that 
	$
	\frac{1}{\sin(2\theta_\star)}\left(
	\partial_{\theta_\star}f\right)$ has a limit at $((r_\star)_0,0)$, then 
	$\partial_{\theta_\star}f((r_\star)_0,0)=0$, so
	there is no indetermination 
	in~\eqref{direita}, and
	$$
	\partial_x\hat{f}(r_\star,0,0)=0.
	$$
	The function $\hat{f}$ has continuous partial derivatives 
	with respect to $r_\star$, $x$ and $y$ in a neighborhood
	of each point $((r_\star)_0,0,0)$, and so it is 
	differentiable. 
\end{proof}
Note that we can write~\eqref{direita} as
\be\label{first}
\partial_x\hat{f}
\ =\
\left(
\partial_{\theta_\star}f\right)\frac{1}{\cos\theta_\star}\frac{x}
{\sin\theta_\star}
\ =\ 2x\frac{1}{\sin(2\theta_\star)}\left(
\partial_{\theta_\star}f\right).
\ee
(This is the derivative of $f$ with respect to $\sin^2\theta_\star=x^2+y^2$
multiplied by the derivative of $\sin^2\theta_\star$ with respect to $x$.)
Expression~\eqref{first} shows that if the quotient 
$\frac{1}{\sin(2\theta_\star)}\left(
\partial_{\theta_\star}f\right)
$ were unbounded, then one might have problems with the differentiability 
when $(x,y)=(0,0)$. For example, if this quotient were to behave like 
$\frac{1}{\sin\theta_\star}$ around $((r_\star)_0,0)$, then $\partial_x\hat{f}$ would behave
like $\cos\phi_\star$, so that $\partial_x\hat{f}((r_\star)_0,0,0)$
would not exist.
\begin{rmk}\label{polo3}
	Let $f:(-\infty,+\infty)\times\left[0,\pi\right]\to\bbR$ be $C^\infty$ 
	in the sense of\/ {\rm Remark~\ref{sense}}.
	Then $\hat{f}$ is $C^\infty$.
\end{rmk}
\begin{proof}
	It is clear that $\partial_{r_\star}^2\hat{f}$,
	$\partial_x\partial_{r_\star}\hat{f}$ and
	$\partial_y\partial_{r_\star}\hat{f}$ are continuous.
	The continuity of $\partial_x^2\hat{f}$ is a consequence of
	$$
	\partial_x^2\hat{f}=
	\partial_x\left(
	2x\frac{1}{\sin(2\theta_\star)}\left(
	\partial_{\theta_\star}f\right)
	\right)=2\frac{1}{\sin(2\theta_\star)}\left(
	\partial_{\theta_\star}f\right)+4x^2\left(
	\frac{1}{\sin(2\theta_\star)}\partial_{\theta_\star}
	\right)^2f.
	$$ 
	The continuity of $\partial_x\partial_y\hat{f}$ and
	$\partial_y^2\hat{f}$ follow in the same way. So $\hat{f}$ is $C^2$.
	One proves by induction that $\hat{f}$ is $C^\infty$.
\end{proof}

\begin{rmk}\label{polo2}
	A differentiable function such that $f(r_\star,\theta_\star)
	=f(r_\star,\pi-\theta_\star)$ satisfies 
	$\partial_{\theta_\star}f\left(r_\star,\frac{\pi}{2}\right)=0$.
	The quotient $\frac{1}{\sin(2\theta_\star)}\partial_{\theta_\star}f$
	has a finite limit at $\left(r_\star,\frac{\pi}{2}\right)$,
	provided that the numerator is analytic. 
\end{rmk}
Indeed,
both the numerator and the denominator vanish at $\left(r_\star,\frac{\pi}{2}\right)$ and the denominator has a first order 
zero there.

\begin{rmk}\label{corona}
	Suppose that $f:(-\infty,+\infty)\times[0,\pi]\times S^1\to\bbR$ is smooth
	and define
	$$
	\hat{f}(r_\star,x,y)=f\left(r_\star,\arcsin\sqrt{x^2+y^2},\arg(x+iy)\right),
	$$
	where $\phi_\star=\arg(x+iy)=\arctan\frac{y}{x}$ for $x>0$, and otherwise 
	$\arg(x+iy)$ is $\arctan\frac{y}{x}$ with an appropriate constant added.
	Then
\begin{eqnarray}
\partial_x\hat{f}&=&\frac{x}{\sqrt{1-(x^2+y^2)}\sqrt{x^2+y^2}}\partial_{\theta_\star}f
-\,\frac{y}{x^2+y^2}\partial_{\phi_\star}f\nonumber\\
&=&\frac{\cos\phi_\star}{\cos\theta_\star}\partial_{\theta_\star}f-\,
\frac{\sin\phi_\star}{\sin\theta_\star}\partial_{\phi_\star}f,\nonumber\\
\partial_y\hat{f}&=&\frac{y}{\sqrt{1-(x^2+y^2)}\sqrt{x^2+y^2}}\partial_{\theta_\star}f
+\frac{x}{x^2+y^2}\partial_{\phi_\star}f\nonumber\\
&=&\frac{\sin\phi_\star}{\cos\theta_\star}\partial_{\theta_\star}f+
\frac{\cos\phi_\star}{\sin\theta_\star}\partial_{\phi_\star}f.\nonumber
\end{eqnarray}

\end{rmk}

\paragraph{Regularity of the change of coordinates
	$(t,r,\theta,\phi)\mapsto(t,r_\star,\theta_\star,\phi)$.}
Define
$$ 
\begin{array}{rcl}
\hat{x}&=&\sin\hat{\theta}_\star\cos{\phi},\\
\hat{y}&=&\sin\hat{\theta}_\star\sin{\phi},
\end{array}\qquad\qquad
\begin{array}{rcl}
\tilde{x}&=&\sin\theta\cos\phi,\\
\tilde{y}&=&\sin\theta\sin\phi.
\end{array}
$$
\begin{enumerate}[(i)]
	\item 
The map $(t,r_\star,\hat{x},\hat{y})\mapsto(t,r,\tilde{x},\tilde{y})$ is smooth.
The derivative $\partial_{r_\star}r$ is given by~\eqref{r_r*}. Moreover,
using~\eqref{first}, we get
\begin{eqnarray*}
	\partial_{\hat{x}} r&=&2{\hat{x}} \frac{1}{\sin(2\hat{\theta}_\star)}
	\partial_{\theta_\star}r\quad (\mbox{see}~\eqref{r_t*}\ \mbox{for}\ 
	\partial_{\theta_\star}r),\\
	\partial_{\hat{y}} r&=&2{\hat{y}} \frac{1}{\sin(2\hat{\theta}_\star)}
	\partial_{\theta_\star}r,\\
	\partial_{r_\star}\sin\theta&=&\cos\theta
	\frac{\partial\theta}{\partial r_\star}\ =\
	2\frac{\sin(2\hat{\theta}_\star)}{\sin(2\theta)}\sin\theta\cos^2\theta
	\left(
	\frac{1}{\sin(2\hat{\theta}_\star)}\partial_{r_\star}\theta\right)
	\quad (\mbox{see}~\eqref{t_r*}\ \mbox{for}\ 
	\partial_{r_\star}\theta).
\end{eqnarray*}

Applying Remark~\ref{corona}, we obtain
\begin{eqnarray*}
	\partial_{\hat{x}}\tilde{x}&=&\frac{\cos\theta}{\cos\theta_\star}\cos^2\phi\,
	\partial_{\theta_\star}\theta
	+
	\frac{\sin\theta}{\sin\theta_\star}\sin^2\phi,\\
	\partial_{\hat{y}}\tilde{x}&=&\frac{\cos\theta}{\cos\theta_\star}\sin\phi\cos\phi\,
	\partial_{\theta_\star}\theta
	-\,
	\frac{\sin\theta}{\sin\theta_\star}\sin\phi\cos\phi,\\
	\partial_{\hat{x}}\tilde{y}&=&\frac{\cos\theta}{\cos\theta_\star}\sin\phi\cos\phi\,
	\partial_{\theta_\star}\theta
	-\,
	\frac{\sin\theta}{\sin\theta_\star}\sin\phi\cos\phi,\\
	\partial_{\hat{y}}\tilde{y}&=&\frac{\cos\theta}{\cos\theta_\star}\sin^2\phi\,
	\partial_{\theta_\star}\theta
	+
	\frac{\sin\theta}{\sin\theta_\star}\cos^2\phi.
\end{eqnarray*}
Notice that when $\theta=0$, we have
\begin{eqnarray*}
	\partial_{\hat{x}}\tilde{x}&=&
	\partial_{\theta_\star}\theta,\\
	\partial_{\hat{y}}\tilde{x}&=&0,\\
	\partial_{\hat{x}}\tilde{y}&=&0,\\
	\partial_{\hat{y}}\tilde{y}&=&
	\partial_{\theta_\star}\theta.
\end{eqnarray*}
The quotient $\frac{\cos\theta}{\cos\theta_\star}$ is smooth because
\be\label{analogous}
\frac{1}{\sin(2\theta_\star)}\partial_{\theta_\star}\left(
\frac{\cos\theta}{\cos\theta_\star}\right)=
-\,\frac{1}{2}\frac{\sin(2\theta)}{\sin(2\theta_\star)}
\frac{\cos\theta_\star}{\cos\theta}\frac{\partial\theta}{\partial\theta_\star}
+\frac{1}{2}\frac{\cos\theta}{\cos\theta_\star}\frac{1}{\cos^2\theta_\star}.
\ee
The quotient $\frac{\sin\theta}{\sin\theta_\star}$ is smooth because
$$
\frac{\sin\theta}{\sin\theta_\star}=\frac{\sin(2\theta)}{\sin(2\theta_\star)}
\frac{\cos\theta_\star}{\cos\theta}.
$$
Thus $r$, $\tilde{x}$ and $\tilde{y}$ are $C^\infty$ functions of 
$r_\star$, $\hat{x}$ and $\hat{y}$.
\item
The map $(t,r,\tilde{x},\tilde{y})\mapsto(t,r_\star,\hat{x},\hat{y})$ is smooth.
Recall that $\partial_rr_\star$ is given in~\eqref{QP}, and
$$
\partial_r\sin\theta_\star=\cos\theta_\star\partial_r\theta_\star
=\cos\theta_\star\frac{1}{GQ}=
2\frac{\sin(2\theta)}{\sin(2\theta_\star)}\sin\theta_\star\cos^2\theta_\star
\left(
\frac{1}{\sin(2\theta)}\frac{1}{GQ}
\right),
$$
as $\partial_r\theta_\star$ is given in~\eqref{UQP}. 
Inequalities~\eqref{bounds_for_G} yield
$$
\lim_{\theta\to 0}\frac{1}{G(r,\theta,\theta_\star(r,\theta))}=0.
$$
One other consequence of~\eqref{QP} and~\eqref{UQP} is
\begin{eqnarray*}
\partial_{\tilde{x}}r_\star&=&2\tilde{x}\frac{1}{\sin(2\theta)}\partial_\theta r_\star\ =\
2a\tilde{x}\frac{\sqrt{\sin^2\theta_\star\Delta_\theta-\sin^2\theta}}{\Delta_\theta\sin(2\theta)},
\\
\partial_{\tilde{y}}r_\star&=&2\tilde{y}\frac{1}{\sin(2\theta)}\partial_\theta r_\star\ =\
2a\tilde{y}\frac{\sqrt{\sin^2\theta_\star\Delta_\theta-\sin^2\theta}}{\Delta_\theta\sin(2\theta)}.
\end{eqnarray*}
The formulas for $\partial_{\tilde{x}}\hat{x}$, $\partial_{\tilde{y}}\hat{x}$,
		$\partial_{\tilde{x}}\hat{y}$ and $\partial_{\tilde{y}}\hat{y}$,
are similar to the ones for 
$\partial_{\hat{x}}\tilde{x}$, $\partial_{\hat{y}}\tilde{x}$,
		$\partial_{\hat{x}}\tilde{y}$ and $\partial_{\hat{y}}\tilde{y}$
				 (interchange $\theta$ and $\theta_\star$).
				 Recall that $\partial_\theta\theta_\star=-\,\frac{1}{GP}$.
It follows that $r_\star$, $\hat{x}$ and $\hat{y}$ are $C^\infty$ functions of $r$, $\tilde{x}$
and $\tilde{y}$.
\end{enumerate}

\paragraph{Regularity of the change of coordinates
	$(t,r_\star,\theta_\star,\phi)\mapsto (t,r_\star,\theta_\star,\phi_\star)$.}
Recall that
$$
\phi_\star=\phi-h(r_\star,\theta_\star).
$$
The fact that both $(t,r_\star,\theta_\star,\phi)\mapsto (t,r_\star,\theta_\star,\phi_\star)$ and
$(t,r_\star,\theta_\star,\phi_\star)\mapsto (t,r_\star,\theta_\star,\phi)$ are
$C^\infty$ is a simple consequence of~\eqref{prh},
which gives $\partial_{r_\star}h$, and
Lemma~\ref{smooth-functions}, which gives $\frac{1}{\sin(2\theta_\star)}\partial_{\theta_\star}h$.\\
We remark that
\begin{eqnarray}
	\partial_x&=&\frac{\cos\phi_\star}{\cos\theta_\star}\partial_{\theta_\star}
	-\,\frac{\sin\phi_\star}{\sin\theta_\star}\partial_{\phi_\star},\\
	\partial_y&=&\frac{\sin\phi_\star}{\cos\theta_\star}\partial_{\theta_\star}
	+\frac{\cos\phi_\star}{\sin\theta_\star}\partial_{\phi_\star}.
\end{eqnarray}
Hence, the differentiability of
$(t,r_\star,x,y)\mapsto(t,r_\star,\hat{x},\hat{y})$ follows from
\begin{eqnarray*}
\partial_{r_\star}\hat{x}&=&-\sin\theta_\star\sin\phi\,\partial_{r_\star}h,\\
	\partial_x\hat{x}&=&
	\cos(h(r_\star,\theta_\star))
	-\,\frac{\sin\theta_\star}{\cos\theta_\star}\cos\phi_\star\sin\phi\,
	\partial_{\theta_\star}h,\\
	\partial_y\hat{x}&=&
	-\sin(h(r_\star,\theta_\star))
	-\,\frac{\sin\theta_\star}{\cos\theta_\star}\sin\phi_\star\sin\phi\,
	\partial_{\theta_\star}h,\\
\partial_{r_\star}\hat{y}&=&\sin\theta_\star\cos\phi\,\partial_{r_\star}h,\\
	\partial_x\hat{y}&=&
	\sin(h(r_\star,\theta_\star))
	+\frac{\sin\theta_\star}{\cos\theta_\star}\cos\phi_\star\cos\phi\,
	\partial_{\theta_\star}h,\\
	\partial_y\hat{y}&=&
	\cos(h(r_\star,\theta_\star))
	+\frac{\sin\theta_\star}{\cos\theta_\star}\sin\phi_\star\cos\phi\,
	\partial_{\theta_\star}h.
\end{eqnarray*}	
Note that these are smooth functions and that at $\theta_\star=0$ they
are independent of $\phi$ and $\phi_\star$.
The differentiability of
$(t,r_\star,\hat{x},\hat{y})\mapsto(t,r_\star,x,y)$ follows in a similar way.

\paragraph{Regularity of the spheres given by the intersection of
	hypersurfaces $u=\mbox{constant}$ and $v=\mbox{constant}$.}

It is obvious that the two atlases 
$\{(t,r_\star,\theta_\star,\phi_\star),(t,r_\star,x,y)\}$
and
${\cal A}_{\mbox{\tiny DN}}=\{(u,v,\theta_\star,\phi_\star),(u,v,x,y)\}$ are compatible. So, combining
the conclusions of the previous paragraphs, the two atlases
${\cal A}_{\mbox{\tiny BL}}=\{(t,r,\theta,\phi),(t,r,\tilde{x},\tilde{y})\}$
and
${\cal A}_{\mbox{\tiny DN}}$ are compatible.
The spheres given by the intersection of
hypersurfaces $u=\mbox{constant}$ and $v=\mbox{constant}$ are smooth
in the atlas ${\cal A}_{\mbox{\tiny DN}}$ by definition.
Therefore, they are smooth in the atlas ${\cal A}_{\mbox{\tiny BL}}$. This proves
\begin{thm}\label{thm1}
The topological two-spheres given by the intersection of
hypersurfaces $u=\mbox{constant}$ and $v=\mbox{constant}$ are $C^\infty$
in the Boyer--Lindquist coordinates.
\end{thm}

\subsection{Coordinates at the horizons}

\subsubsection{The decay of $\Omega^2$ at the horizons}

Recall that the surface gravities of the horizons are given by
\begin{eqnarray}
\kappa_-&=&-\,\frac{\Lambda}{6}\frac{(r_--r_c)(r_--r_+)(r_--r_n)}{r_-^2+a^2},\label{kappa-}\\
\kappa_+&=&-\,\frac{\Lambda}{6}\frac{(r_+-r_c)(r_+-r_-)(r_+-r_n)}{r_+^2+a^2},\nonumber\\
\kappa_c&=&-\,\frac{\Lambda}{6}\frac{(r_c-r_+)(r_c-r_-)(r_c-r_n)}{r_c^2+a^2}\nonumber
\end{eqnarray}
(confirm the formula for $\kappa_-$ with Example~\ref{surface-gravity}).

\begin{lem}
	Given $C_R\in\bbR$, there exist $c,C>0$ such that
	\bea
	ce^{2\kappa_-r_\star}\leq
	\Omega^2\leq Ce^{2\kappa_-r_\star}\ \mbox{for}\ r_\star\geq C_R,
	\label{omega1}\\
	ce^{2\kappa_+r_\star}\leq
	\Omega^2\leq Ce^{2\kappa_+r_\star}\ \mbox{for}\ r_\star\leq C_R.\,
	\label{omega2}
	\eea
\end{lem}
\begin{proof}
The formula
$$
\frac{\partial r}{\partial r_{\star}}=\frac{\Delta _{\theta } \Delta _r \sqrt{(r^2+a^2)^2-a^2 \sin ^2\theta_{\star} \Delta _r}}{\Upsilon}
$$
shows that $\frac{\partial r}{\partial r_{\star}}(r=r_-,\theta)=0$.
To write the linear approximation for this function, we start by calculating
\begin{eqnarray*}
\left.\partial_r\left(\frac{\partial r}{\partial r_{\star}}\right)\right|_{r=r_-}
&=&\left.\frac{\Delta _{\theta }(\partial_r \Delta _r) \sqrt{(r^2+a^2)^2-a^2 \sin ^2\theta_{\star} \Delta _r}}{\Upsilon}\right|_{r=r_-}
\\
&=&\left.\frac{\Delta _{\theta }(\partial_r \Delta _r) (r^2+a^2)}{(r^2+a^2)^2\Delta_{\theta }}\right|_{r=r_-}
=\ \frac{\left.(\partial_r\Delta_r)\right|_{r=r_-}}{r_-^2+a^2}
\ =\ 2\kappa_-.
\end{eqnarray*}
Hence, we get
$$
\frac{\partial r}{\partial r_{\star}}=2\kappa_-(r-r_-)+O((r-r_-)^2).
$$
According to Remark~\ref{swing},
there exists $C > 0$ and $(r_\star)_0$ such that for $r_\star\geq (r_\star)_0$
we have
$$
( 1 - C (r - r_-)) (1 + O (r - r_-)) \leq 1 \leq (1 + 
C (r - r_-)) (1 + O (r - r_-)).
$$
So, for $r\geq(r_\star)_0$ we obtain
$$
\frac{1}{r-r_-}-C\leq \frac{1}{\left(r-r_-\right) \left(1+O \left(r-r_-\right)\right)}\leq \frac{1}{r-r_-}+C.
$$
Remembering that $\kappa_-<0$, it follows that
$$
\left(\frac{1}{r-r_-}+C\right)\frac{\partial r}{\partial r_{\star}}\leq \frac{1}{\left(r-r_-\right) \left(1+O \left(r-r_-\right)\right)}\frac{\partial r}{\partial r_{\star}}=2 \kappa _-\leq \left(\frac{1}{r-r_-}-C\right)\frac{\partial r}{\partial r_{\star}}.
$$
Integrating from $(r_{\star})_0$ to $r_\star$ yields
\begin{eqnarray*}
&&\log \left(r(r_{\star},\theta_\star)-r_-\right)
-\log \left(r((r_{\star})_0,\theta_\star)-r_-\right)
+C \left(r(r_{\star},\theta_\star)-r((r_{\star})_0,\theta_\star)\right)\\
&&\qquad\leq 2 \kappa _- (r_{\star}-(r_{\star})_0)\leq\log \left(r(r_{\star},\theta_\star)-r_-\right)
-\log \left(r((r_{\star})_0,\theta_\star)-r_-\right)
-C \left(r(r_{\star},\theta_\star)-r((r_{\star})_0,\theta_\star)\right).
\end{eqnarray*}
These inequalities can be rearranged to
\begin{eqnarray*}
	 &&2 \kappa _- (r_{\star}-(r_{\star})_0)
	 +\log \left(r((r_{\star})_0,\theta_\star)-r_-\right)
	 +C \left(r(r_{\star},\theta_\star)-r((r_{\star})_0,\theta_\star)\right)\\
	 &&\qquad\leq\log \left(r(r_{\star},\theta_\star)-r_-\right)\leq
	 2 \kappa _- (r_{\star}-(r_{\star})_0)
	 +\log \left(r((r_{\star})_0,\theta_\star)-r_-\right)
	 -C \left(r(r_{\star},\theta_\star)-r((r_{\star})_0,\theta_\star)\right).
\end{eqnarray*}
This shows that
there exists $D > 0$ and $(r_\star)_0$ such that 
$$
2 \kappa _- r_{\star}-D
\leq\log \left(r(r_{\star},\theta_\star)-r_-\right)\leq 2 \kappa _- r_{\star}+D
$$
for $r_\star\geq (r_\star)_0$.
Thus, there exist $c,C>0$, such that $r_\star\geq (r_\star)_0$ implies
$$
ce^{2\kappa_-r_\star}\leq r-r_-\leq Ce^{2\kappa_-r_\star}.
$$
Let $C_R$ be a real number. Decreasing $c$ and increasing $C$ if necessary,
one sees that
$$
ce^{2\kappa_-r_\star}\leq r-r_-\leq Ce^{2\kappa_-r_\star}
$$
for $C_R\leq r_\star\leq(r_\star)_0$. Combining the previous two inequalities,
they hold for $r_\star\geq C_R$. Moreover, in this region, 
$$
c(r-r_-)\leq-\Delta_r\leq C(r-r_-),
$$
for other appropriate constants $c,C>0$. Hence, given $C_R\in\bbR$,
there exist $c,C>0$ such that 
\be\label{r-r}
ce^{2\kappa_-r_\star}\leq-\Delta_r\leq Ce^{2\kappa_-r_\star}
\ee
for $r_\star\geq C_R$. As $\Omega^2$ is comparable to $|\Delta_r|$, we conclude that $\Omega^2$ is comparable to $e^{2\kappa_-r_\star}$
in the region $r_\star\geq C_R$.
This proves~\eqref{omega1}.
The proof of~\eqref{omega2} is analogous.
\end{proof}

\begin{rmk}
	The function $\Omega^2$ given in~\eqref{Omega-squared} is obviously smooth on the spheres where 
	$u$ and $v$ are simultaneously constant.
\end{rmk}

\subsubsection{Coordinates at the Cauchy horizon}\label{sub-Cauchy}

Let us recall how one may define coordinates to cover the Cauchy horizon.
We consider a new smooth coordinate $\vch(v)$, with positive derivative, equal to $v$ for $v\leq -1$, satisfying $\vch\to 0$ as $v\to+\infty$, and satisfying
\be\label{v-v}
\mmd v=e^{-2\kappa_-v}\md\vch
\ee
for $v\geq 0$. Moreover, we define
$$
\phich=\phi_{\star}-\left.b^{\phi_{\star}}\right|_{r=r_-}v.
$$
\begin{rmk}
$\vch$ is also a smooth function of $v$, and so the change of coordinates $(v,\phi_{\star})\leftrightarrow
(\vch,\phich)$ is smooth.
\end{rmk}
From~\eqref{b-phi}, we see that
$$
\left.b^{\phi_{\star}}\right|_{r=r_-}=\left.2\frac{\Xi a((r^2+a^2)\Delta_{\theta}-\Delta_r)}{((r^2+a^2)^2\Delta_{\theta}-a^2\sin^2\theta\Delta_r)}\right|_{r=r_-}\ =\ 
\frac{2a\Xi}{r_-^2+a^2}.
$$
For $v\geq 0$, the differentials of $\phi_\star$ and $\phich$ are related by
$$
\mmd\phi_{\star}=\mmd\phich+\left.b^{\phi_{\star}}\right|_{r=r_-}\mmd v
=\mmd\phich+\left.b^{\phi_{\star}}\right|_{r=r_-}
e^{-2\kappa_-v}\md\vch.
$$
For $v\geq 0$, the expression of the metric~\eqref{metric-final} in the new coordinates is 
\begin{eqnarray*}
g&=&-2\Omegasqch \left(\mmd u\otimes \mmd \vch+\md \vch\otimes \mmd u\right)
\nonumber\\
&&+\gamma_{\theta_{\star}\theta_{\star}}
\md\theta_{\star}\otimes \mmd\theta_{\star}
+\gamma_{\theta_{\star}\phi_{\star}}\md\theta_{\star}\otimes(\mmd\phich-\bch\md \vch)
+\gamma_{\theta_{\star}\phi_{\star}}(\mmd\phich-\bch\md \vch)\otimes
\mmd\theta_{\star}\nonumber\\
&&+\gamma_{\phi_{\star}\phi_{\star}}
(\mmd\phich-\bch\md \vch)\otimes(\mmd\phich-\bch\md \vch),
\end{eqnarray*}
with
\begin{eqnarray}
	\Omegasqch&=&\Omega^2e^{-2\kappa_-v},\label{qomega}\\
\bch&=&\left(b^{\phi_{\star}}-\left.b^{\phi_{\star}}\right|_{r=r_-}\right)
e^{-2\kappa_-v}.\nonumber
\end{eqnarray}
Henceforth we will assume that we are working in the region $v\geq 0$,
our formulas will always refer to this region.
To estimate $\bch$ we calculate
$$
b^{\phi_{\star}}-\left.b^{\phi_{\star}}\right|_{r=r_-}=
-2a\Xi \frac{(r_-^2+a^2\cos^2\theta)\Delta_r+(r^2+a^2)(r+r_-)(r-r_-)\Delta_{\theta}}
{(r_-^2+a^2)((r^2+a^2)^2\Delta_{\theta}-a^2\sin^2\theta\Delta_r)}.
$$
Using~\eqref{u_v} and~\eqref{r-r}, we estimate
$$
\left|b^{\phi_{\star}}-\left.b^{\phi_{\star}}\right|_{r=r_-}\right|
\lesssim e^{2\kappa_-(u+v)},
$$
for $u+v\geq C_R$.
Thus, inequalities~\eqref{omega1} implies the following bounds
for $\Omegasqch$ and $\bch$, when $u+v\geq C_R$:
\begin{eqnarray*}
	e^{2\kappa_-u}\lesssim&\Omegasqch&\lesssim e^{2\kappa_-u},\\
&|\bch|&\lesssim e^{2\kappa_-u}.
\end{eqnarray*}

For a general function $f$, we have
$$
f(u,v,\theta_{\star},\phi_\star)=f\left(u,v,\theta_{\star},\phich+\bphir v
\right)=
\tilde{f}
(u,\tilde{v},\theta_{\star},\phich),$$
where $\tilde{f}$ is the function $f$ written in the coordinates
$(u,\tilde{v},\theta_{\star},\phich)$ and $\tilde v=v$. So
\be\label{muda}
\partial_{\tilde{v}}=\partial_v+\bphir\partial_{\phi_\star}.
\ee
We define
\be\label{sigma}
\sigma=e^{-2\kappa_-v}.
\ee
This, \eqref{v-v} and $\partial_{\phi_{\star}}=\partial_{\phich}$ imply that
\begin{eqnarray*}
\partial_{\vch}+\bch\pp&=&e^{-2\kappa_-v}\left(\partial_{\tilde v}+\left(\bphi-\bphir\right)\pp\right)
\\
&=&e^{-2\kappa_-v}\left(
\partial_v+\bphi\partial_{\phi_\star}
\right)\\
&=&\sigma\left(
\partial_v+\bphi\partial_{\phi_\star}
\right).
\end{eqnarray*}

We can write the vector field $\partial_t$ using the coordinates at the Cauchy horizon as
\begin{eqnarray*}
	\partial_t&=&\frac{1}{2}\left(
	\partial_{\tilde{v}}-\bphir\pp
	\right)-\,\frac{1}{2}\partial_u\\
	&=&\frac{e^{2\kappa_-v}}{2}\left(
	\partial_{\vch}+\bch\pp
	\right)-\,\frac{1}{2}\bphi\pp-\,\frac{1}{2}\partial_u.
\end{eqnarray*}
The vector field $\partial_t$ is not null on the Cauchy horizon. A Killing 
vector field which is null on the Cauchy horizon is 
\begin{eqnarray}
Z&=&-\partial_t-\,\frac{1}{2}\bphi|_{r=r_-}\pp\nonumber\\
&=&\frac{1}{2}\partial_u+\frac{1}{2}\left(\bphi-\bphi|_{r=r_-}\right)\pp
-\,\frac{e^{2\kappa_-v}}{2}(\partial_{\vch}+\bch\pp)\nonumber\\
&=&\frac{1}{2}\partial_u+\frac{1}{2}\left(\bphi-\bphi|_{r=r_-}\right)\pp
-\,\frac{\Omega^2}{2}\pvpdois.
\label{Z}
\end{eqnarray}

\subsubsection{Coordinates at the event horizon}\label{sub-event}

We consider a new smooth coordinate $\uh(u)$, with positive derivative, equal to $u$ for $u\geq 1$, satisfying $\uh\to 0$ as $u\to-\infty$, and satisfying
$$
\mmd u=e^{-2\kappa_+u}\md\uh
$$
for $u\leq 0$. Moreover, we define $\vh=v$ and
$$
\phih=\phi_\star-\bphi|_{r=r_+}v.
$$
\begin{rmk}
$\uh$ is also a smooth function of $u$, and the change of coordinates $(v,\phi_{\star})\leftrightarrow
(\vh,\phih)$ is smooth.
\end{rmk}
Note that
$$
\bphi|_{r=r_+}=\frac{2a\Xi}{r_+^2+a^2}.
$$
For $u\leq 0$, we may write the metric as
\begin{eqnarray*}
	g&=&-2\Omegasqh \left(\mmd \uh\otimes \mmd \vh+\md \vh\otimes \mmd \uh\right)
	\nonumber\\
	&&+\gamma_{\theta_{\star}\theta_{\star}}
	\md\theta_{\star}\otimes \mmd\theta_{\star}
	+\gamma_{\theta_{\star}\phi_{\star}}\md\theta_{\star}\otimes(\mmd\phih-\bh\md \vh)
	+\gamma_{\theta_{\star}\phi_{\star}}(\mmd\phih-\bh\md \vh)\otimes
	\mmd\theta_{\star}\nonumber\\
	&&+\gamma_{\phi_{\star}\phi_{\star}}
	(\mmd\phih-\bh\md \vh)\otimes(\mmd\phih-\bh\md \vh),
\end{eqnarray*}
with  
\begin{eqnarray*}
\Omegasqh&=&\Omega^2e^{-2\kappa_+u},\\
\bh&=&\bphi-\bphi|_{r=r_+}.
\end{eqnarray*}
Henceforth we will assume that we are working in the region $u\leq 0$,
our formulas will always refer to this region.
For a general function $f$, we have
$$
f(u,v,\theta_{\star},\phi_\star)=f\left(u(\uh),v,\theta_{\star},\phih+\bphirh v
\right)=
\tilde{f}
(\uh,\vh,\theta_{\star},\phih),$$
where $\tilde{f}$ is the function $f$ written in the coordinates
$(\uh,\vh,\theta_{\star},\phih)$. So, defining 
\be\label{varsigma}
\varsigma=e^{-2\kappa_+u},
\ee
we get
\begin{eqnarray*}
\puh&=&e^{-2\kappa_+u}\partial_u\ =\ \varsigma\partial_u,\\
\pvh&=&\partial_v+\bphirh\partial_{\phi_\star}.
\end{eqnarray*}
This, \eqref{v-v} and $\partial_{\phi_{\star}}=\partial_{\phih}$ imply that
\begin{eqnarray*}
\partial_v+\bphi\partial_{\phi_\star}&=&
\partial_{\vh}+\bh\pp.
\end{eqnarray*}
A Killing vector field which is null on the event horizon is
\be\label{W}
W=\partial_t+\left.\frac{1}{2}\bphi\right|_{r=r_+}\pp=
\frac{1}{2}\pvh-\,\frac{e^{2\kappa_+u}}{2}\puh.
\ee
The value 
\begin{equation}\label{angular}
\left.\frac{1}{2}\bphi\right|_{r=r_+}
=\frac{a\Xi}{r_+^2+a^2}
\end{equation}
is the angular velocity $\Omega_H$ on the event horizon.

\section{The energy of the solutions of the wave equation}\label{aplicacao}

We will use the vector field method to study the energy of solutions
of the wave equation which have compact support on ${\cal H}^+$.
As is well known, the method, used by Morawetz~\cite{Morawetz}, John~\cite{John}, Klainerman~\cite{K,KL},
Dafermos~\cite{DHR,m-luk,DR,m-lec,DS} and Rodnianski~\cite{m-lec,KL}, among many others,
consists in applying the Divergence 
Theorem to some currents obtained by contracting the energy-momentum
tensor $T_{\mu\nu}$ with appropriate vector field multipliers constructed 
specifically according to each region of spacetime. 

We refer to the region close to the Cauchy horizon as the blue-shift 
region, and the region close to the event horizon as the
red-shift region. We call the intermediate region the
no-shift region. A very general construction of red-shift vector fields
on general spacetimes which contain Killing horizons with positive surface
gravity is carried out in the lecture notes~\cite{m-lec}.
Here we perform the computations explicitly in double null coordinates.

The blue-shift vector field~\eqref{Nb} is
constructed using the vector field $Y$ in Lemma~\ref{Y}, and the red-shift vector field~\eqref{nr}
is constructed using the vector field $V$ in Lemma~\ref{VV}.
We go on to calculate the covariant derivative of $Y$, $\nabla^\mu Y^\nu$, and the
scalar current associated to $Y$, $T_{\mu\nu}\nabla^\mu Y^\nu$. We obtain the usual inequalities
for the currents associated to the blue-shift vector field,
and for the currents associated to the red-shift vector field.
We finish with Theorem~\ref{thm}, which is Sbierski's result, for the Reissner--Nordstr\"{o}m and Kerr spacetimes, applied to Kerr--Newman--de Sitter spacetimes.

\subsection{The blue-shift and red-shift vector fields}
\subsubsection{Construction}
Here the blue-shift vector field is defined to be
\be\label{Nb}
N_b=Y+Z,
\ee
where $Y$ is given in
\begin{lem}\label{Y}
	Let $\iota\in\bbR^+$ be given.
	The initial value problem 
	\begin{eqnarray}\label{diff-equation}
	\nabla_YY&=&-\iota(Y+Z),\\
	\label{ic}
	Y|_{{\cal CH}^+}&=&\pvpdois
	\end{eqnarray}
	($Z$ as in~\eqref{Z})
	has a unique time invariant solution, $Y$, defined in a
	neighborhood of the Cauchy horizon, i.e.~defined for 
	$r_\star$ sufficiently large.
\end{lem}
\begin{proof}
\begin{enumerate}[{\bf (a)}]
\item {\bf ($Y$ time invariant.)} The vector field $\pvp$ commutes with $\partial_t$. 
In fact, we have
$$
\left[\partial_t,\frac{1}{\Omegasqch}\pv\right]=
\left[\partial_t,\frac{1}{\Omega^2}\partial_{\tilde{v}}\right]
=\frac{1}{\Omega^2}\left[\frac{1}{2}\partial_v-\,\frac{1}{2}\partial_u,
\left(\partial_v+\bphi|_{r=r_-}\partial_{\phi_\star}\right)\right]
=0
$$
and
$$
\frac{\bch}{\Omegasqch}=\frac{\bphi-\bphi|_{r=r_-}}{\Omega^2}
\ \Rightarrow\ 
\left[
\partial_t,\frac{\bch}{\Omegasqch}\pp
\right]=0.
$$
So, any $Y$ of the form
\begin{eqnarray*}
	Y&=&\f\pu+\g\pvp+\h\pt+\tldej\pp,
\end{eqnarray*}
with
$$
\f=\f(r,\theta),\ \g=\g(r,\theta),\ \h=\h(r,\theta),\ \tldej=\tldej(r,\theta),
$$
commutes with $\partial_t$.
\item {\bf (The differential equation.)}
Expanding the left-hand side of~\eqref{diff-equation}, we get
\begin{eqnarray*}
\nabla_YY&=&\f\pr\f\pu+\f\pr\g\pvp+\f\pr \h\pt+\f\pr\tldej\pp\\
&&+\f^2\nabla_{\pu}\pu+\f\g\nabla_{\pu}\pvp+\f\h\nabla_{\pu}\pt+\f\tldej\nabla_{\pu}\pp\\
&&+\g\frac{\sigma}{\Omegasqch}\pr\f\pu+\g\frac{\sigma}{\Omegasqch}\pr\g\pvp+\g\frac{\sigma}{\Omegasqch}\pr \h\pt+\g\frac{\sigma}{\Omegasqch}\pr\tldej\pp\\
&&+\f\g\nabla_{\pvp}\pu+\g\h\nabla_{\pvp}\pt+\g\tldej\nabla_{\pvp}\pp\\
&&+\h\pt\f\pu+\h\pt\g\pvp+\h\pt\h\pt+\h\pt\tldej\pp\\
&&+\f\h\nabla_{\pt}\pu+\g\h\nabla_{\pt}\pvp+\h^2\nabla_{\pt}\pt+\h\tldej\nabla_{\pt}\pp\\
&&+\f\tldej\nabla_{\pp}\pu+\g\tldej\nabla_{\pp}\pvp+\h\tldej\nabla_{\pp}\pt+\tldej^2\nabla_{\pp}\pp.
\end{eqnarray*}
We used~\eqref{z} to eliminate the term in $\g^2$.
This initial value problem~\eqref{diff-equation}, \eqref{ic}
is equivalent to a system of four equations, for the $\pu$, $\pvp$, $\pt$ and $\pp$ components of each side, for the four unknowns $\f$, $\g$, $\h$ and $\tldej$,
with 
\be\label{inicio}
\f(r_-,\theta)=0,\quad \g(r_-,\theta)=1,\quad
\h(r_-,\theta)=0,\quad \tldej(r_-,\theta)=0.
\ee
 The system reads
\begin{eqnarray*}
	\f\pr\f+\g\frac{\sigma}{\Omegasqch}\pr\f+\h\pt\f&=&\ldots,\\
	\f\pr\g+\g\frac{\sigma}{\Omegasqch}\pr\g+\h\pt\g&=&\ldots,\\
	\f\pr\h+\g\frac{\sigma}{\Omegasqch}\pr\h+\h\pt\h&=&\ldots,\\
	\f\pr\tldej+\g\frac{\sigma}{\Omegasqch}\pr\tldej+\h\pt\tldej&=&\ldots,
\end{eqnarray*}
where the right-hand sides involve the Christoffel symbols of the metric,
$\pr\bch$, $\pt\bch$, and $\f$, $\g$, $\h$ and~$\tldej$, but do not involve any derivatives of these last four functions.
The system may be written as
\begin{eqnarray}
{\cal V}\cdot\f&=&\ldots,\label{f}\\
{\cal V}\cdot\g&=&\ldots,\label{g}\\
{\cal V}\cdot\h&=&\ldots,\label{h}\\
{\cal V}\cdot\,\tldej&=&\ldots,\label{j}
\end{eqnarray}
where
$$
{\cal V}=
\left(\left(\frac{1}{\Delta_r}\frac{\partial r}{\partial r_\star}\right)
\left(\f\Delta_r+\g\left(-\,\frac{\Upsilon}{\rho^2\Delta_\theta}\right)\right)
+\h\frac{\partial r}{\partial\theta_\star}\right)
\frac{\partial}{\partial r}+
\left(\left(\frac{1}{\Delta_r}\frac{\partial\theta}{\partial r_\star}\right)
\left(\f\Delta_r+\g\left(-\,\frac{\Upsilon}{\rho^2\Delta_\theta}\right)\right)
+\h\frac{\partial \theta}{\partial\theta_\star}\right)
\frac{\partial}{\partial\theta}.
$$
We used
$$
\frac{\sigma\Delta_r}{\Omegasqch}=-\,\frac{\Upsilon}{\rho^2\Delta_\theta},
$$
which we obtain from~\eqref{Omega-squared}, \eqref{qomega} and~\eqref{sigma}.
\item {\bf (An auxiliary calculation.)} In the next step we will use
the following identity. We mention that it implies
that $\h$ is not identically equal to zero. 
Using~\eqref{long1a} and~\eqref{long1b}, we have
\begin{eqnarray}
&&\md\theta_\star\left(\nabla_{\pu}\pvp+\nabla_{\pvp}\pu\right)
=2\gamma^{\theta_\star\theta_\star}\frac{\partial_{\theta_\star}(\Omegasqch)}{\Omegasqch}.
\label{long1}
\end{eqnarray}
\item {\bf (The characteristics do not cross the boundary.)}
We now check that
\be\label{zero}
\h(r,0)=0\qquad\mbox{and}\qquad \h\left(r,\frac{\pi}{2}\right)=0.
\ee
This implies 
\be\label{V}
{\cal V}^\theta(r,0)=0\qquad\mbox{and}\qquad {\cal V}^\theta\left(r,\frac{\pi}{2}\right)=0,
\ee
and guarantees that the characteristics
of our differential equations do not leave the region,
$[r_-,r_+]\times\left[0,\frac{\pi}{2}\right]$,
where we want to solve our system. It also guarantees that $Y$ is well
defined when $\theta=0$, notwithstanding $\partial_{\theta_\star}$ not being 
well defined when $\theta=0$. The vanishing of $\h$ at 
$\theta=\frac{\pi}{2}$ can also be seen as a consequence of the
symmetry of our problem under the reflection $\theta\mapsto\pi-\theta$,
which implies that $Y$ should not have any component in the 
$\partial_{\theta_\star}$ direction at the equators of the spheres
where $u$ and $v$ are both constant.
The right hand side of~\eqref{h} consists of a sum of terms which 
we divide into two parts. The first part consists of sum of the eight summands
that have $\h$ as a factor.
The term $$\md\theta_\star(-\iota(Y+Z))=-\iota \md\theta_\star(Y)
=-\iota \tilde{h}$$
is proportional to $\tilde{h}$.
The second part consists of the sum of the
remaining eight summands, which are $\md \theta_\star$ applied to
\begin{eqnarray*}
&&-\f^2\nabla_{\pu}\pu-\f\g\nabla_{\pu}\pvp-2\f\tldej\nabla_{\pu}\pp\\
&&-\f\g\nabla_{\pvp}\pu-2\g\tldej\nabla_{\pvp}\pp-\tldej^2\nabla_{\pp}\pp.
\end{eqnarray*}
Taking into account~\eqref{long1}, the second part is
\begin{eqnarray*}
&&-\f^2\Gamma_{uu}^{\theta_\star}-4\f\g\frac{\gamma^{\theta_\star\theta_\star}}
{\Omegasqch}\partial_{\theta_\star}(\Omegasqch)-2\f\tldej\Gamma_{u\phich}^{\theta_\star}\\
&&-2\g\tldej\frac{1}{\Omegasqch}
\left(\Gamma_{\vch\phich}^{\theta_\star}+\bch\Gamma_{\phich\phich}^{\theta_\star}\right)-\tldej^2\Gamma_{\phich\phich}^{\theta_\star}.
\end{eqnarray*}
At $\theta_\star=0$ and at $\theta_\star=\frac{\pi}{2}$ this sum is zero
because each of the terms is equal to zero. 
Indeed, all terms are a product of a differentiable function by $\sin(2\theta_\star)$. This is easy to check. Let us exemplify this assertion with
one of the less immediate terms to analyze, the one which contains
$$
\Gamma_{u\phich}^{\theta_\star}=
\frac{\gamma^{\theta_\star\theta_\star}}{2}
\partial_{r_\star}\gamma_{\theta_\star\phi_\star}
+\frac{\gamma^{\theta_\star\phi_\star}}{2}
\partial_{r_\star}\gamma_{\phi_\star\phi_\star}
$$
(see~\eqref{byke}).
Since both $\gamma_{\theta_\star\phi_\star}$ and $\gamma^{\theta_\star\phi_\star}$
contain $\sin(2\theta_\star)$, so does this Christoffel symbol.

Let us examine in more detail~\eqref{h} for points 
$(r,0)$ and $\left(r,\frac{\pi}{2}\right)$. 
Since $\partial_{r_\star}\theta$ also contains a factor $\sin(2\theta_\star)$,
over these two segments, $[r_-,r_+)\times\{0\}$ and
$[r_-,r_+)\times\{\frac{\pi}{2}\}$, \eqref{h} reads
$$
{\cal V}^r\partial_r\h=-\h\frac{\partial\theta}{\partial\theta_\star}\frac{\partial\h}
{\partial\theta}+\h\times\mbox{smooth function}=\h\times\mbox{smooth function}.
$$
As 
$$
\h(r_-,0)=0\qquad\mbox{and}\qquad\h\left(r_-,\frac{\pi}{2}\right)=0
$$
and ${\cal V}^r$ is not zero (at least initially at $r_-$, see below),
we conclude that, if there exists a solution to our initial value problem, then it must satisfy~\eqref{zero}. Using the fact that
$\partial_{r_\star}\theta$ contains a factor $\sin(2\theta_\star)$ once again,
we obtain~\eqref{V}.
\item {\bf (Existence and uniqueness of solution.)}
Using~\eqref{Upsilon} and~\eqref{r_r*}, we have
$$
{\cal V}^r(r_-,\theta)=\left(\frac{1}{\Delta_r}\frac{\partial r}{\partial r_\star}\right)
\left(-\,\frac{\Upsilon}{\rho^2\Delta_\theta}\right)=\frac{1}{(r_-^2+a^2)}
\left(-\,\frac{(r_-^2+a^2)^2}{\rho^2}\right)
=-\,\frac{(r_-^2+a^2)}{r_-^2+a^2\cos^2\theta}.
$$
This shows that the segment $\{r_-\}\times\left[0,\frac{\pi}{2}\right]$ is
noncharacteristic for our system of four first-order quasilinear partial
differential equations.

We observe that when the Christoffel symbols 
$\Gamma_{v\theta_\star}^{\phi_\star}$ and $\Gamma_{\theta_\star\phi_\star}^{\phi_\star}$, which 
blow up at $\theta_\star=0$ like $\frac{1}{\sin\theta_\star}$
(see Corollary~\ref{mountain}), appear in the system above, then they appear multiplied by $\h$, which has to vanish to first order at $(r,0)$. So that the summands where these Christoffel 
appear are continuous functions.

By a standard existence and uniqueness theorem for non characteristic first order
quasilinear partial differential equations, we know that our initial value 
problem has a solution for $(r,\theta)\in[r_-,r_-+\delta]\times[0,\pi]$,
for some positive $\delta$. Recall that constructing the solution involves
 solving the system of ordinary differential equations
$$
\left\{
\begin{array}{rcl}
\dot{r}&=&{\cal V}^r\\
\dot{\theta}&=&{\cal V}^{\theta}\\
\dot{\f}&=&\mbox{right-hand side of~\eqref{f}}\\
\dot{\g}&=&\mbox{right-hand side of~\eqref{g}}\\
\dot{\h}&=&\mbox{right-hand side of~\eqref{h}}\\
\dot{\tldej}&=&\mbox{right-hand side of~\eqref{j}}
\end{array}
\right.
\qquad
\mbox{with}\qquad
\left\{
\begin{array}{rcl}
r(0)&=&r_-\\
\theta(0)&=&\theta_0\\
\f(0)&=&0\\
\g(0)&=&1\\
h(0)&=&0\\
\tldej(0)&=&0
\end{array}
\right..
$$
We know from Remark~\ref{swing} that, for $r_\star$ sufficiently large,
$r$ is close to $r_-$. So we have the existence of a solution for large $r_\star$.
\end{enumerate}
\end{proof}

Here the red-shift vector field is defined to be
\be\label{nr}
N_r=V+W,
\ee
where $V$ is given in
\begin{lem}\label{VV}
		Let $\iota\in\bbR^+$ be given.
	The initial value problem 
	\begin{eqnarray*}\label{diff-eqn}
	\nabla_VV=&&-\iota(V+W),\\
	\label{ictwo}
	V|_{{\cal H}^+}&=&\frac{\puh}{\Omegasqh}
	\end{eqnarray*}
	($W$ as in~\eqref{W})
	has a unique time invariant solution, $V$, defined in a
	neighborhood of the event horizon, i.e.~defined for 
	$r_\star$ sufficiently negative (i.e.\ for $r_\star<-C$ with $C$ sufficiently large).
\end{lem}
\begin{proof}
	Choose $V$ of the form
	\begin{eqnarray*}
		V&=&\f\puph+\g\pvph+\h\pt+\tldej\pp,
	\end{eqnarray*}
	with
	$$
	\f=\f(r,\theta),\ \g=\g(r,\theta),\ \h=\h(r,\theta),\ \tldej=\tldej(r,\theta),
	$$
	$$
	\f(r_+,\theta)=1,\ \g(r_+,\theta)=0,\ \h(r_+,\theta)=0,\ \tldej(r_+,\theta)=0.
	$$
\end{proof}

\subsubsection{Covariant derivative}

One could consider working in the frame $(Z,Y,\partial_{\theta_\star},\partial_{\phi_\star})$ but this is not
a good choice because the energy-momentum tensor does not have a simple expression in this frame.
So, instead, we define
$$
\overline{T}=\frac{1}{2}\partial_u\qquad\mbox{and}\qquad\overline{Y}=\pvp=
\frac{\partial_v+\bphi\partial_{\phi_\star}}{\Omega^2},
$$
and work in the frame
$$
(X_{\tinyt},X_{\tinyy},X_{\theta_\star},X_{\phi_\star}):=
(\overline{T},\overline{Y},\partial_{\theta_\star},\partial_{\phi_\star}).
$$ 
We see that
$$
Y=2\f\overline{T}+\g\overline{Y}+\h\partial_{\theta_\star}+\tldej\pp.
$$
The dual frame is
\begin{eqnarray*}
	(\omega_{\tinyt},\omega_{\tinyy},\omega_{\theta_\star},\omega_{\phi_\star})&=&
	\left(
	2\md u,\Omegasqch\md\vch,\md\theta_\star,\md\phich-\bch\md\vch
	\right)\\
	&=&
	\left(
	2\md u,\Omega^2\md v,\md\theta_\star,\md\phi_\star-\bphi\md v
	\right).
\end{eqnarray*}

\paragraph{Calculation of the covariant derivative of $Y$.} 

The covariant derivative of $Y$ is
\begin{eqnarray*}
	\nabla Y&=&\omega_{\tinyt}^\sharp\otimes\nabla_{\overline{T}}Y
	+\omega_{\tinyy}^\sharp\otimes\nabla_{\overline{Y}}Y
	+\omega_{\theta_\star}^\sharp\otimes\nabla_{\partial_{\theta_\star}}Y
	+\omega_{\phi_\star}^\sharp\otimes\nabla_{\partial_{\phi_\star}}Y.
\end{eqnarray*}
One readily checks that the metric dual basis to  $(\omega_{\tinyt},\omega_{\tinyy},\omega_{\theta_\star},\omega_{\phi_\star})$ is
\begin{eqnarray*}
	(\omega_{\tinyt}^\sharp,\omega_{\tinyy}^\sharp,
	\omega_{\theta_\star}^\sharp,\omega_{\phi_\star}^\sharp)&=&
	(-\overline{Y},-\overline{T},
	\gamma^{\theta_\star\theta_\star}\partial_{\theta_\star}
	+\gamma^{\theta_\star\phi_\star}\partial_{\phi_\star},
	\gamma^{\theta_\star\phi_\star}\partial_{\theta_\star}
	+\gamma^{\phi_\star\phi_\star}\partial_{\phi_\star})\\
	&=&(-\overline{Y},-\overline{T},
	\partial^{\theta_\star},
	\partial^{\phi_\star}).
\end{eqnarray*}	
Using the fact that $2\overline{T}=\partial_u=\partial_{r_\star}=
(\partial_{r_\star}r)\partial_r+(\partial_{r_\star}\theta)\partial_\theta$,
\eqref{inicio} and~\eqref{long1b}, we obtain
\begin{eqnarray}
	\nabla_{\overline{T}}Y&=&
	\nabla_{\overline{T}}\left(2\f\overline{T}+\g\overline{Y}+\h\partial_{\theta_\star}+\tldej\pp\right)\nonumber\\
	&=&\nabla_{\overline{T}}\overline{Y}+\overline{O}(r-r_-)\nonumber\\
	&=&-\kappa_-\overline{Y}+a^{\theta_\star}\partial_{\theta_\star}+
	a^{\phi_\star}\partial_{\phi_\star}
	+\tilde{O}(r-r_-),\label{Emi1}
\end{eqnarray}	
where 
$$
\tilde{O}(r-r_-)=O(r-r_-)\overline{T}+O(r-r_-)\overline{Y}+O(r-r_-)\partial_{\theta_\star}
+O(r-r_-)\partial_{\phi_\star}.
$$
The values of $a^{\theta_\star}$ and $a^{\phi_\star}$ can be read off 
from~\eqref{long1b}. 
Using~\eqref{diff-equation} and the formulas in Appendix~\ref{Chris}, we get
\begin{eqnarray}
\nabla_{\overline{Y}}Y&=&\nabla_YY-\nabla_{Y-\overline{Y}}Y\nonumber\\
&=&\nabla_YY-\nabla_{\left(2\f\overline{T}+(\g-1)\overline{Y}+\h\partial_{\theta_\star}+\tldej\pp\right)}
\left(2\f\overline{T}+\g\overline{Y}+\h\partial_{\theta_\star}+\tldej\pp\right)\nonumber\\
&=&-\iota(Y+Z)+\overline{O}(r-r_-)\nonumber\\
&=&-\iota(\overline{Y}+\overline{T})+\overline{O}(r-r_-),\label{Emi2}
\end{eqnarray}
where 
$$
\overline{O}(r-r_-)=O(r-r_-)\overline{T}+O(r-r_-)\overline{Y}+O(r-r_-)\partial_{\theta_\star}
+O(r-r_-)\frac{1}{\sin\theta}\partial_{\phi_\star}.
$$
The factor $\frac{1}{\sin\theta}$ in front of $\partial_{\phi_\star}$
in the last summand arises from $\Gamma_{v\theta_\star}^{\phi_\star}$
and $\Gamma_{\theta_\star\phi_\star}^{\phi_\star}$.
Using the fact that 
$\partial_{\theta_\star}=
(\partial_{\theta_\star}r)\partial_r
+(\partial_{\theta_\star}\theta)\partial_\theta$
and~\eqref{h1},
 we obtain
\begin{eqnarray}
	\nabla_{\partial_{\theta_\star}}Y&=&
	\nabla_{\partial_{\theta_\star}}\left(2\f\overline{T}+\g\overline{Y}+\h\partial_{\theta_\star}+\tldej\pp\right)\nonumber\\
	&=&\nabla_{\partial_{\theta_\star}}\overline{Y}+\overline{O}(r-r_-)\nonumber\\
	&=&-a_{\theta_\star}\overline{Y}
	+h^{\ \ \theta_\star}_{\theta_\star}\partial_{\theta_\star}+
	h^{\ \ \phi_\star}_{\theta_\star}\partial_{\phi_\star}
	+\overline{O}(r-r_-).\label{Emi3}
\end{eqnarray}	
 The values of $h^{\ \ \theta_\star}_{\theta_\star}$ and $h^{\ \ \phi_\star}_{\theta_\star}$ can be read off 
 from~\eqref{h1}.
 Finally, we have 
\begin{eqnarray}
	\nabla_{\partial_{\phi_\star}}Y&=&
	\nabla_{\partial_{\phi_\star}}\left(2\f\overline{T}+\g\overline{Y}+\h\partial_{\theta_\star}+\tldej\pp\right)\nonumber\\
	&=&\nabla_{\partial_{\phi_\star}}\overline{Y}+\overline{O}(r-r_-)\nonumber\\
	&=&-a_{\phi_\star}\overline{Y}
	+h^{\ \ \theta_\star}_{\phi_\star}\partial_{\theta_\star}+
	h^{\ \ \phi_\star}_{\phi_\star}\partial_{\phi_\star}
	+\hat{O}(r-r_-).\label{Emi4}
\end{eqnarray}	
The values of $h^{\ \ \theta_\star}_{\phi_\star}$ and $h^{\ \ \phi_\star}_{\phi_\star}$ can be read off 
from~\eqref{h2}.
Here
$$
\hat{O}(r-r_-)=O(r-r_-)\sin\theta\overline{T}+O(r-r_-)\sin\theta\overline{Y}+O(r-r_-)
\sin\theta\partial_{\theta_\star}
+O(r-r_-)\partial_{\phi_\star}.
$$
Note that although $\Gamma_{\theta_\star\phi_\star}^{\phi_\star}$ does contain 
the factor $\frac{1}{\sin\theta}$, in the last calculation this Christoffel symbol
appears multiplied by $\h$ which vanishes to first order at $\theta_\star=0$ and
$\theta_\star=\pi$. So, the last equality is a consequence of Lemma~\ref{behave}
and of
the fact that $\f$, $\g$, $\h$ and $\tldej$ do not depend on $\phi_\star$. Indeed, the components of
$\nabla_{\pp}X_\dagger$ in $\overline{T}$, $\overline{Y}$ and $\partial_{\theta_\star}$
all contain the factor $\sin\theta$, for 
$\dagger\,\in\,\{\smallt,\smally,\theta_\star,\phi_\star\}$.
The expressions~\eqref{Emi1}--\eqref{Emi4} above correspond to~\cite[(19)--(22)]{m-lec}.
Combining the previous results, we can write the covariant derivative of $Y$
as
\begin{eqnarray*}
	\nabla Y&=&-\overline{Y}\otimes\left(
	-\kappa_-\overline{Y}+a^{\theta_\star}\partial_{\theta_\star}+
	a^{\phi_\star}\partial_{\phi_\star}
	\right)+\overline{Y}\otimes\tilde{O}(r-r_-)
	\\
	&&-\overline{T}\otimes\left(
	-\iota(\overline{Y}+\overline{T})
	\right)+\overline{T}\otimes\overline{O}(r-r_-)\\
	&&+\partial^{\theta_\star}\otimes
	\left(
	-a_{\theta_\star}\overline{Y}
	+h_{\theta_\star\theta_\star}\partial^{\theta_\star}+
	h_{\theta_\star\phi_\star}\partial^{\phi_\star}
	\right)+\partial^{\theta_\star}\otimes\overline{O}(r-r_-)\\
	&&+\partial^{\phi_\star}\otimes
	\left(
	-a_{\phi_\star}\overline{Y}
	+h_{\phi_\star\theta_\star}\partial^{\theta_\star}+
	h_{\phi_\star\phi_\star}\partial^{\phi_\star}
	\right)+\partial^{\phi_\star}\otimes\hat{O}(r-r_-)\\
	&=&\kappa_-\overline{Y}\otimes\overline{Y}
	-\overline{Y}\otimes(a^{\theta_\star}\partial_{\theta_\star})
	-\overline{Y}\otimes(a^{\phi_\star}\partial_{\phi_\star})\\
	&&+\iota\overline{T}\otimes\overline{Y}+\iota\overline{T}\otimes\overline{T}\\
	&&-(a^{\theta_\star}\partial_{\theta_\star})\otimes\overline{Y}
	-(a^{\phi_\star}\partial_{\phi_\star})\otimes\overline{Y}\\
	&&+h^{\theta_\star\theta_\star}\partial_{\theta_\star}\otimes\partial_{\theta_\star}
	+h^{\theta_\star\phi_\star}\partial_{\theta_\star}\otimes\partial_{\phi_\star}
	+h^{\phi_\star\theta_\star}\partial_{\phi_\star}\otimes\partial_{\theta_\star}
	+h^{\phi_\star\phi_\star}\partial_{\phi_\star}\otimes\partial_{\phi_\star}\\
	&&+\overline{\overline{O}}(r-r_-).
\end{eqnarray*}
We may write the error term as
\begin{eqnarray*}
\overline{\overline{O}}(r-r_-)&=&\sum_{\dagger,\ddagger\,
\in\,\{\tinyt,\tinyy,\theta_\star\}}
O(r-r_-)\,\omega^\sharp_\dagger\otimes X_\ddagger\\
&&+O(r-r_-)\,\overline{Y}\otimes\frac{1}{\sin\theta}\partial_{\phi_\star}
+O(r-r_-)\,\frac{1}{\sin\theta}\partial_{\phi_\star}\otimes\overline{Y}\\
&&+O(r-r_-)\,\overline{T}\otimes\frac{1}{\sin\theta}\partial_{\phi_\star}
+O(r-r_-)\,\frac{1}{\sin\theta}\partial_{\phi_\star}\otimes\overline{T}\\
&&+O(r-r_-)\,\partial_{\theta_\star}\otimes\frac{1}{\sin\theta}\partial_{\phi_\star}
+O(r-r_-)\,\frac{1}{\sin\theta}\partial_{\phi_\star}\otimes\partial_{\theta_\star}\\
&&+O(r-r_-)\frac{1}{\sin^2\theta}\partial_{\phi_\star}\otimes\partial_{\phi_\star}
\end{eqnarray*}
because $\omega_{\phi_\star}^\sharp=\partial^{\phi_\star}$ behaves like
$\gamma^{\phi_\star\phi_\star}\partial_{\phi_\star}$, which in turn behaves like
$\frac{1}{\sin^2\theta}\partial_{\phi_\star}$.

\subsubsection{Currents}

\paragraph{The energy momentum tensor of a massless scalar field and the vector current.}

The energy momentum tensor is given by
$$
T_{\mu\nu}=\partial_\mu\psi\partial_\nu\psi-\,\frac{1}{2}g_{\mu\nu}\partial^\alpha\psi\,\partial_\alpha\psi,
$$
and one readily checks that
\begin{eqnarray*}
	\partial^\alpha\psi\,\partial_\alpha\psi&=&-\,\frac{1}{\Omega^2}(\partial_u\psi)(
	\partial_v\psi+\bphi\partial_{\phi_{\star}}\psi)
	+|\nabb\psi|_\gamma^2\\
	&=&-\,\frac{1}{\Omega^2}(\partial_u\psi)(
	\partial_v\psi+\bphi\partial_{\phi_{\star}}\psi)
	+\gamma^{\theta_\star\theta_\star}(\partial_{\theta_\star}\psi)^2
	+2\gamma^{\theta_\star\phi_\star}(\partial_{\theta_\star}\psi)(\partial_{\phi_\star}\psi)
	+\gamma^{\phi_\star\phi_\star}(\pp\psi)^2.
\end{eqnarray*}
The energy-momentum tensor is written as				
\begin{eqnarray*}
	T=\!\!\!\!\!\!\!\!\!\!&&\frac{1}{4}
	\left(\partial_{u}\psi\right)^2\omt\ot\omt
	+
	\frac{1}{2}|\nabb\psi|_\gamma^2\,(\omt\ot\omy
	+\omy\ot\omt)+
	\frac{1}{\Omegafch}\left(\partial_{\vch}\psi+\bch\pp\psi\right)^2\omy\ot\omy
	\\
	&&+\frac{1}{2}
	\left(\partial_{u}\psi\right)(\partial_{\theta_\star}\psi)
	(\omt\ot\omtt+\omtt\ot\omt)
	+\frac{1}{2}\left(\partial_{u}\psi\right)(\pp\psi)
	(\omt\ot\omp+\omp\ot\omt)
	\\
	&&+\frac{1}{\Omegasqch}
	\left(\partial_{\vch}\psi+\bch\pp\psi\right)
	(\partial_{\theta_\star}\psi)(\omy\ot\omtt+\omtt\ot\omy)\\
	&&+\frac{1}{\Omegasqch}
	\left(\partial_{\vch}\psi+\bch\pp\psi\right)
	(\pp\psi)(\omy\ot\omp+\omp\ot\omy)\\
	&&+\left(
	(\partial_{\theta_\star}\psi)^2-\,\frac{\gamma_{\theta_{\star},\theta_{\star}}}
	{2}\left(-\,\frac{1}{\Omegasqch}(\partial_u\psi)
	(\partial_{\vch}\psi+\bch\partial_{\phi_{\star}}\psi)
	+|\nabb\psi|_\gamma^2\right)\right)\omtt\ot\omtt\\
	&&+\left(
	(\partial_{\theta_\star}\psi)(\pp\psi)-\,\frac{\gamma_{\theta_{\star},\phi_{\star}}}
	{2}\left(-\,\frac{1}{\Omegasqch}(\partial_u\psi)
	(\partial_{\vch}\psi+\bch\partial_{\phi_{\star}}\psi)
	+|\nabb\psi|_\gamma^2\right)\right)(\omtt\ot\omp+\omp\ot\omtt)\\
	&&+\left(
	(\pp\psi)^2-\,\frac{\gamma_{\phi_{\star},\phi_{\star}}}
	{2}\left(-\,\frac{1}{\Omegasqch}(\partial_u\psi)
	(\partial_{\vch}\psi+\bch\partial_{\phi_{\star}}\psi)
	+|\nabb\psi|_\gamma^2\right)\right)\omp\ot\omp.
\end{eqnarray*}
The vector currents associated to the blue-shift and red-shift vector fields
are
$$
J^{N_b}_\mu=T_{\mu\nu}N_b^\nu\qquad\mbox{and}\qquad
J^{N_r}_\mu=T_{\mu\nu}N_r^\nu,
$$
respectively.

\paragraph{The scalar current associated to the blue-shift vector field.}

The scalar current associated to $N_b$ is $K^{N_b}=T_{\mu\nu}\nabla^\mu N_b^\nu$.
Since $Z$ is a Killing vector field this is equal to $K^Y$, which is
\begin{eqnarray*}
	K^Y&=&T_{\mu\nu}\nabla^\mu Y^\nu\\
	&=&\kappa_-T(\overline{Y},\overline{Y})
	-2T(\overline{Y},a^{\theta_\star}\partial_{\theta_\star})
	-2T(\overline{Y},a^{\phi_\star}\partial_{\phi_\star})\\
	&&+\iota T(\overline{T},\overline{Y})+\iota T(\overline{T},\overline{T})\\
	&&+h^{\theta_\star\theta_\star}T(\partial_{\theta_\star},\partial_{\theta_\star})
	+h^{\theta_\star\phi_\star}T(\partial_{\theta_\star},\partial_{\phi_\star})
	+h^{\phi_\star\theta_\star}T(\partial_{\phi_\star},\partial_{\theta_\star})
	+h^{\phi_\star\phi_\star}T(\partial_{\phi_\star},\partial_{\phi_\star})\\
	&&+O(r-r_-)(T(\overline{T},\overline{T})+T(\overline{T},\overline{Y})+T(\overline{Y},\overline{Y})).
\end{eqnarray*}
We estimate $K^Y$. Suppose we are given $\delta\in\bigl(0,\frac 13\bigr)$. Let $\overline{c}=2(-\kappa_-)\delta$. 
Choose $\iota=c+\overline{c}+1$, where $c=c_\delta$ is the constant below
(which is independent of $Y$).
There exists $r_0>r_-$ such that 
$r\in(r_-,r_0)$ implies that
\begin{eqnarray}
K^Y
&\geq&\kappa_-T(\overline{Y},\overline{Y})
+\iota T(\overline{T},\overline{Y})+\iota T(\overline{T},\overline{T})\nonumber\\
&&-cT(\overline{T},\overline{T})-cT(\overline{T},\overline{Y})-(-\kappa_-)(\delta/2) T(\overline{Y},\overline{Y})\nonumber\\
&&+O(r-r_-)(T(\overline{T},\overline{T})+T(\overline{T},\overline{Y})+T(\overline{Y},\overline{Y}))\nonumber\\
&\geq&\kappa_-(1+\delta)T(\overline{Y},\overline{Y})+\overline{c}(T(\overline{T},\overline{T})+
T(\overline{T},\overline{Y})).\label{K}
\end{eqnarray}
We have used the fact that
$$
\gamma^{\phi_\star\phi_\star}(\partial_{\phi_\star}\psi)^2\leq|\nabb\psi|^2_\gamma.
$$

\paragraph{Inequalities relating the currents.}
The blue-shift vector field satisfies
\begin{lem}
Let $0<\delta<\frac{1}{3}$. If $(r_\star)_0$ is chosen sufficiently large (see\/ {\rm Remark~\ref{swing}}), then
\be\label{blue}
\int\limits_{\stackrel{u_0<u<u_1}{\mbox{\tiny $v\!\!>\!\!(r_\star)_0\!\!-\!\!u$}}} K^{N_b}\dV\geq 2\kappa_-(1+3\delta)\int_{[u_0,u_1]}\left(
\int_{v>(r_\star)_0-u} J^{N_b}_\mu\nCu^\mu\VCu
\right)\md u.
\ee
\end{lem}
\begin{proof} Using~\eqref{volume} and~\eqref{K}, we obtain
$$
\int\limits_{\stackrel{u_0<u<u_1}{\mbox{\tiny $v\!\!>\!\!(r_\star)_0\!\!-\!\!u$}}}
K^N\dV\geq\int\limits_{\stackrel{u_0<u<u_1}{\mbox{\tiny $v\!\!>\!\!(r_\star)_0\!\!-\!\!u$}}}
2\kappa_-(1+\delta)T(\overline{Y},\overline{Y})+
2\overline{c}(
T(\overline{T},\overline{T})+
T(\overline{T},\overline{Y}))\Omegasqch\frac{L\sin\theta}{\Xi}
\md u\md\vch\md\theta_\star\md\phi_\star.
$$
Using~\eqref{ncu}, we see that
\begin{eqnarray*}
	T\left(N_b,\frac{\nCu}{\Omega^2}\right)&=&T(N_b,\overline{Y})\\
	&=&T(\overline{T}+\overline{Y},\overline{Y})+O(r-r_-)
	\left(T(\overline{T},\overline{Y})+T(\overline{Y},\overline{Y})
	+T(\partial_{\theta_\star},\overline{Y})+T(\partial_{\phi_\star},\overline{Y})
	\right).
\end{eqnarray*}
For each $0<\delta<\frac{1}{3}$, there exists $r_0>r_-$ such that  
$r\in(r_-,r_0)$ implies that
$$
\int\limits_{v>(r_\star)_0-u}
T\left(N_b,\frac{\nCu}{\Omega^2}\right)\Omega^2\VCu\geq
\int\limits_{v>(r_\star)_0-u}
\left(
(1-\delta)\left(T(\overline{T},\overline{Y})+ T(\overline{Y},\overline{Y})\right)-\delta T(\overline{T},\overline{T})
\right)
\Omegasqch\frac{L\sin\theta}{\Xi}
\md\vch\md\theta_\star\md\phi_\star.
$$
Multiplying both sides by $2\kappa_-(1+\delta)/(1-\delta)$, we get
\begin{eqnarray*}
&&2\kappa_-(1+\delta)\int\limits_{v>(r_\star)_0-u}
\left(
\left(T(\overline{T},\overline{Y})+ T(\overline{Y},\overline{Y})\right)-\,
\frac{\delta}{1-\delta} T(\overline{T},\overline{T})\right)\Omegasqch\frac{L\sin\theta}{\Xi}
\md\vch\md\theta_\star\md\phi_\star\\
&&\qquad\geq 2\kappa_-\,\frac{1+\delta}{1-\delta}\int\limits_{v>(r_\star)_0-u} T\left(N_b,\frac{\nCu}{\Omega^2}\right)\Omega^2\VCu.
\end{eqnarray*}
This implies~\eqref{blue} because $\frac{1+\delta}{1-\delta}<1+3\delta$
and $-2\kappa_-\delta\frac{1+\delta}{1-\delta}<-4\kappa_-\delta=2\overline{c}$.
\end{proof}
Similarly to~\eqref{blue}, the red-shift vector field satisfies
\begin{lem}
Let $\delta>0$. For $(r_\star)_0$ sufficiently negative (see\/ {\rm Remark~\ref{swing}}), we have
\be\label{Nr}
\int_{\stackrel{v_0<v<v_1}{\mbox{\tiny $u\!\!<\!\!(r_\star)_0\!\!-\!\!v$}}} K^{N_r}\dV\geq 2\kappa_+(1-\delta)\int_{[v_0,v_1]}\left(
\int_{u<(r_\star)_0-v} J^{N_r}_\mu\nCv^\mu\VCv
\right)\md v.
\ee
\end{lem}
\begin{proof} Work in the frame $(\overline{V},\underline{T},\partial_{\theta_\star},\partial_{\phi_\star})$,
	with 
$$
\overline{V}=\puph,
\qquad
\underline{T}=\frac{1}{2}\left(\pvh+\bh\pp\right).
$$
	
\end{proof}

\subsection{Energy estimates}
We are interested in solutions of the wave equation which are regular
up to, and including, ${\cal H}^+$, and which have
compact support on ${\cal H}^+$, i.e.\ we are interested in functions belonging to the space
$$
	{\cal F}:=\left\{\psi\in C^\infty({\cal M}\cup{\cal H}^+):\Box_g\psi=0\
	\mbox{and there exists $v_0\in\bbR$ such that}\
	 \psi|_{{\cal H}^+\cap\{v\geq v_0\}}=0\right\}.
	\label{Fs}
$$
The following theorem, established by Sbierski in his thesis~\cite{Sbierski} for
Reissner--Nordstr\"{o}m and Kerr black holes, applies to 
Kerr--Newman--de Sitter black holes.
\begin{thm}\label{thm}
	Let $2\kappa_+>-\kappa_-$ and $\psi\in{\cal F}$. Then, for any $u_0$, $v_0\in\bbR$,
	we have
	\begin{eqnarray*}
		&&\int\limits_{\stackrel{{\cal CH}^+}{{\mbox{\tiny $(-\infty,u_0]$}}}}
		J^{N_b}_\mu
		n^\mu_{{\cal CH}^+}\dV_{{\cal CH}^+}+
		\int\limits_{\stackrel{{\cal C}_{u_0}}
			{{\mbox{\tiny $[v_0,\infty)$}}}}
		J^{N_b}_\mu
		n^\mu_{{\cal C}_{u_0}}\dV_{{{\cal C}_{u_0}}}
		<+\infty.
	\end{eqnarray*}
\end{thm}
\begin{proof}
	We sketch the proof and refer to~\cite{Sbierski} for further details.
	
\noindent {\bf (a) Estimates in the red-shift region.}
For $\kappa<\kappa_+$, define
$
\underline{N}_r=e^{2\kappa v}N_r
$. 
Choose the $\delta$ in~\eqref{Nr} such that $\kappa<\kappa_+(1-\delta)$.
Since $K^{\underline{N}_r}\geq 0$, the Divergence Theorem implies that we have
\begin{eqnarray*}
	\int\limits_{\stackrel{\Sigma_{(r_\star)_0}}{{\mbox{\tiny $[0,\uh]\cap[v_0,v_1]$}}}} J^{\underline{N}_r}_\mu
	n^\mu_{\Sigma_{(r_\star)_0}}\dV_{\Sigma_{(r_\star)_0}}
		&\leq&
		\int\limits_{\stackrel{\underline{\cal C}_{v_0}}{{\mbox{\tiny $[0,\uh]$}}}}
		J^{\underline{N}_r}_\mu
		n^\mu_{\underline{\cal C}_{v_0}}\dV_{\underline{\cal C}_{v_0}}
=e^{2\kappa v_0}\int\limits_{\stackrel{\underline{\cal C}_{v_0}}
			{{\mbox{\tiny $[0,\uh]$}}}}
		J^{N_r}_\mu 
		n^\mu_{\underline{\cal C}_{v_0}}\dV_{\underline{\cal C}_{v_0}}.
\end{eqnarray*}
The left-hand side can be bounded below by
$$
\int\limits_{\stackrel{\Sigma_{(r_\star)_0}}{{\mbox{\tiny $[0,\uh]\cap[v_0,v_1]$}}}}
J^{\underline{N}_r}_\mu
n^\mu_{\Sigma_{(r_\star)_0}}\dV_{\Sigma_{(r_\star)_0}}\geq e^{2\kappa (r_\star)_0}e^{-2\kappa u(\uh)}
\int\limits_{\stackrel{\Sigma_{(r_\star)_0}}{{\mbox{\tiny $[0,\uh]\cap[v_0,v_1]$}}}}
J^{N_r}_\mu
n^\mu_{\Sigma_{(r_\star)_0}}\dV_{\Sigma_{(r_\star)_0}},
$$
where $u(\uh)$ denotes the $u$ corresponding to $\uh$.
As $\psi$ is a regular function on ${\cal H}^+$, we have
\begin{eqnarray*}
&&\int\limits_{\stackrel{\underline{\cal C}_{v_0}}{{\mbox{\tiny $[0,\uh]$}}}}J^{N_r}_\mu n^\mu_{\underline{\cal C}_{v_0}}\dV_{\underline{\cal C}_{v_0}}\\
&&\qquad\qquad\leq C
\int_0^{\uh}\left(\left(\frac{(\partial_{\tuh}\psi)^2}{\Omegafh}+\frac{1}{2}|\nabb\psi|^2_\gamma\right)\Omegasqh\frac{L\sin\theta}{\Xi}
\right)(\,\cdot\,,v_0,\,\cdot\,,\,\cdot\,)\,\md\tuh\md\theta_\star\md\phi_\star\\
&&\qquad\qquad\leq C\uh=Ce^{2\kappa_+u(\uh)}.
\end{eqnarray*}
We conclude that
$$
\int\limits_{\stackrel{\Sigma_{(r_\star)_0}}{{\mbox{\tiny $[0,\uh]\cap[v_0,v_1]$}}}}
J^{N_r}_\mu
n^\mu_{\Sigma_{(r_\star)_0}}\dV_{\Sigma_{(r_\star)_0}}\leq Ce^{4\kappa u(\uh)}.
$$

\noindent {\bf (b) Estimates in the no-shift region.}
We recall~\cite[Lemma~4.5.6]{Sbierski}:
	Given $(r_\star)_1>(r_\star)_0$ and a smooth future directed timelike
	time invariant vector field $N$, there exists a constant $C>0$ such that
	\be\label{no-shift-region}
	\int\limits_{\stackrel{\Sigma_{(r_\star)_1}}{{\mbox{\tiny $[0,\uh]$}}}}J^N_\mu n^\mu_{\Sigma_{(r_\star)_1}}
	\dV_{\Sigma_{(r_\star)_1}}
			\leq C
	\int\limits_{\stackrel{\Sigma_{(r_\star)_0}}{{\mbox{\tiny $[0,\uh]$}}}}
	J^N_\mu n^\mu_{\Sigma_{(r_\star)_0}}
	\dV_{\Sigma_{(r_\star)_0}}
	\ee
	holds for all solutions of the wave equation.
	
\noindent {\bf (c) Estimates in the blue-shift region.}	
We also recall~\cite[Lemma~4.5]{CF}:
	Let $f:[t_0,\infty[\to\bbR$ and assume that for some $\alpha_1, C>0$, and for all
	$t\geq t_0$,
	$$
	\int_t^\infty f(s)\,ds\leq Ce^{-\alpha_1 t}.
	$$
	Then, for all $0<\alpha_2<\alpha_1$ and $t\geq t_0$, we have
	\be\label{lemmaExp}
	\int_t^\infty e^{\alpha_2s} f(s)\,ds\leq Ce^{-(\alpha_1-\alpha_2)t}.
	\ee
By assumption $2\kappa_+>-\kappa_-$. If we choose $\kappa<\kappa_+$ 
sufficiently close to $\kappa_+$ and
and $\overline{\kappa}<\kappa_-$ sufficiently close to $\kappa_-$, then
$2\kappa>-\overline{\kappa}$.
Using~\eqref{no-shift-region}, \eqref{lemmaExp} and the fact that
$$\int\limits_{\stackrel{\Sigma_{(r_\star)_1}}{{\mbox{\tiny $(-\infty,u]$}}}}
J^{N_r}_\mu n^\mu_{\Sigma_{(r_\star)_1}}
\dV_{\Sigma_{(r_\star)_1}}
\quad\mbox{and}\quad
\int\limits_{\stackrel{\Sigma_{(r_\star)_1}}{{\mbox{\tiny $(-\infty,u]$}}}}
J^{N_b}_\mu n^\mu_{\Sigma_{(r_\star)_1}}
\dV_{\Sigma_{(r_\star)_1}}
$$
are comparable, we conclude that
$$
\int\limits_{\stackrel{\Sigma_{(r_\star)_1}}{{\mbox{\tiny $(-\infty,u_0]\cap\left(-\infty,\left(\vch\right)_0\right]$}}}}
	e^{-2\overline{\kappa} u}J^{N_b}_\mu
n^\mu_{\Sigma_{(r_\star)_1}}\dV_{\Sigma_{(r_\star)_1}}<+\infty.
$$
Define
$
\underline{N}_b=e^{-2\overline{\kappa}u}N_b
$. 
Choose the $\delta$ in~\eqref{blue} such that $\overline{\kappa}<\kappa_-(1+3\delta)$.
Since $K^{\underline{N}_b}\geq 0$, the Divergence Theorem implies that
\begin{eqnarray*}
	&&\int\limits_{\stackrel{\Sigma_{(r_\star)_2}}{{\mbox{\tiny $(-\infty,u_0]\cap\left(-\infty,\left(\vch\right)_0\right]$}}}}
	e^{-2\overline{\kappa} u}J^{N_b}_\mu
	n^\mu_{\Sigma_{(r_\star)_2}}\dV_{\Sigma_{(r_\star)_2}}+
	\int\limits_{\stackrel{{\cal C}_{u_0}}
		{{\mbox{\tiny $(r_\star)_1\!\!-\!u_0\!<\!v\left(\vch\right)\!<\!(r_\star)_2\!\!-\!u_0$}}}}
	e^{-2\overline{\kappa} u_0}
	J^{N_b}_\mu
	n^\mu_{{\cal C}_{u_0}}\dV_{{{\cal C}_{u_0}}}
	\\
	&&\qquad\qquad\qquad\leq
	\int\limits_{\stackrel{\Sigma_{(r_\star)_1}}{{\mbox{\tiny $(-\infty,u_0]\cap\left(-\infty,\left(\vch\right)_0\right]$}}}}
		e^{-2\overline{\kappa} u}J^{N_b}_\mu 
	n^\mu_{\Sigma_{(r_\star)_1}}\dV_{\Sigma_{(r_\star)_1}}.
\end{eqnarray*}
This finishes the proof.
\end{proof}
This theorem shows that for KNdS black holes, with surface gravities satisfying $2\kappa_+>-\kappa_-$,
the energy of a wave 
with compact support on ${\cal H}^+$ has finite energy
along a null hypersurface intersecting the Cauchy horizon.
This suggests that there exists a set of 
 of black hole parameters for which compactly supported perturbations do not 
 lead to  mass inflation and to instability of the Cauchy horizon. It is one more step towards understanding  the conjectured instability 
 of black hole interiors in the context of KNdS spacetimes.

\appendix

\section{The Christoffel symbols}\label{Chris}

\subsection{$(u,v,\theta_\star,\phi_\star)$ coordinates}

Using the notation
$$
\bthetab:=\bphi\gamma_{\theta_\star\phi_\star},\qquad
\bphib:=\bphi\gamma_{\phi_\star\phi_\star},\qquad
\normb:=(\bphi)^2\gamma_{\phi_\star\phi_\star}=\bphi\bphib,
$$
we write the Christoffel symbols of the metric $g$. If we use coordinates
$(u,v,\theta_\star,\phi_\star)$, the values of both $\sigma$ and $\varsigma$
are equal to~$1$. We have that
\begin{eqnarray*}
	\Gamma_{uu}^u&=&\varsigma\frac{\partial_{r_\star}(\Omega^2)}{\Omega^2},\\
	\Gamma_{uu}^v&=& \Gamma_{uu}^{\theta_\star}\ =\ \Gamma_{uu}^{\phi_\star}\ =\ 0;
\end{eqnarray*}

\begin{eqnarray*}
	\Gamma_{uv}^u&=&-\varsigma\frac{\bphib}{4\Omega^2}
	\partial_{r_\star}\bphi,\\
	\Gamma_{uv}^v&=&0,\\
	\Gamma_{uv}^{\theta_\star}
	&=&-\varsigma\frac{\gamma^{\theta_\star\theta_\star}\bphi}{2}
	\partial_{r_\star}\gamma_{\theta_\star\phi_\star}
	-\varsigma\frac{\gamma^{\theta_\star\phi_\star}\bphi}{2}
	\partial_{r_\star}\gamma_{\phi_\star\phi_\star}+
	\gamma^{\theta_\star\theta_\star}\partial_{\theta_\star}(\Omega^2),\\
	\Gamma_{uv}^{\phi_\star}
	&=&-\varsigma\frac{\gamma^{\theta_\star\phi_\star}\bphi}{2}
	\partial_{r_\star}\gamma_{\theta_\star\phi_\star}
	-\varsigma\frac{\gamma^{\phi_\star\phi_\star}\bphi}{2}
	\partial_{r_\star}\gamma_{\phi_\star\phi_\star}+
	\gamma^{\theta_\star\phi_\star}\partial_{\theta_\star}(\Omega^2)
	-\,\frac{\varsigma}{2}\partial_{r_\star}\bphi;
\end{eqnarray*}	

\begin{eqnarray*}
	\Gamma_{vv}^u&=&\sigma\frac{(\bphi)^2}{4\Omega^2}\partial_{r_\star}
	\gamma_{\phi_\star\phi_\star},\\
	\Gamma_{vv}^v&=&\varsigma\frac{1}{4\Omega^2}
	\partial_{r_\star}(\normb)
	+\sigma\frac{1}{\Omega^2}\partial_{r_\star}(\Omega^2),\\
	\Gamma_{vv}^{\theta_\star}&=&-\sigma\gamma^{\theta_\star\theta_\star}\bphi
	\partial_{r_\star}\gamma_{\theta_\star\phi_\star}
	-\sigma\gamma^{\theta_\star\phi_\star}\bphi
	\partial_{r_\star}\gamma_{\phi_\star\phi_\star}
	-\,\frac{\gamma^{\theta_\star\theta_\star}}{2}
	\partial_{\theta_\star}(\normb),\\
	\Gamma_{vv}^{\phi_\star}&=&-\sigma\gamma^{\theta_\star\phi_\star}\bphi
	\partial_{r_\star}\gamma_{\theta_\star\phi_\star}
	-\sigma\gamma^{\phi_\star\phi_\star}\bphi
	\partial_{r_\star}\gamma_{\phi_\star\phi_\star}
	-\,\frac{\gamma^{\theta_\star\phi_\star}}{2}
	\partial_{\theta_\star}(\normb)\\
	&&+\sigma\frac{\bphi}{\Omega^2}\partial_{r_\star}(\Omega^2)+\varsigma
	\frac{\bphi}{4\Omega^2}
	\partial_{r_\star}(\normb)-\sigma\partial_{r_\star}\bphi;
\end{eqnarray*}

\begin{eqnarray*}
	\Gamma_{u\theta_\star}^u&=&\varsigma	\frac{\gamma_{\theta_\star\phi_\star}}{4\Omega^2}
	\partial_{r_\star}\bphi
	+\frac{1}{2\Omega^2}
	\partial_{\theta_\star}(\Omega^2),\\
	\Gamma_{u\theta_\star}^v&=&0,\\
	\Gamma_{u\theta_\star}^{\theta_\star}&=&
	\varsigma\frac{\gamma^{\theta_\star\theta_\star}}{2}
	\partial_{r_\star}\gamma_{\theta_\star\theta_\star}
	+\varsigma\frac{\gamma^{\theta_\star\phi_\star}}{2}
	\partial_{r_\star}\gamma_{\theta_\star\phi_\star},\\
	\Gamma_{u\theta_\star}^{\phi_\star}&=&
	\varsigma\frac{\gamma^{\theta_\star\phi_\star}}{2}
	\partial_{r_\star}\gamma_{\theta_\star\theta_\star}
	+\varsigma\frac{\gamma^{\phi_\star\phi_\star}}{2}
	\partial_{r_\star}\gamma_{\theta_\star\phi_\star};
\end{eqnarray*}

\begin{eqnarray}
	\Gamma_{u\phi_\star}^u&=&\varsigma\frac{\gamma_{\phi_\star\phi_\star}}{4\Omega^2}
	\partial_{r_\star}\bphi,\nonumber\\
	\Gamma_{u\phi_\star}^v&=&0,\nonumber\\
	\Gamma_{u\phi_\star}^{\theta_\star}&=&
	\varsigma\frac{\gamma^{\theta_\star\theta_\star}}{2}
	\partial_{r_\star}\gamma_{\theta_\star\phi_\star}
	+\varsigma\frac{\gamma^{\theta_\star\phi_\star}}{2}
	\partial_{r_\star}\gamma_{\phi_\star\phi_\star},\label{byke}\\
	\Gamma_{u\phi_\star}^{\phi_\star}&=&
	\varsigma\frac{\gamma^{\theta_\star\phi_\star}}{2}
	\partial_{r_\star}\gamma_{\theta_\star\phi_\star}
	+\varsigma\frac{\gamma^{\phi_\star\phi_\star}}{2}
	\partial_{r_\star}\gamma_{\phi_\star\phi_\star};\nonumber
\end{eqnarray}	

\begin{eqnarray}
\Gamma_{v\theta_\star}^u&=&-\sigma\frac{\bphi}{4\Omega^2}
\partial_{r_\star}\gamma_{\theta_\star\phi_\star}
-\,\frac{\bphib}{4\Omega^2}
\partial_{\theta_\star}\bphi,\nonumber\\
\Gamma_{v\theta_\star}^v&=&-\varsigma\frac{1}{4\Omega^2}
\partial_{r_\star}\bthetab
+\frac{1}{2\Omega^2}
\partial_{\theta_\star}(\Omega^2),\nonumber\\
\Gamma_{v\theta_\star}^{\theta_\star}&=&\sigma
\frac{\gamma^{\theta_\star\theta_\star}}{2}
\partial_{r_\star}\gamma_{\theta_\star\theta_\star}
+\sigma\frac{\gamma^{\theta_\star\phi_\star}}{2}
\partial_{r_\star}\gamma_{\theta_\star\phi_\star}
-\,\frac{\gamma^{\theta_\star\phi_\star}}{2}
\partial_{\theta_\star}\bphib,\nonumber\\
\Gamma_{v\theta_\star}^{\phi_\star}&=&
-\varsigma\frac{\bphi}{4\Omega^2}\partial_{r_\star}\bthetab
-\,
\frac{\gamma^{\phi_\star\phi_\star}}{2}
\partial_{\theta_\star}\bphib\label{blow-up-1}\\
&&+\sigma\frac{\gamma^{\theta_\star\phi_\star}}{2}
\partial_{r_\star}\gamma_{\theta_\star\theta_\star}
+\sigma\frac{\gamma^{\phi_\star\phi_\star}}{2}
\partial_{r_\star}\gamma_{\theta_\star\phi_\star}
+\frac{\bphi}{2\Omega^2}\partial_{\theta_\star}(\Omega^2);\nonumber
\end{eqnarray}

\begin{eqnarray*}
	\Gamma_{v\phi_\star}^u&=&-\sigma\frac{\bphi}{4\Omega^2}
	\partial_{r_\star}\gamma_{\phi_\star\phi_\star},\\
	\Gamma_{v\phi_\star}^v&=&-\varsigma\frac{1}{4\Omega^2}
	\partial_{r_\star}\bphib,\\
	\Gamma_{v\phi_\star}^{\theta_\star}&=&\sigma
	\frac{\gamma^{\theta_\star\theta_\star}}{2}
	\partial_{r_\star}\gamma_{\theta_\star\phi_\star}
	+\sigma\frac{\gamma^{\theta_\star\phi_\star}}{2}
	\partial_{r_\star}\gamma_{\phi_\star\phi_\star}
	+\frac{\gamma^{\theta_\star\theta_\star}}{2}\partial_{\theta_\star}\bphib,\\
	\Gamma_{v\phi_\star}^{\phi_\star}&=&-\varsigma
	\frac{\bphi}{4\Omega^2}\partial_{r_\star}\bphib
	+\frac{\gamma^{\theta_\star\phi_\star}}{2}
	\partial_{\theta_\star}\bphib\\
	&&+\sigma\frac{\gamma^{\theta_\star\phi_\star}}{2}
	\partial_{r_\star}\gamma_{\theta_\star\phi_\star}
	+\sigma\frac{\gamma^{\phi_\star\phi_\star}}{2}
	\partial_{r_\star}\gamma_{\phi_\star\phi_\star};
\end{eqnarray*}

\begin{eqnarray*}
	\Gamma_{\theta_\star\theta_\star}^u&=&\frac{\gamma_{\theta_\star\phi_\star}}{2\Omega^2}
	\partial_{\theta_\star}\bphi
	+\sigma\frac{1}{4\Omega^2}
	\partial_{r_\star}\gamma_{\theta_\star\theta\star},\\
	\Gamma_{\theta_\star\theta_\star}^v&=&\varsigma\frac{1}{4\Omega^2}
	\partial_{r_\star}\gamma_{\theta_\star\theta_\star},\\
	\Gamma_{\theta_\star\theta_\star}^{\theta_\star}&=&
	\frac{\gamma^{\theta_\star\theta_\star}}{2}
	\partial_{\theta_\star}\gamma_{\theta_\star\theta_\star}
	+\gamma^{\theta_\star\phi_\star}
	\partial_{\theta_\star}\gamma_{\theta_\star\phi_\star},\\
	\Gamma_{\theta_\star\theta_\star}^{\phi_\star}&=&
	\varsigma\frac{\bphi}{4\Omega^2}\partial_{r_\star}\gamma_{\theta_\star\theta_\star}
	+
	\frac{\gamma^{\theta_\star\phi_\star}}{2}
	\partial_{\theta_\star}\gamma_{\theta_\star\theta_\star}
	+\gamma^{\phi_\star\phi_\star}
	\partial_{\theta_\star}\gamma_{\theta_\star\phi_\star};
\end{eqnarray*}

\begin{eqnarray}
\Gamma_{\theta_\star\phi_\star}^u&=&\sigma
\frac{1}{4\Omega^2}
\partial_{r_\star}\gamma_{\theta_\star\phi_\star}
+\frac{\gamma_{\phi_\star\phi_\star}}{4\Omega^2}
\partial_{\theta_\star}\bphi,\nonumber\\
\Gamma_{\theta_\star\phi_\star}^v&=&\varsigma\frac{1}{4\Omega^2}
\partial_{r_\star}\gamma_{\theta_\star\phi_\star},\nonumber\\
\Gamma_{\theta_\star\phi_\star}^{\theta_\star}&=&
\frac{\gamma^{\theta_\star\phi_\star}}{2}
\partial_{\theta_\star}\gamma_{\phi_\star\phi_\star},\nonumber\\
\Gamma_{\theta_\star\phi_\star}^{\phi_\star}&=&\varsigma
\frac{\bphi}{4\Omega^2}\partial_{r_\star}\gamma_{\theta_\star\phi_\star}
+
\frac{\gamma^{\phi_\star\phi_\star}}{2}
\partial_{\theta_\star}\gamma_{\phi_\star\phi_\star};\label{blow-up-2}
\end{eqnarray}

\begin{eqnarray*}
	\Gamma_{\phi_\star\phi_\star}^u&=&\sigma
	\frac{1}{4\Omega^2}
	\partial_{r_\star}\gamma_{\phi_\star\phi_\star},\\
	\Gamma_{\phi_\star\phi_\star}^v&=&\varsigma\frac{1}{4\Omega^2}
	\partial_{r_\star}\gamma_{\phi_\star\phi_\star},\\
	\Gamma_{\phi_\star\phi_\star}^{\theta_\star}&=&
	-\,\frac{\gamma^{\theta_\star\theta_\star}}{2}
	\partial_{\theta_\star}\gamma_{\phi_\star\phi_\star},\\
	\Gamma_{\phi_\star\phi_\star}^{\phi_\star}&=&\varsigma
	\frac{\bphi}{4\Omega^2}\partial_{r_\star}\gamma_{\phi_\star\phi_\star}
	-\,
	\frac{\gamma^{\theta_\star\phi_\star}}{2}
	\partial_{\theta_\star}\gamma_{\phi_\star\phi_\star}.
\end{eqnarray*}

Recall that $\gamma^{\phi_\star\phi_\star}$ behaves like
$\frac{1}{\sin^2\theta_\star}$.
In the Christoffel symbols above, this metric coefficient appears multiplied 
by
$$
\partial_{r_\star}\bphib,\ \ \ \partial_{\theta_\star}\bphib,\ \ \ 
\partial_{r_\star}\gamma_{\theta_\star\phi_\star},\ \ \ \partial_{\theta_\star}\gamma_{\theta_\star\phi_\star},\ \ \ 
\partial_{r_\star}\gamma_{\phi_\star\phi_\star},\ \ \ \partial_{\theta_\star}\gamma_{\phi_\star\phi_\star}.
$$
\begin{lem}\label{behave}
	The derivatives\/
	$
	\partial_{r_\star}\bphib$,
	$\partial_{r_\star}\gamma_{\theta_\star\phi_\star}$, 
	$\partial_{\theta_\star}\gamma_{\theta_\star\phi_\star}$,
	$\partial_{r_\star}\gamma_{\phi_\star\phi_\star}$ 
	behave like\/ $\sin^2\theta_\star$, and the derivatives
	$
	\partial_{\theta_\star}\bphib$ and\/
	$\partial_{\theta_\star}\gamma_{\phi_\star\phi_\star}$
	behave like\/ $\sin(2\theta_\star)$.
\end{lem}
Lemma~\ref{behave} implies 
\begin{cor}\label{mountain}
	The Christoffel symbols of the metric $g$ are all bounded,
	except for
	$$
	\Gamma_{v\theta_\star}^{\phi_\star}\qquad\mbox{and}\qquad
	\Gamma_{\theta_\star\phi_\star}^{\phi_\star},
	$$
	which blow up like $\frac{1}{\sin\theta_\star}$.
\end{cor}
In the formulas above,
the terms that blow up appear in~\eqref{blow-up-1} and~\eqref{blow-up-2}.

\begin{altproof}{\,{\rm Lemma~\ref{behave}}} Using~\eqref{gamma-ff},
	$\gamma_{\phi_\star\phi_\star}$, and thus $\bphib$, is of the form
	${\cal A}\sin^2\theta$ for a regular function ${\cal A}$. We have
	$$
	\partial_{r_\star}({\cal A}\sin^2\theta)=
	\sin^2\theta\,\partial_{r_\star}{\cal A}+{\cal A}\sin(2\theta)\frac{\partial\theta}{\partial r_\star}.
	$$
	Using~\eqref{t_r*}, we see that $\frac{\partial\theta}{\partial r_\star}$
	contains a factor $\sin(2\theta_\star)$. Moreover, 
	$$
	\frac{1}{\sin(2\theta_\star)}
	\partial_{\theta_\star}({\cal A}\sin^2\theta)=
	\sin^2\theta\left(\frac{1}{\sin(2\theta_\star)}
	\partial_{\theta_\star}{\cal A}\right)+{\cal A}\frac{\sin(2\theta)}{\sin(2\theta_\star)}\frac{\partial\theta}{\partial \theta_\star}.
	$$
	We recall that, according to~\eqref{deriv_bounds}, $\frac{\partial\theta}{\partial \theta_\star}$ is bounded.
	Thus, recalling~\eqref{bound}, $\partial_{r_\star}({\cal A}\sin^2\theta)$
	behaves like $\sin^2\theta_\star$ and $\partial_{\theta_\star}({\cal A}\sin^2\theta)$ behaves like $\sin(2\theta_\star)$.
	On the other hand, according to Lemma~\ref{smooth-functions},
	$\gamma_{\theta_\star\phi_\star}={\cal A}\sin^2\theta_\star\sin(2\theta_\star)$,
	for another smooth function ${\cal A}$.
	Therefore, $\partial_{\theta_\star}\gamma_{\theta_\star\phi_\star}$
	behaves like $\sin^2\theta_\star$.
\end{altproof}

Note that the vector fields $\frac{\partial_u}{\Omega^2}$ and
$\pvpnormal$ are geodesic:
\begin{eqnarray}
\nabla_{\frac{\partial_u}{\Omega^2}}\frac{\partial_u}{\Omega^2}&=&0,\nonumber\\
\nabla_{\pvpnormal}\pvpnormal&=&0.\label{z}
\end{eqnarray}
Since we work in the frame 
\be\label{frame}
\left(
\frac{\partial_u}{2},\pvpnormaldois,\partial_{\theta_\star},\partial_{\phi_\star}
\right)=
\left(
\frac{\partial_u}{2},\pvpdois,\partial_{\theta_\star},\partial_{\phi_\star}
\right)
,
\ee
it is also convenient to have the following covariant derivatives:
\begin{eqnarray}\label{long1a}
\nabla_{\pvpnormal}\partial_u&=&\gamma^{\theta_\star\theta_\star}
\frac{\partial_{\theta_\star}(\Omega^2)}{\Omega^2}\partial_{\theta_\star}
	+\gamma^{\theta_\star\phi_\star}
	\frac{\partial_{\theta_\star}(\Omega^2)}{\Omega^2}\partial_{\phi_\star}
		-\,\frac{\partial_{r_\star}\bphi}{2\Omega^2}\pp,\\
\label{long1b}
\nabla_{\partial_u}\pvpnormal&=&\gamma^{\theta_\star\theta_\star}
\frac{\partial_{\theta_\star}(\Omega^2)}{\Omega^2}\partial_{\theta_\star}
+\left(\gamma^{\theta_\star\phi_\star}
\frac{\partial_{\theta_\star}(\Omega^2)}{\Omega^2}
+\frac{\partial_{r_\star}\bphi}{2\Omega^2}\right)\pp
-\,\frac{\partial_{r_\star}\Omega^2}{\Omega^2}\pvpnormaldois;
\end{eqnarray}
\begin{eqnarray}
	\nabla_{\pvpnormal}\partial_{\theta_\star}
	&=&\left(
	\sigma\frac{\gamma^{\theta_\star\theta_\star}}{2\Omega^2}\partial_{r_\star}
	\gamma_{\theta_\star\theta_\star}
	+\sigma\frac{\gamma^{\theta_\star\phi_\star}}{2\Omega^2}\partial_{r_\star}
	\gamma_{\theta_\star\phi_\star}
	-\,\frac{\gamma^{\theta_\star\phi_\star}\gamma_{\phi_\star\phi_\star}}{2\Omega^2}
	\partial_{\theta_\star}\bphi
	\right)\partial_{\theta_\star}\nonumber\\
	&&+\left(
	\sigma\frac{\gamma^{\theta_\star\phi_\star}}{2\Omega^2}\partial_{r_\star}
	\gamma_{\theta_\star\theta_\star}
	+\sigma\frac{\gamma^{\phi_\star\phi_\star}}{2\Omega^2}\partial_{r_\star}
	\gamma_{\theta_\star\phi_\star}
	-\,\frac{\gamma^{\phi_\star\phi_\star}\gamma_{\phi_\star\phi_\star}}{2\Omega^2}
	\partial_{\theta_\star}\bphi
	\right)\partial_{\phi_\star}\nonumber\\
	&&+\left(
	-\,\frac{\gamma_{\theta_\star\phi_\star}}{4\Omega^2}\partial_{r_\star}\bphi
	+\frac{\partial_{\theta_\star}\Omega^2}{2\Omega^2}
	\right)\pvpnormaldois,\nonumber\\
	\nabla_{\partial_{\theta_\star}}\pvpnormal
	&=&\left(
	\sigma\frac{\gamma^{\theta_\star\theta_\star}}{2\Omega^2}\partial_{r_\star}
	\gamma_{\theta_\star\theta_\star}
	+\sigma\frac{\gamma^{\theta_\star\phi_\star}}{2\Omega^2}\partial_{r_\star}
	\gamma_{\theta_\star\phi_\star}
	-\,\frac{\gamma^{\theta_\star\phi_\star}\gamma_{\phi_\star\phi_\star}}{2\Omega^2}
	\partial_{\theta_\star}\bphi
	\right)\partial_{\theta_\star}\nonumber\\
	&&+\left(
	\sigma\frac{\gamma^{\theta_\star\phi_\star}}{2\Omega^2}\partial_{r_\star}
	\gamma_{\theta_\star\theta_\star}
	+\sigma\frac{\gamma^{\phi_\star\phi_\star}}{2\Omega^2}\partial_{r_\star}
	\gamma_{\theta_\star\phi_\star}
	+\frac{\gamma^{\theta_\star\phi_\star}\gamma_{\theta_\star\phi_\star}}{2\Omega^2}
	\partial_{\theta_\star}\bphi+\frac{1}{2\Omega^2}\partial_{\theta_\star}\bphi
	\right)\partial_{\phi_\star}\nonumber\\
	&&+\left(
	-\,\frac{\gamma_{\theta_\star\phi_\star}}{4\Omega^2}\partial_{r_\star}\bphi
	-\,\frac{\partial_{\theta_\star}\Omega^2}{2\Omega^2}
	\right)\pvpnormaldois;\label{h1}
\end{eqnarray}
\begin{eqnarray}
	\nabla_{\pvpnormal}\partial_{\phi_\star}
	&=&\nabla_{\partial_{\phi_\star}}\pvpnormal
	\nonumber\\
	&=&\left(
	\sigma\frac{\gamma^{\theta_\star\theta_\star}}{2\Omega^2}\partial_{r_\star}
	\gamma_{\theta_\star\phi_\star}
	+\sigma\frac{\gamma^{\theta_\star\phi_\star}}{2\Omega^2}\partial_{r_\star}
	\gamma_{\phi_\star\phi_\star}
	+\frac{\gamma^{\theta_\star\theta_\star}\gamma_{\phi_\star\phi_\star}}{2\Omega^2}
	\partial_{\theta_\star}\bphi
	\right)\partial_{\theta_\star}\nonumber\\
	&&+\left(
	\sigma\frac{\gamma^{\theta_\star\phi_\star}}{2\Omega^2}\partial_{r_\star}
	\gamma_{\theta_\star\phi_\star}
	+\sigma\frac{\gamma^{\phi_\star\phi_\star}}{2\Omega^2}\partial_{r_\star}
	\gamma_{\phi_\star\phi_\star}
	+\frac{\gamma^{\theta_\star\phi_\star}\gamma_{\phi_\star\phi_\star}}{2\Omega^2}
	\partial_{\theta_\star}\bphi
	\right)\partial_{\phi_\star}\nonumber\\
	&&-\left(
	\frac{\gamma_{\phi_\star\phi_\star}}{4\Omega^2}\partial_{r_\star}\bphi
	\right)\pvpnormaldois.\label{h2}
\end{eqnarray}

\subsection{$(u,\vch,\theta_\star,\phich)$ coordinates}

If we use coordinates $(u,\vch,\theta_\star,\phich)$ then, in the formulas above,
the value of $\varsigma$ continues to be~$1$, but
one has to replace $\Omega^2$ by $\Omegasqch$, $\bphi$ by $\bch$, and
the value of $\sigma$ has to be as in~\eqref{sigma}.
Indeed, the expressions $\sigma\partial_{r_\star}$ arise when taking
derivatives with respect to $\vch$. Now, using~\eqref{v-v} and~\eqref{muda}, we have
$$
\partial_{\vch}=e^{-2\kappa_-v}\partial_{\tilde{v}}=
e^{-2\kappa_-v}\left(\partial_v+\left.\bphi\right|_{r=r_-}\partial_{\phi_\star}\right)\qquad\mbox{and}\qquad \partial_v=\partial_{r_\star}+\partial_t.
$$
But the coefficients of the metric do not depend either on $\phi_\star$ 
or on $t$. So, when computing the Christoffel symbols, and
differentiating functions that depend exclusively on $r$ and $\theta$,
the derivative
$\partial_{\vch}$ may be replaced by 
$$\partial_{\vch}=\sigma\partial_{r_\star}=e^{-2\kappa_-v}\partial_{r_\star}=e^{-2\kappa_-v}\left(
\frac{\partial r}{\partial r_\star}\partial_r+\frac{\partial \theta}{\partial r_\star}\partial_\theta
\right)=
\left(e^{-2\kappa_-v}\Delta_r\right)\left(\left(\frac{1}{\Delta_r}
\frac{\partial r}{\partial r_\star}\right)\partial_r+\left(\frac{1}{\Delta_r}\frac{\partial \theta}{\partial r_\star}\right)\partial_\theta
\right)
.$$
The expressions for the innermost parenthesis are obtained from~\eqref{r_r*}
and~\eqref{t_r*}. 
Note that they are well defined on the Cauchy horizon,
notwithstanding the coordinate $r_\star$ not being defined there.
The factor $e^{-2\kappa_-v}\Delta_r$ is bounded above and
below according to~\eqref{r-r}. So, given $C_R\in\bbR$, there exist constants
$c,C>0$ such that
$$
ce^{2\kappa_-u}\leq -e^{-2\kappa_-v}\Delta_r\leq Ce^{2\kappa_-u},
$$
for $u+v\geq C_R$. 
Special care has to be taken when differentiating $\bch$ and $\Omegasqch$
with respect to $\vch$
because these functions also depend on $v$.
This occurs when calculating $\Gamma_{\vch\vch}^*$:
\begin{eqnarray*}
	\Gamma_{\vch\vch}^u&=&\sigma\frac{(\bch)^2}{4\Omegasqch}\partial_{r_\star}
	\gamma_{\phi_\star\phi_\star},\\
	\Gamma_{\vch\vch}^{\vch}
	&=&\sigma\frac{\gamma_{\phi_{\star},\phi_{\star}}\bch}{2\Omegasqch}
	\partial_{r_\star}\bphi
	+\frac{(\bch)^2}{4\Omegasqch}\partial_{r_\star}\gamma_{\phi_{\star},\phi_{\star}}
	+\sigma\left(\frac{\partial_{r_\star}(\Omega^2)}{\Omega^2}-2\kappa_-\right),\\
	\Gamma_{\vch\vch}^{\theta_\star}&=&-\sigma\gamma^{\theta_\star\theta_\star}\bch
	\partial_{r_\star}\gamma_{\theta_\star\phi_\star}
	-\sigma\gamma^{\theta_\star\phi_\star}\bch
	\partial_{r_\star}\gamma_{\phi_\star\phi_\star}
	-\,\frac{\gamma^{\theta_\star\theta_\star}}{2}
	\partial_{\theta_\star}(\normbch),\\
	\Gamma_{\vch\vch}^{\phi_\star}&=&-\sigma\gamma^{\theta_\star\phi_\star}\bch
	\partial_{r_\star}\gamma_{\theta_\star\phi_\star}
	-\sigma\gamma^{\phi_\star\phi_\star}\bch
	\partial_{r_\star}\gamma_{\phi_\star\phi_\star}
	-\,\frac{\gamma^{\theta_\star\phi_\star}}{2}
	\partial_{\theta_\star}(\normbch)\\
	&&+\sigma\left(\frac{\partial_{r_\star}(\Omega^2)}{\Omega^2}-
	2\kappa_-
	\right)\bch
	+\sigma\frac{\gamma_{\phi_{\star},\phi_{\star}}(\bch)^2}{2\Omegasqch}
	\partial_{r_\star}\bphi
	+\frac{(\bch)^3}{4\Omegasqch}\partial_{r_\star}\gamma_{\phi_{\star},\phi_{\star}}
	\\
	&&+2\kappa_-\sigma\bch-\sigma^2\partial_{r_\star}\bphi.
\end{eqnarray*}
Notice that all terms are bounded except the last two summands of 
$\Gamma_{\vch\vch}^{\phi_\star}$ that add up to
$$
2\kappa_-\sigma\bch-\sigma^2\partial_{r_\star}\bphi=
-\,\frac{\partial\bch}{\partial\vch}.
$$
We remark that no terms that blow up at the Cauchy horizon appear
when calculating $\nabla_{X_\dagger}X_\ddagger$, for $X_\dagger$
and $X_\ddagger$ elements of the frame~\eqref{frame}.
For example, we have that
$$
\Gamma_{\vch\vch}^u=\frac{(\bch)^2}{4\Omegasqch}(e^{-2\kappa_-v}\Delta_r)
\left(\left(\frac{1}{\Delta_r}
\frac{\partial r}{\partial r_\star}\right)\partial_r\gamma_{\phi_\star\phi_\star}+\left(\frac{1}{\Delta_r}\frac{\partial \theta}{\partial r_\star}\right)\partial_\theta\gamma_{\phi_\star\phi_\star}
\right).
$$

\begin{ex}\label{surface-gravity}{\rm
	As a simple example, we calculate directly the surface gravity
	of the Cauchy horizon. Recall that the surface gravity $\kappa_-$ 
	is given by
	$$
	\nabla_ZZ=\kappa_-Z,
	$$
	where $Z$ is the Killing vector field given by~\eqref{Z}. 
	On the Cauchy horizon
	$\nabla_{\partial_u}\pp=\nabla_{\pp}\partial_u=0$ apply,
	because $\Gamma_{u\phich}^u=\Gamma_{u\phich}^{\theta_\star}=
	\Gamma_{u\phich}^{\phich}=0$ on the Cauchy horizon, and $\Gamma_{u\phich}^v$ is identically zero. Obviously, 
	$$
	\partial_u\left(\bphi-\bphi|_{r=r_-}\right)=0\qquad\mbox{and}\qquad
	\pp\left(\bphi-\bphi|_{r=r_-}\right)=0
	$$
	hold
	because the vector fields $\partial_u$ and $\pp$ are tangent to
	the Cauchy horizon. Moreover, we know that
	$$
	\partial_{\vch}e^{2\kappa_-v}=2\kappa_-e^{2\kappa_-v}\partial_{\vch}v=2\kappa_-.
	$$
	So, we conclude that 
	\begin{eqnarray*}
		\nabla_ZZ&=&\frac{1}{4}\nabla_{\partial_u}\partial_u+\frac{1}{4}
		\partial_u\left(\bphi-\bphi|_{r=r_-}\right)\pp
		+\frac{1}{4}\left(\bphi-\bphi|_{r=r_-}\right)\nabla_{\partial_u}\pp\\
		&&-\,
		\frac{e^{2\kappa_-v}}{4}\nabla_{\partial_u}(\partial_{\vch}+\bch\pp)
		+\frac{1}{2}(\bphi-\bphi|_{r=r_-})\nabla_{\pp}Z
		-\,\frac{e^{2\kappa_-v}}{2}
		\nabla_{\partial_{\vch}+\bch\pp}Z\\
		&=&\frac{1}{4}\nabla_{\partial_u}\partial_u\qquad\mbox{at}\ {\cal CH}^+\\
		&=&\frac{1}{4}\Gamma_{uu}^u\partial_u
		\ = \
		\frac{1}{4}\frac{\partial_{r_\star}\Omegasqch}{\Omegasqch}\,\partial_u
		\ =\ \frac{1}{4}\frac{\partial_{r_\star}\Omega^2}{\Omega^2}\,\partial_u
		\ =\ \frac{1}{4}\frac{\partial_r\Delta_r}{\Delta_r}\frac{\partial r}{\partial r_\star}\,\partial_u\\
		&=&\frac{1}{4}\frac{\partial_r\Delta_r}{\Delta_r}\frac{\Delta_r}{(r^2+a^2)}\,\partial_u
		\ =\ 
		\frac{1}{2}\frac{(\partial_r\Delta_r)_{r=r_-}}{(r_-^2+a^2)}\,\frac{\partial_u}{2}\ =\ \frac{1}{2}\frac{(\partial_r\Delta_r)_{r=r_-}}{(r_-^2+a^2)} Z.
	\end{eqnarray*}
	We used~\eqref{qomega}, then~\eqref{Omega-squared} and finally~\eqref{r_r*}.
	This shows that
	\be\label{surface}
	\kappa_-=\frac{\partial_{r_\star}\Omega^2}{2\Omega^2}=\frac{(\partial_r\Delta_r)_{r=r_-}}{2(r_-^2+a^2)},
	\ee
	which is equality~\eqref{kappa-}.}
\end{ex}

\subsection{$(\uh,\vh,\theta_\star,\phih)$ coordinates}

If we use coordinates $(\uh,\vh,\theta_\star,\phih)$ then, in the formulas above,
the value of $\sigma$ is~$1$, but
one has to replace $\Omega^2$ by $\Omegasqh$, $\bphi$ by $\bh$, and
the value of $\varsigma$ has to be as in~\eqref{varsigma}.
Special care has to be taken when calculating $\Gamma_{\uh\uh}^*$
because in this case one has to differentiate $\Omegasqh$ with respect to $u$:
\begin{eqnarray*}
	\Gamma_{\uh\uh}^{\uh}&=&\varsigma\left(\frac{\partial_{r_\star}(\Omega^2)}{\Omega^2}
	-2\kappa_+\right),\\
	\Gamma_{\uh\uh}^{\vh}&=&\Gamma_{\uh\uh}^{\theta_\star}\ =\ \Gamma_{\uh\uh}^{\phih}
	\ =\ 
	0.
\end{eqnarray*}

\renewcommand{\gamma}{\truegamma}

\section{Characterization of the parameters of subextremal KNdS spacetimes}
\label{charge_appendix}

The main result of this appendix is summarized in Lemma~\ref{subextremal},
which characterizes subextremal
Kerr--Newman--de Sitter black holes in terms of $(r_-, r_+,\Lambda a^2,\Lambda e^2)$. In Remarks~\ref{rmk-l} and~\ref{rmk-scale}, we consider alternative choices of parameters,
namely $\left(\Lambda,\frac{r_+}{r_-},a,e\right)$ or
$(\Lambda,M,a,e)$. As mentioned in the Introduction, in the case that $e=0$,
related results can be found in Lake and Zannias~\cite{lake} 
and Borthwick~\cite{borthwick}.

\begin{lem}\label{subextremal}
	Each Kerr--Newman--de Sitter subextremal solution is determined by a quadruple 
	$(r_-, r_+,\Lambda a^2,\Lambda e^2)$ satisfying\/ $0<r_-<r_+$,
	$$
	0\leq\Lambda e^2<\frac{3 \alpha^2 (1 + 2 \alpha)}{(1 + 2 \alpha + 3 \alpha^2)^2}<\frac{1}{4}
	$$
	and
	$$
	0<\Lambda a^2<l(\alpha, \Lambda e^2), 
	$$
	with 
	$$
	\alpha=\frac{r_+}{r_-},
	$$
	where $l$ is the function given by\/~\eqref{l}. 
	The graph of $l$ is sketched in  {\rm Figure~\ref{region}}.
	\begin{figure}[ht]
		\begin{psfrags}
			\centering
			\includegraphics[scale=1]{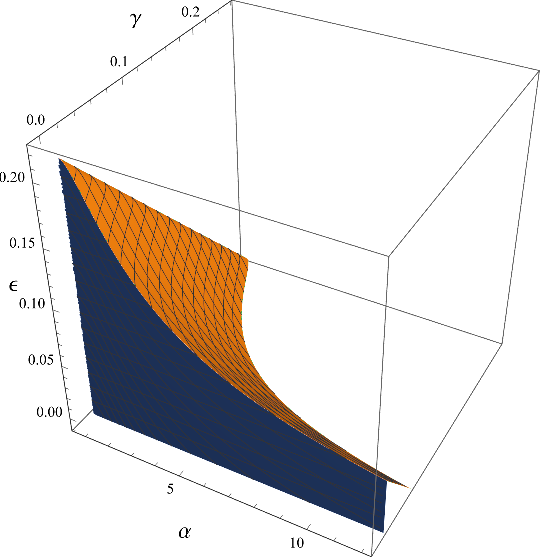}
			\caption[Text that appears in the picture directory]{The graph of $l$ bounds
				the region where the parameters $\alpha=\frac{r_+}{r_-}$, $\epsilon=\Lambda a^2$ and $\gamma=\Lambda e^2$
				can vary.}
			\label{region}
		\end{psfrags}
	\end{figure}
	The graphs of $l(\,\cdot\,,\Lambda e^2)$ for several values of $\Lambda e^2\in\left[0,\frac{1}{4}\right)$
	are sketched in\/ {\rm Figure~\ref{graph-of-l}}.
	\begin{figure}[ht]
		\begin{psfrags}
			\centering
			\includegraphics[scale=1]{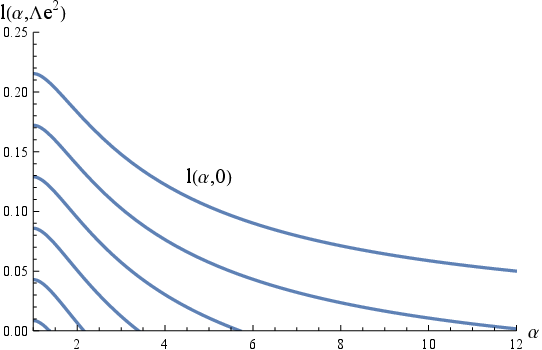}
			\caption[Text that appears in the picture directory]
			{Sketch of the graphs of $l(\,\cdot\,,0)$,
				$l(\,\cdot\,,0.05)$, $l(\,\cdot\,,0.1)$, $l(\,\cdot\,,0.15)$, 
				$l(\,\cdot\,,0.20)$, $l(\,\cdot\,,0.24)$.}
			\label{graph-of-l}
		\end{psfrags}
	\end{figure}
	For a choice of parameters on the graph of $l$ we have $r_+=r_c$ (see {\rm Figure~\ref{graph-l}}).
	\begin{figure}[ht]
		\begin{psfrags}
			\centering
			\includegraphics[scale=1]{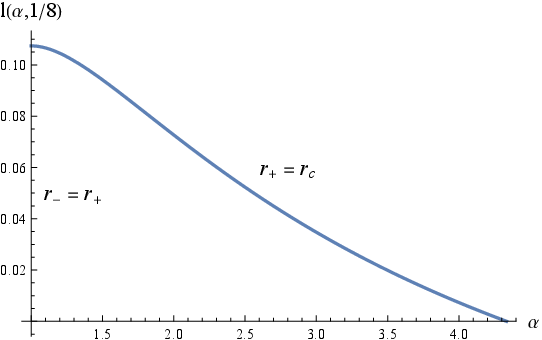}
			\caption[Text that appears in the picture directory]{The graph of $l(\,\cdot\,,1/8)$.}
			\label{graph-l}
		\end{psfrags}
	\end{figure}
	The value of $\Lambda$ is given by
	\bea
	\label{lambda_0}
	\Lambda=\frac{\left(3-\Lambda a^2 \right)-3
		\sqrt{\left(1-\,\frac{\Lambda a^2}{3} \right)^2-\,\frac{4(\Lambda a^2+\Lambda e^2) (r_-^2+r_+^2+r_-r_+)}{3r_-r_+}}}{2(r_-^2+r_+^2+r_-r_+)},
	\eea
	and the value of $M$ is given by
	\be\label{M}
	M=\frac{\Lambda}{6}(2r_-r_+r_c+r_-^2r_++r_-^2r_c+r_+^2r_-+r_+^2r_c+r_c^2r_-+r_c^2r_+),
	\ee
	with $r_c$ given by
	\bea
	\label{rc}
	r_c=-\,\frac{r_-+r_+}{2}+
	\sqrt{\frac{(r_-+r_+)^2}{4}+\frac{3(a^2+e^2)}{\Lambda r_-r_+}}.
	\eea
	Alternative formulas for the mass are
	\begin{eqnarray}\label{alternative-for-mass}
	\sqrt{\Lambda}M&=&\frac{(1+\alpha)(\alpha^2(\sqrt{\Lambda}r_-)^4+3(\Lambda a^2+\Lambda e^2))}{6\alpha (\sqrt{\Lambda}r_-)}\\
	&=&\frac{-(\sqrt{\Lambda}r_-)^4+(3-\Lambda a^2)(\sqrt{\Lambda}r_-)^2+3(\Lambda a^2+\Lambda e^2)}
	{6(\sqrt{\Lambda}r_-)}.\label{alt-2}
	\end{eqnarray}
\end{lem}
\begin{rmk}\label{rmk-l}
	We can write~\eqref{lambda_0} in the form
	\be\label{two}
	r_-^2=
	\frac{3-\Lambda a^2-\,\frac{\sqrt{9 \alpha-6 \left(2+3 \alpha+2 \alpha^2\right) \Lambda a^2 +\alpha (\Lambda a^2)^2
				-12(1+\alpha+\alpha^2)\Lambda e^2
		}}{\sqrt{\alpha}}}{2\Lambda \left(1+\alpha+\alpha^2\right)}.
	\ee
	This formula is useful if we start with a parameter set 
	$$(\Lambda,\alpha,a,e),\ \mbox{with}\ 0\leq\Lambda e^2<\frac{1}{4},\
	0<\Lambda a^2\leq 21-6\sqrt{12+\Lambda e^2}\ \mbox{and}\/\
	1<\alpha< l^{-1}\left(\Lambda a^2,\Lambda e^2\right)$$
	(see\/ {\rm Remark~\ref{beauty}} for a clarification of the meaning of\/ $l^{-1}$).
	Fixing $\Lambda$, $e$ and $a$, consider the function $r_-(\Lambda,\,\cdot\,,a,e)$.
	We have
	\begin{eqnarray*}
		\partial_\alpha r_-(\Lambda,1,a,e)&=&-\,\frac{r}{2}(\Lambda,1,a,e)
		\ =\ -\partial_\alpha r_+(\Lambda,1,a,e),\\
		\partial_\alpha r_-(\Lambda,l^{-1}(\Lambda a^2,\Lambda e^2),a,e)&=&0,\\
		\partial_\alpha r_c(\Lambda,1,a,e)&=&0,\\
		\partial_\alpha r_c(\Lambda,l^{-1}(\Lambda a^2,\Lambda e^2),a,e)&=&-
		\partial_\alpha r_+(\Lambda,l^{-1}(\Lambda a^2,\Lambda e^2),a,e).
	\end{eqnarray*}
	Note that
	the functions $\sqrt{\Lambda}r_\#$ depend solely on $(\alpha,\Lambda a^2,\Lambda e^2)$.
	In {\rm Figure~\ref{rrr}}, we sketch their graphs for 
	$\Lambda e^2=\gamma=1/8$ and
	$\Lambda a^2=\epsilon=(21-6\sqrt{12+\gamma})/2$. (See~\eqref{formula-for-lmo}
	for a formula for $l^{-1}(\Lambda a^2,\Lambda e^2)$.)
	\begin{figure}[ht]
		\begin{psfrags}
			\centering
			\includegraphics[scale=0.8]{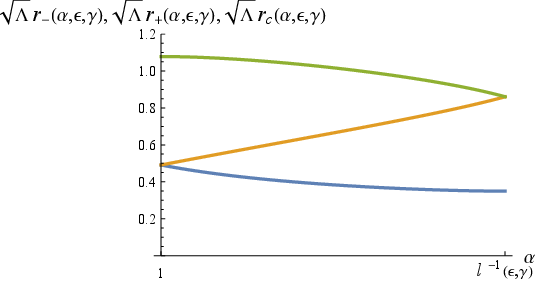}
			\caption[Text that appears in the picture directory]{Sketch the graphs of
				$r_-$, $r_+$ and $r_c$, for $\Lambda e^2=\gamma=1/8$ and
				$\Lambda a^2=\epsilon=(21-6\sqrt{12+\gamma})/2$.}
			\label{rrr}
		\end{psfrags}
	\end{figure}
	This implies that the mass $M$ (whose expression is given 
	in~\eqref{M}, \eqref{alternative-for-mass} and~\eqref{alt-2}) satisfies
	$$
	\partial_\alpha M_-(\Lambda,1,a,e)\ =\ \partial_\alpha M(\Lambda,l^{-1}(\Lambda a^2,\Lambda e^2),a,e)\ =\ 0.
	$$
	The function $M(\Lambda,\,\cdot\,,a,e)$ is strictly increasing
	in the interval $[1,l^{-1}(\Lambda a^2,\Lambda e^2)]$. 
	The function $\sqrt{\Lambda}M$ also depends solely on $(\alpha,\Lambda a^2,\Lambda e^2)$.
	In Figure~\eqref{figura-da-massa}, we sketch the graph of
	$\sqrt{\Lambda}M$ for $\Lambda e^2=1/8$ (recall that 
	$\mbox{$\Lambda e^2\in[0,\frac{1}{4})$}$).
	\begin{figure}[ht]
		\begin{psfrags}
			\psfrag{M}{{\small $M$}}
			\centering
			\includegraphics[scale=.8]{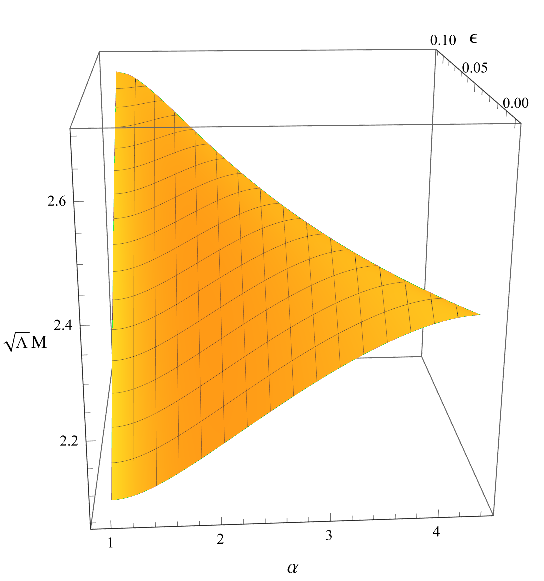}
			\caption[Text that appears in the picture directory]{Sketch of the graph of $\sqrt{\Lambda}M$ for $\Lambda e^2=\gamma=1/8$, $\Lambda a^2=\epsilon\in[0,21-6\sqrt{12+\gamma}]$
				and $\alpha\in[1,l^{-1}(\epsilon,\gamma)]$.}
			\label{figura-da-massa}
		\end{psfrags}
	\end{figure}
\end{rmk}
We can use $\partial_\alpha r_-(\Lambda,l^{-1}(\Lambda a^2,\Lambda e^2),a,e)=0$ to 
show that $\partial_\alpha M(\Lambda,l^{-1}(\Lambda a^2,\Lambda e^2),a,e)=0$; this
together with~\eqref{M} allows us to deduce that $	\partial_\alpha r_c(\Lambda,l^{-1}(\Lambda a^2,\Lambda e^2),a,e)=-
\partial_\alpha r_+(\Lambda,l^{-1}(\Lambda a^2,\Lambda e^2),a,e)$.
The other assertions follow by direct calculation.
\begin{rmk}\label{rmk-scale}
	In\/ {\rm Lemma~\ref{subextremal}} we can think that our subextremal solution is
	characterized by $(r_-,\alpha,\Lambda a^2,\Lambda e^2)$ and in\/ {\rm  Remark~\ref{rmk-l}} we can think that
	our subextremal solution is
	characterized by $(\Lambda,\alpha,\Lambda a^2,\Lambda e^2)$. Formula~\eqref{two}
	relates $r_-$ with $\Lambda$. In this paper, the preferred viewpoint
	is that each spacetime is characterized by the three parameters
	$(\alpha,\epsilon,\gamma)=\left(\frac{\sqrt{\Lambda}r_+}{\sqrt{\Lambda}r_-},(\sqrt{\Lambda}a)^2,(\sqrt{\Lambda}e)^2\right)$.
	These determine the other quantities, like $\sqrt{\Lambda}r_-$ or $\sqrt{\Lambda}M$.
	In this sense, the cosmological constant appears as a scale parameter.
\end{rmk}

\begin{rmk}
	The reader interested in the characterization of Kerr--de Sitter spacetimes
	in terms of the parameters $(\Lambda,M,a)$ (which is physically more natural,
	although not what we need in\/ {\rm Section~\ref{null})} should consult\/
	{\rm Borthwick~\cite{borthwick}}, where extreme and fast Kerr--de Sitter 
	are also considered.
\end{rmk}

\begin{rmk}
	In summary, the region where the parameters $(\alpha,\epsilon,\gamma)$
	can vary 
	(which is sketched in\/ {\rm Figure~\ref{region})} is defined by
	$$
	\left\{(\alpha,\epsilon,\gamma)\in\bbR^3:1<\alpha<\infty,\
	0\leq \gamma <E(\alpha),\ 0<\epsilon<l(\alpha,\gamma) \right\},$$
	where $E$ is given in~\eqref{um} and $l$ is given in~\eqref{l};
	alternatively, by
	$$\left\{(\alpha,\epsilon,\gamma)\in\bbR^3:0\leq\gamma<\frac{1}{4},\
	1<\alpha<\Gamma(\gamma),\  0<\epsilon<l(\alpha,\gamma)\right\},
	$$	
	where~$\Gamma$ is given in~\eqref{Gamma};
	or still by
	$$\left\{(\alpha,\epsilon,\gamma)\in\bbR^3:0\leq\gamma<\frac{1}{4},\
	0<\epsilon<21-6\sqrt{12+\gamma},\ 1<\alpha<l^{-1}(\epsilon,\gamma)\right\},
	$$
	where $l^{-1}$ is given in~\eqref{formula-for-lmo}.
\end{rmk}

\begin{altproof}{\,{\rm Lemma~\ref{subextremal}}}\\
	\noindent {\bf 1) ($\Delta_r$)}
	The relationship between $(M, \Lambda, a,e)$ and $(r_-, r_+,\Lambda a^2,\Lambda e^2)$ is
	obtained via the polynomial $\Delta_r$. On the one hand this polynomial is given by
	\bea
	\Delta_r&=&(r^2+a^2)\left(1-\,\frac{\Lambda}{3}r^2\right)-2Mr+ e^2\nonumber\\
	&=&-\,\frac{\Lambda}{3}r^4+\left(1-\,\frac{\Lambda}{3}a^2\right)r^2
	-2Mr+a^2+ e^2,\nonumber
	\eea
	and on the other hand it is given by
	\bea
	\Delta_r&=&-\,\flt(r-r_n)(r-r_-)(r-r_+)(r-r_c),\nonumber
	\eea
	since,
	according to our subextremal hypothesis, $\Delta_r$ has one negative root $r_n$,
	and three positive roots $\mbox{$r_-<r_+<r_c$}$. Expanding the last expression
	yields
	\bea
	\Delta_r&=&-\,\flt r^4+\flt(r_n+r_-+r_++r_c)r^3\nonumber\\
	&&-\flt(r_nr_-+r_nr_++r_nr_c+r_-r_++r_-r_c+r_+r_c)r^2\nonumber\\
	&&-\flt(-r_nr_-r_+-r_nr_-r_c-r_nr_+r_c-r_-r_+r_c)r\nonumber\\
	&&-\flt r_nr_-r_+r_c.\nonumber
	\eea
	We will now compare the coefficients of $\Delta_r$ in $r^0$, $r^1$, $r^2$ and $r^3$;
	the coefficient in $r^4$ is already matched.
	
	\noindent {\bf 1a) (Coefficient in $r^3$)}
	We start with the coefficient of $\Delta_r$ in $r^3$. As it is equal to $0$, we have
	\be\label{r-n}
	r_n=-(r_-+r_++r_c).
	\ee
	
	\noindent {\bf 1b) (Coefficient in $r^0$)}
	Equating the coefficients of $\Delta_r$ in $r^0$, we get, using~\eqref{r-n},
	\be\label{r-0}
	\flt(r_-+r_++r_c)r_-r_+r_c=a^2+e^2.
	\ee
	We regard this as a quadratic equation for $r_c$, which will be satisfied
	by choosing an appropriate $r_c$.
	Defining
	\be\label{x}
	x:=\frac{3(a^2+e^2)}{\Lambda r_-r_+},
	\ee
	\eqref{r-0} can be rewritten as 
	\be\label{x-0}
	x=r_c(r_c+r_-+r_+).
	\ee
	
	\noindent {\bf 1c) (Coefficient in $r^1$)} Equating the coefficients of $\Delta_r$ in $r^1$ we arrive at
	the value of $M$ given in~\eqref{M}. The expression~\eqref{alternative-for-mass} is
	obtained from~\eqref{M} using $r_+=\alpha r_-$ and~\eqref{rc} (deduced below).
	The formula~\eqref{alt-2} is the statement that $r_-$ is a root of $\Delta_r$.
	Of course, in~\eqref{alt-2} we may replace $r_-$ by $r_+$ or by $r_c$.
	
	\noindent {\bf 1d) (Coefficient in $r^2$)}
	Finally, we compare the coefficients of $\Delta_r$ in $r^2$. Using again~\eqref{r-n}, leads to
	\be\label{x2}
	\flt(r_-^2+r_+^2+r_-r_++x)=1-\,\flt a^2.
	\ee
	Defining
	\begin{eqnarray*}
		b&:=&r_-^2+r_+^2+r_-r_+,\\
		\epsilon&:=&\Lambda a^2,\\
		\gamma&:=&\Lambda e^2,
	\end{eqnarray*}
	and using~\eqref{x}, equation~\eqref{x2} is equivalent to
	$$
	\flt b+\frac{\epsilon+\gamma}{\Lambda r_-r_+}=1-\,\frac{\epsilon}{3},
	$$
	or
	\be\label{L2}
	\frac{b}{3}\Lambda^2-\left(1-\,\frac{\epsilon}{3} \right)\Lambda+\frac{\epsilon+\gamma}{r_-r_+}=0.
	\ee
	This is a quadratic equation which will determine $\Lambda$.
	
	\noindent {\bf 2) ($\Lambda$)}
	As we are considering de Sitter black holes, $\Lambda$ 
	is positive, and therefore $\epsilon$ is positive 
	and $\gamma$ is nonnegative.
	So, the product of the roots of~\eqref{L2} is positive,
	and the two roots have the same sign. Since $\Lambda$ is positive,
	we must have
	\be\label{3}
	0<\epsilon<3.
	\ee
	Solving~\eqref{L2}, we get
	\be\label{quadratic}
	\Lambda=\frac{\left(1-\,\frac{\epsilon}{3} \right)\pm
		\sqrt{\left(1-\,\frac{\epsilon}{3} \right)^2-\,\frac{4(\epsilon+\gamma) b}{3r_-r_+}}}{\frac{2b}{3}}.
	\ee
	We denote these values of $\Lambda$ by $\Lambda_0$ and $\Lambda_1$.
	$\Lambda_0$ and $\Lambda_1$ will have the minus sign and the plus sign in front of the square root,
	respectively,
	so that $\Lambda_0\leq\Lambda_1$. We will see below that only the
	value of $\Lambda_0$ is admissible.
	
	\noindent {\bf 2a) ($\Lambda$ real)}
	Define
	$$
	\alpha:=\frac{r_+}{r_-}>1,
	$$
	so that 
	$$
	\frac{b}{r_-r_+}=1+\alpha+\frac{1}{\alpha}.
	$$
	The discriminant in~\eqref{quadratic} should be nonnegative, that is
	\be
	\left(1-\,\frac{\epsilon}{3} \right)^2-\,\frac{4(1+\alpha+\alpha^2)}{3\alpha}(\epsilon+\gamma)\geq 0.
	\label{mmm}
	\ee
	Since $\epsilon$ is positive,
	we see that $\gamma$ has to satisfy
	\bea
	\label{gammaupper}
	0\leq\gamma<\frac{3}{4\left( 1+\alpha+\frac{1}{\alpha}\right)}\leq\frac{1}{4}.
	\eea
	This is not our final restriction on $\epsilon$ though.
	By requiring that $r_c>r_+$ a further restriction on $\epsilon$ will arise,
	as well as the requirement that $\Lambda=\Lambda_0$.
	
	\noindent {\bf 3) ($r_c$)}
	To obtain $r_c$, we rewrite~\eqref{x-0} as
	$$
	r_c^2+(r_-+r_+)r_c-x=0.
	$$
	Hence,
	$$
	r_c=-\,\frac{r_-+r_+}{2}+
	\sqrt{\frac{(r_-+r_+)^2}{4}+x}.
	$$
	This is~\eqref{rc}.
	The condition $r_c>r_+$ is equivalent to
	$$
	(r_-+r_+)^2+4x>(r_-+3r_+)^2,
	$$
	or
	$$
	x>2r_+^2+r_-r_+.
	$$
	Using the definition of $x$ in~\eqref{x}, this is the same as
	$$
	\frac{3(a^2+e^2)}{\Lambda r_-r_+}>2r_+^2+r_-r_+,
	$$
	or
	\be\label{k}
	\frac{a^2+e^2}{\Lambda}>\frac{1}{3}r_-r_+^2(r_-+2r_+)=:k.
	\ee
	Now we use the definition of $\epsilon$ and the assumption that $\Lambda$ is positive
	which, as we have seen, implies $0<\epsilon<3$. For $\epsilon$ in this range,
	\eqref{k} is equivalent to
	\be\label{choice}
	\epsilon+\gamma>k\Lambda^2=\frac{3k}{b}\left(1-\,\frac{\epsilon}{3}\right)\Lambda-
	\frac{3k}{br_-r_+}(\epsilon+\gamma),
	\ee
	where, for the last equality, we used the quadratic equation~\eqref{L2} for $\Lambda$.
	Our last task is to guarantee~\eqref{choice}.
	
	\noindent {\bf 3a) ($\Lambda=\Lambda_0$)}
	First we take $\Lambda=\Lambda_0$ in~\eqref{choice}. Then
	$$
	\left(\frac{b}{3k}+\frac{1}{r_-r_+}\right)\frac{\epsilon+\gamma}{1-\,\frac{\epsilon}{3}}>\Lambda_0=
	\frac{\left(1-\,\frac{\epsilon}{3} \right)-
		\sqrt{\left(1-\,\frac{\epsilon}{3} \right)^2-\,\frac{4(\epsilon+\gamma) b}{3r_-r_+}}}{\frac{2b}{3}},
	$$
	which can be rewritten as 
	\be\label{negative}
	\sqrt{\left(1-\,\frac{\epsilon}{3} \right)^2-\,\frac{4(\epsilon+\gamma) b}{3r_-r_+}}>
	\left(1-\,\frac{\epsilon}{3}\right)
	-\,\frac{2b}{3}\left(\frac{b}{3k}+\frac{1}{r_-r_+}\right)\frac{\epsilon+\gamma}
	{1-\,\frac{\epsilon}{3}}.
	\ee
	We need to examine the sign of the
	right-hand side of~\eqref{negative}. This can be written as
	$$
	\left(1-\,\frac{\epsilon }{3}\right)-\,\frac{2 \left(1+\alpha+\alpha^2\right) 
		\left(1+2\alpha+3\alpha^2\right)
		(\epsilon+\gamma)}{3 \alpha^2 (1+2 \alpha) \left(1-\,\frac{\epsilon }{3}\right)}.
	$$
	Above we guaranteed~\eqref{mmm}.
	Since we have that
	$$
	\frac{\frac{4(1+\alpha+\alpha^2)}{3\alpha}}
	{\frac{2 \left(1+\alpha+\alpha^2\right) 
			\left(1+2\alpha+3\alpha^2\right)}{3 \alpha^2 (1+2 \alpha)}}=
	\frac{2\alpha(1+2\alpha)}{1+2\alpha+3\alpha^2}\geq 1\ \ \mbox{for}\ \alpha\geq 1,
	$$
	the right-hand side of~\eqref{negative} is nonnegative.
	Therefore, \eqref{negative} is equivalent to
	\bea
	&&\left(1-\,\frac{\epsilon }{3}\right)^2
	-\,\frac{4}{3} \left(1+\alpha+\frac{1}{\alpha}\right) (\epsilon+\gamma) \nonumber\\
	&&\qquad-\left[\left(1-\,\frac{\epsilon }{3}\right)-\,\frac{2 \left(1+\alpha+\alpha^2\right) 
		\left(1+2\alpha+3\alpha^2\right)
		(\epsilon+\gamma)}{3 \alpha^2 (1+2 \alpha) \left(1-\,\frac{\epsilon }{3}\right)}\right]^2>0.
	\label{expand}
	\eea
	Let us define
	\begin{eqnarray*}
		c &=&\frac{4}{3} \left(1+\alpha+\frac{1}{\alpha}\right),\\
		d &=&\frac{2 \left(1+\alpha+\alpha^2\right) 
			\left(1+2\alpha+3\alpha^2\right)}{3 \alpha^2 (1+2 \alpha)}.
	\end{eqnarray*}
	Then, expanding the left-hand side of~\eqref{expand}, we get
	$$
	\left(1-\,\frac{\epsilon}{3}\right)^2(2d-c)-d^2(\epsilon+\gamma)>0,
	$$
	or
	$$
	\left(1-\,\frac{\epsilon}{3}\right)^2-\,\frac{d^2}{2d-c}(\epsilon+\gamma)>0,
	$$
	because
	$$
	2d-c=\frac{4 (1 + \alpha + \alpha^2)^2}{3 \alpha^2 (1 + 2 \alpha)}
	$$
	is positive. Substituting $c$ and $d$ by their expressions in terms of $\alpha$,
	we obtain
	\be
	\left(1-\,\frac{\epsilon}{3}\right)^2-\,\frac{(1 + 2 \alpha + 3 \alpha^2)^2}{3 \alpha^2 (1 + 2 \alpha)}(\epsilon+\gamma)>0.\label{gamma-bound}
	\ee
	The inequality~\eqref{gamma-bound} has solutions for $\epsilon$ positive
	if and only if
	\bea
	\gamma<E(\alpha):=\frac{3 \alpha^2 (1 + 2 \alpha)}{(1 + 2 \alpha + 3 \alpha^2)^2}\leq\frac{1}{4}.
	\label{um}
	\eea
	This inequality for $\gamma$ is more restrictive than \eqref{gammaupper}.
	We remark that
	$$
	\frac{\frac{(1 + 2 \alpha + 3 \alpha^2)^2}{3 \alpha^2 (1 + 2 \alpha)}}{\frac{4(1+\alpha+\alpha^2)}{3\alpha}}=
	\frac{(1 + 2 a + 3 a^2)^2}{4\alpha(1+2\alpha)(1+\alpha+\alpha^2)}\geq 1\ \ 
	\mbox{for}\ \alpha\geq 1
	$$
	(this function is $1$ at $\alpha=1$ and grows to $\frac{9}{8}$ as 
	$\alpha$ goes to $+\infty$).
	Therefore, all solutions of~\eqref{gamma-bound} are solutions of~\eqref{mmm}.
	The inequality~\eqref{gamma-bound} is equivalent to
	$$
	0<\epsilon <l(\alpha,\gamma)\quad\mbox{or}\quad\epsilon>\tilde{r}(\alpha,\gamma),
	$$
	where
	\bea
	l(\alpha,\gamma)&:=&
	\frac{3}{2 \alpha^2 (1+2 \alpha)}
	\Biggl(\left(1+4 \alpha+12 \alpha^2+16 \alpha^3+9 \alpha^4\right)\Biggr.\nonumber\\
	&&\qquad\qquad\Biggl.
	-\left(1+2 \alpha+3 \alpha^2\right) \sqrt{(1+\alpha)^2(1+2\alpha+9\alpha^2)
		+\frac{4\gamma}{3}\alpha^2(1+2\alpha)}\Biggr),
	\label{l}
	\eea
	and $\tilde{r}(\alpha,\gamma)$ is defined by the same expression, with the minus sign replaced by
	a plus sign.
	The function~$\tilde{r}$ satisfies
	$$
	\tilde{r}(\alpha,\gamma)>\frac{3}{2 \alpha^2 (1+2 \alpha)}
	\left(12 \alpha^2+16 \alpha^3\right)>3.
	$$
	So, the inequality $\epsilon>\tilde{r}(\alpha,\gamma)$ is incompatible with~\eqref{3}.
	Fix a $\gamma$ belonging to the interval $\left[0,\frac{1}{4}\right)$.
	The function $l(\,\cdot\,,\gamma)$ is strictly decreasing and satisfies
	$$
	l(1,\gamma)=21-6\sqrt{12+\gamma},\qquad
	\partial_\alpha l(1,\gamma)=0.
	$$
	In Figure~\ref{graph-of-l}, we sketch the graphs of $l(\,\cdot\,,0)$,
	$l(\,\cdot\,,0.05)$, $l(\,\cdot\,,0.1)$, $l(\,\cdot\,,0.15)$, 
	$l(\,\cdot\,,0.20)$ and $l(\,\cdot\,,0.24)$.
	The function~$l(\,\cdot\,,\gamma)$ is nonnegative for
	$\alpha\in
	\left(1,
	\Gamma(\gamma)
	\right)$, and
	$$
	l(\Gamma(\gamma),\gamma)=0,
	$$ 
	where $\Gamma(\gamma)$ is the inverse of the function $E$ in~\eqref{um}, i.e.\ is the solution greater than one of 
	$$
	\gamma=\frac{3 \Gamma^2(\gamma) (1 + 2 \Gamma(\gamma))}{(1 + 2 \Gamma(\gamma) + 3 \Gamma^2(\gamma))^2},
	$$
	namely
	\be\label{Gamma}
	\Gamma(\gamma):=\frac{1-2 \gamma+\sqrt{1-4 \gamma }+\sqrt{2} \sqrt{1-\gamma-4 \gamma ^2 +(1+\gamma) \sqrt{1-4 \gamma }} }{6 \gamma }
	\ee
	for $\gamma\in\left(0,\frac{1}{4}\right)$, $\Gamma(0)=+\infty$
	(in which case $\alpha\in(1,+\infty)$). 
	Note that
	$$
	\lim_{\gamma\nearrow \frac{1}{4}}\Gamma'(\gamma)=-\infty,\qquad
	\lim_{\gamma\searrow 0}\frac{\Gamma(\gamma)}{\frac{2}{3\gamma}}=1.
	$$
	The graph of $\Gamma$ is sketched in
	Figure~\ref{gamma-fig}. 
	\begin{figure}[ht]
		\begin{psfrags}
			\psfrag{g}{{\small $\gamma$}}
			\psfrag{B}{{\small $\alpha$}}
			\psfrag{z}{{\small $(\gamma,\Gamma(\gamma))$}}
			\psfrag{w}{{\small $(E(\alpha),\alpha)$}} 		
			\centering
			\includegraphics[scale=0.8]{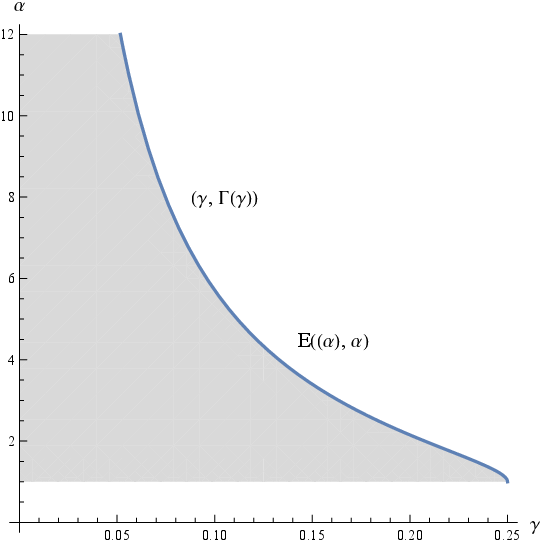}
			\caption[Text that appears in the picture directory]{Sketch the graph of
				$\Gamma$. The shaded region is the projection of the parameter 
				region in the $(\alpha,0,\gamma)$ plane.}
			\label{gamma-fig}
		\end{psfrags}
	\end{figure}
	We remark that 
	$$
	\lim_{\alpha\to\infty}\frac{l(\alpha,0)}{\frac{2}{3\alpha}}=1.
	$$
	In summary, for $\Lambda=\Lambda_0$ inequality~\eqref{choice} is 
	satisfied for
	$$
	0<\epsilon<l(\alpha,\gamma).
	$$

	\noindent {\bf 3b) ($\Lambda=\Lambda_1$)}
	Now we take $\Lambda=\Lambda_1$ in~\eqref{choice}. Then, instead of~\eqref{negative}, we obtain
	\be\label{positive}
	\frac{2b}{3}\left(\frac{b}{3k}+\frac{1}{r_-r_+}\right)\frac{\epsilon+\gamma}
	{1-\,\frac{\epsilon}{3}}
	-\left(1-\,\frac{\epsilon}{3}\right)>\sqrt{\left(1-\,\frac{\epsilon}{3} \right)^2-\,\frac{4(\epsilon+\gamma) b}{3r_-r_+}}.
	\ee
	Note that the left-hand side of~\eqref{positive} is the symmetric
	of the right-hand side of~\eqref{negative}. 
	We have shown above that the right-hand side of~\eqref{negative} is nonnegative.
	Hence, inequality~\eqref{positive} never holds.
	
	We conclude that we must choose \mbox{$\Lambda=\Lambda_0$} in~\eqref{quadratic}.
	This establishes expression \eqref{lambda_0} of Lemma~\ref{subextremal}.\\
\end{altproof}

Define
$$
s(\epsilon,\gamma):=\sqrt{\left(\epsilon - (21-6 \sqrt{12+\gamma})\right) \left(\epsilon - (21+6 \sqrt{12+\gamma})\right)}=
\sqrt{9-42 \epsilon +\epsilon ^2-36\gamma}.
$$
\begin{rmk}\label{beauty}
	The inverse of the function $l$, in the sense that
	$$
	l^{-1}(l(\alpha,\gamma),\gamma)=\alpha\qquad\mbox{and}\qquad l(l^{-1}(\epsilon,\gamma),\gamma)=\epsilon,
	$$
	is
	\bea
	&&	l^{-1}(\epsilon,\gamma)\ =\ 
	\frac{1}{54 (\epsilon+\gamma)}\Bigl(
	(3-\epsilon)s(\epsilon,\gamma)
	+9
	-24\epsilon
	+\epsilon ^2-18\gamma\Bigr.
	\label{formula-for-lmo}
	\\
	&&\qquad \  \Bigl.+\sqrt{2 \left(
		(27
		-\epsilon ^3 
		+9\gamma(3-\epsilon))s(\epsilon,\gamma)
		+(81
		-189 \epsilon
		-216 \epsilon ^2
		-21\epsilon ^3
		+\epsilon ^4) 
		-9\gamma(9+36\gamma+66\epsilon+\epsilon^2)
		\right)}\Bigr),
	\nonumber
	\eea
	for $0\leq\gamma<\frac{1}{4}$ and $\epsilon\in\,]0,21-6\sqrt{12+\gamma}]$. In particular, $l^{-1}(21-6\sqrt{12+\gamma},\gamma)=
	l^{-1}(l(1,\gamma),\gamma)=1$.
\end{rmk}

\begin{rmk}
	If we start with a parameter set $(\Lambda,M,a,e)$, we can regard~\eqref{lambda_0}
	and~\eqref{alternative-for-mass} as a system for~$r_-$ and~$\alpha$.
	In this case, it is natural to consider the quantities
	\begin{eqnarray*}
		A&:=&\sqrt{\alpha}+\frac{1}{\sqrt{\alpha}},\\
		R&:=&\sqrt{\Lambda\alpha}r_-\ =\ \sqrt{\Lambda r_-r_+}.
	\end{eqnarray*}
	In terms of these quantities, \eqref{lambda_0}
	and~\eqref{alternative-for-mass} can be written as
	\bea
	R^2&=&\frac{3-\epsilon-3\sqrt{\left(1-\,\frac{\epsilon}{3}\right)^2
			-\,\frac{4(A^2-1)(\epsilon+\gamma)}{3}}}{2(A^2-1)},
	\label{RM}\\
	\sqrt{\Lambda}M&=&A\frac{R^4+3(\epsilon+\gamma)}{R}.\label{LM}
	\eea
\end{rmk}
The equation~\eqref{RM} can be solved for $A$, yielding
\bea
A&=&\frac{\sqrt{R^4+(3-\epsilon)R^2-3(\epsilon+\gamma)}}{R^2}.\label{AR}
\eea
The system~\eqref{LM}$-$\eqref{AR} allows us to determine $A$ and $R$ from
$\sqrt{\Lambda}M$, $\Lambda a^2$ and $\Lambda e^2$. Of course, another
way to do this would be by calculating the roots of the polynomial $\Delta_r$.

\section{Komar integrals}\label{app-D}

\renewcommand{\gamma}{{\gin}}

The purpose of this appendix is to obtain expression~\eqref{mass}, for the 
mass of the black hole at the event horizon. Besides the 
expected term $\frac{M}{\Xi^2}$, the mass has one correction
term which is a constant multiple of the product of
the cosmological constant by the Parikh volume of the black hole,
and has another term which reflects the energy due to the electric field.
In this way we recall the physical interpretation of the parameters of our metric.
Let ${\tilde{K}}:=\pp$.
Consider the $2$-dimensional surface on the event horizon
${\cal S}=\{(\uh,\vh,\theta_\star,\phih):\uh=0,\ \vh=(\vh)_0\}$,
where $\vh=(\vh)_0$, for some fixed $(\vh)_0$.
Let us calculate the Komar integral
$$
JQ=\frac{1}{16\pi}\int_{{\cal S}}\star \md {\tilde{K}}_\flat.
$$
($JQ$ is not the angular momentum if $e\neq 0$.
See~\cite[Subsection~III.B]{Gibbons2} for the definition of the angular momentum.)
\begin{rmk} 
In the calculation below we use the coordinates $(u,v,\theta_\star,\phi_\star)$
of\/ {\rm Subsection~\ref{cha}},
notwithstanding the fact that ${\cal S}$ is contained in the event horizon,
where these coordinates are not defined. Our computations may be easily justified by considering a
sequence ${\cal S}_n$ of surfaces which approximate ${\cal S}$ from within the black
hole, which is covered by the coordinates, 
computing the integrals over ${\cal S}_n$,
and then passing to the limit.
\end{rmk}
We have that
$$
{\tilde{K}}_\flat=-\bphi\gamma_{\phi_\star\phi_\star}\md v
+\gamma_{\theta_\star\phi_\star}\md\theta_\star
+\gamma_{\phi_\star\phi_\star}\md\phi_\star
$$
and
$$
(\star \mmd {\tilde{K}}_\flat)_{\alpha\beta}=\frac{1}{2}(\mmd {\tilde{K}}_\flat)^{\mu\nu}\dV_{\mu\nu\alpha\beta}.
$$
As the vectors on the tangent space of ${\cal S}$ do not have neither
components in $\partial_u$ nor components in $\partial_v$,
only the components
$(\mmd {\tilde{K}}_\flat)^{uv}$ and $(\mmd {\tilde{K}}_\flat)^{vu}$
are relevant for the calculation of $JQ$. On the other hand, as
\begin{eqnarray*}
	(\mmd u)^\sharp&=&-\,\frac{1}{2\Omega^2}\partial_v+\frac{\bphi}{2\Omega^2}\pp,\\
		(\mmd v)^\sharp&=&-\,\frac{1}{2\Omega^2}\partial_u,\\
			(\mmd\theta_\star)^\sharp&=&\gamma^{\theta_\star\theta_\star}
			\partial_{\theta_{\star}}+\gamma^{\theta_\star\phi_\star}
			\partial_{\phi_{\star}},\\
	(\mmd\phi_\star)^\sharp&=&-\,\frac{\bphi}{2\Omega^2}\partial_u
	+\gamma^{\theta_\star\phi_\star}
	\partial_{\theta_{\star}}+\gamma^{\phi_\star\phi_\star}
	\partial_{\phi_{\star}},
\end{eqnarray*}
the terms $(\mmd {\tilde{K}}_\flat)_{\theta_\star *}$ and $(\mmd {\tilde{K}}_\flat)_{*\theta_\star}$
do not contribute to $JQ$. Notice that $(\mmd {\tilde{K}}_\flat)_{v\phi_\star}$ and
$(\mmd {\tilde{K}}_\flat)_{\phi_\star v}$ do not enter into the calculation
of $(\mmd {\tilde{K}}_\flat)^{uv}$ and $(\mmd {\tilde{K}}_\flat)^{vu}$
either. Thus, we write
$$ 
\mmd {\tilde{K}}_\flat=-\partial_u(\bphi\gamma_{\phi_\star\phi_\star})\md u\wedge\mmd v
+\partial_u\gamma_{\phi_\star\phi_\star}\mmd u\wedge\mmd\phi_\star
+\ldots,
$$ 
\begin{eqnarray*}
(\mmd {\tilde{K}}_\flat)^\sharp&=&\frac{1}{4\Omega^4}\left(
\partial_u(\bphi\gamma_{\phi_\star\phi_\star})
-\bphi\partial_u\gamma_{\phi_\star\phi_\star}
\right)(\partial_u\otimes\partial_v-\partial_v\otimes\partial_u)
+\ldots\\ 
&=&\frac{1}{4\Omega^4}\left(
\gamma_{\phi_\star\phi_\star}
\partial_{r_\star}\bphi
\right)(\partial_u\otimes\partial_v-\partial_v\otimes\partial_u)
+\ldots
\end{eqnarray*}
and
$$
\star\mmd {\tilde{K}}_\flat=\frac{1}{2\Omega^2}\gamma_{\phi_\star\phi_\star}
\partial_{r_\star}\bphi\frac{L\sin\theta}{\Xi}\md\theta_\star\wedge\mmd\phi_\star+
\ldots.
$$
Remembering that $\bphi$ is constant on the event horizon, we have
$$
\frac{\partial_{r_\star}\bphi}{\Omega^2}=
\frac{\partial_r\bphi\partial_{r_\star}{r}}{\Omega^2}
=\partial_r\bphi\frac{\Delta_r}{\Omega^2}\frac{\partial_{r_\star}{r}}{\Delta_r}=
-\partial_r\bphi\frac{\Upsilon}{\rho^2\Delta_\theta}\frac{1}{r^2+a^2}
=-\partial_r\bphi\frac{r^2+a^2}{\rho^2}
$$
on the event horizon, because $\Upsilon=4(r^2+a^2)^2\Delta_\theta$ there.
Since
\begin{eqnarray*}
\partial_r\bphi&=&\frac{4\Xi ar\Delta_\theta}{\Upsilon}-\,\frac{2\Xi a\partial_r\Delta_r}{\Upsilon}-\,\frac{2\Xi a(r^2+a^2)\Delta_\theta}{\Upsilon^2}
(4(r^2+a^2)r\Delta_\theta-a^2\sin^2\theta\partial_r\Delta_r)\\
&=&-\,\frac{4\Xi ar\Delta_\theta}{\Upsilon}-\,
\frac{2\Xi a\rho^2\partial_r\Delta_r}{\Upsilon(r^2+a^2)}
\end{eqnarray*}
and
$$
\gamma_{\phi_\star\phi_\star}=\frac{\sin^2\theta}{\Xi^2\rho^2}\Upsilon,
$$
it follows that
\begin{eqnarray*}
\star\mmd {\tilde{K}}&=&
\frac{a}{\Xi}\left(
2r(r^2+a^2)\frac{\Delta_\theta\sin^2\theta}{(r^2+a^2\cos^2\theta)^2}+
\partial_r\Delta_r\frac{\sin^2\theta}{r^2+a^2\cos^2\theta}
\right)
\frac{L\sin\theta}{\Xi}\md\theta_\star\wedge\mmd\phi_\star+
\ldots\\
&=&\frac{a}{\Xi}\left(
2r(r^2+a^2)\frac{\Delta_\theta\sin^3\theta}{(r^2+a^2\cos^2\theta)^2}+
\partial_r\Delta_r\frac{\sin^3\theta}{r^2+a^2\cos^2\theta}
\right)
\frac{r^2+a^2}{\Xi}\frac{\partial\theta}{\partial\theta_\star}\md\theta_\star\wedge\mmd\phi_\star+
\ldots,
\end{eqnarray*}
because $L=(r^2+a^2)\partial_{\theta_\star}\theta$ on the event horizon.
Now, we have
\begin{eqnarray*}
\int_0^{\frac{\pi}{2}}
\frac{(1+\frac{\Lambda}{3}a^2\cos^2\theta)\sin^3\theta}{(r^2+a^2\cos^2\theta)^2}\md\theta&=&
\int_0^1\frac{(1+\frac{\Lambda}{3}a^2x^2)(1-x^2)}{(r^2+a^2x^2)^2}\md x\\
&=&-\,\frac{3 a r \left(\Lambda  r^2-1\right)-\left(a^2 \left(\Lambda  r^2+3\right)+3 r^2 \left(\Lambda  r^2-1\right)\right) \arctan\left(\frac{a}{r}\right)}{6 a^3 r^3},\\
\int_0^{\frac{\pi}{2}}
\frac{\sin^3\theta}{r^2+a^2\cos^2\theta}\md\theta&=&
\int_0^1\frac{1-x^2}{r^2+a^2x^2}\md x\\
&=&\frac{-a r+\left(r^2+a^2\right) \arctan\left(\frac{a}{r}\right)}{a^3 r}.
\end{eqnarray*}
Taking into account that
\be\label{account1}
\partial_r\Delta_r=
-\,\frac{2}{3} \Lambda  r \left(r^2+a^2\right)-2 M+2 r \left(1-\,\frac{\Lambda  r^2}{3}\right)
\ee
and
\be\label{account2}
M=\frac{\left(r^2+a^2\right) \left(1-\,\frac{\Lambda  r^2}{3}\right)+e^2}{2 r},
\ee
we arrive at
\begin{eqnarray*}
&&2r(r^2+a^2)\int_0^\pi
\frac{\Delta_\theta\sin^2\theta}{(r^2+a^2\cos^2\theta)^2}\md\theta
+\partial_r\Delta_r\int_0^\pi
\frac{\sin^2\theta}{r^2+a^2\cos^2\theta}\md\theta
\\
&&\qquad\qquad=
\frac{-4 a^3 r \left(\Lambda  r^2-3\right)
+6e^2\left( a  r
-\left(r^2+a^2\right)\arctan\left(\frac{a}{r}\right)\right) }{3 a^3 r^2}
\\
&&\qquad\qquad=
\frac{8 M}{r^2+a^2}-\,\frac{2 e^2 \left(\left(r^2+a^2\right)^2 \arctan \left(\frac{a}{r}\right)-a r \left(r^2-a^2\right)\right)}{a^3 r^2 \left(r^2+a^2\right)}\\
&&\qquad\qquad=
\frac{8 M}{r^2+a^2}-\,\frac{2 e^2}{a(r^2+a^2)}{\rm BT}\left(\frac{a}{r}\right).
\end{eqnarray*}
Here
$$
{\rm BT}(x)=\left(x^2+2\right) \arctan x+\frac{\arctan x-x+x^3}{x^2},
$$
an expression that makes clear the behavior of the function ${\rm BT}$,
as the first two terms of the Taylor expansion of $x\mapsto\arctan x$
around zero are $x-\,\frac{x^3}{3}$. In particular,
the function ${\rm BT}$ is strictly increasing, satisfies ${\rm BT}(0)=0$,
${\rm BT}(0)=\frac{8}{3}$, and $\lim_{x\to+\infty}
\left({\rm BT}\,(x)/
\left(\frac{\pi}{2}x^2\right)\right)=1$.
This finally yields
\be\label{J}
JQ=\frac{Ma}{\Xi^2}-\,\frac{e^2}{4\Xi^2}{\rm BT}\left(\frac{a}{r_+}\right)
=\frac{Ma}{\Xi^2}-\,\frac{Q^2}{4}{\rm BT}\left(\frac{a}{r_+}\right),
\ee
with
$$
Q=\frac{e}{\Xi}
$$
the total charge
(see~\cite[(19)]{CCK} and~\cite[beginning of Subsection~3.2 on p.\ 32]{peter1}).

Now we consider $K$ to be $K=\frac{\partial_t}{\Xi}$ and we turn to the calculation of the Komar integral
$$
{\cal M}=-\,\frac{1}{8\pi}\int_{{\cal S}}\star \md K_\flat.
$$
This is the mass calculated at the event horizon. 
To justify this choice of $K$ we refer to~\cite[(2.12) and Subsection~2.2]{Gibbons}
(see also~\cite[(18)]{CCK}).
As $\partial_t=\frac{1}{2}\partial_v-\,\frac{1}{2}\partial_u$,
we have that
$$
\Xi K_\flat=-\Omega^2\md u+\Omega^2\md v+\frac{(\bphi)^2\gamma_{\phi_\star\phi_\star}}{2}
\md v-\,\frac{\bphi\gamma_{\theta_\star\phi_\star}}{2}\md\theta_\star-\,
\frac{\bphi\gamma_{\phi_\star\phi_\star}}{2}\md\phi_\star,
$$
\begin{eqnarray*}
\Xi \md K_\flat&=&2\partial_{r_\star}\Omega^2\md u\wedge\mmd v
+\frac{1}{2}\partial_{r_\star}((\bphi)^2\gamma_{\phi_\star\phi_\star})\md u\wedge\mmd v-\,\frac{1}{2}
\partial_{r_\star}(\bphi\gamma_{\phi_\star\phi_\star})\md u\wedge\mmd\phi_\star+
\ldots,
\end{eqnarray*}
\begin{eqnarray*}
	\Xi (\md K_\flat)^\sharp&=&
	-\,\frac{1}{4\Omega^4}\left(
	2\partial_{r_\star}\Omega^2
	+\frac{1}{2}\partial_{r_\star}((\bphi)^2\gamma_{\phi_\star\phi_\star})
	-\,\frac{\bphi}{2}
	\partial_{r_\star}(\bphi\gamma_{\phi_\star\phi_\star})
	\right)(\partial_u\otimes\partial_v-\partial_v\otimes\partial_u)
	+
	\ldots\\
	&=&
	-\,\frac{1}{4\Omega^4}\left(
	2\partial_{r_\star}\Omega^2
	+\frac{\bphi\gamma_{\phi_\star\phi_\star}}{2}
	\partial_{r_\star}\bphi
	\right)(\partial_u\otimes\partial_v-\partial_v\otimes\partial_u)
	+
	\ldots
\end{eqnarray*}
and
$$
-\Xi\star\mmd K_\flat=\left(
\frac{\partial_{r_\star}\Omega^2}{\Omega^2}
+\frac{\bphi}{2}\frac{1}{2\Omega^2}\gamma_{\phi_\star\phi_\star}
\partial_{r_\star}\bphi
\right)
\frac{L\sin\theta}{\Xi}\md\theta_\star\wedge\mmd\phi_\star+
\ldots.
$$
Using~\eqref{surface}, \eqref{account1}, \eqref{account2} and~\eqref{J}, this yields
\begin{eqnarray}
{\cal M}&=&\frac{1}{8\pi\Xi}\int_{{\cal S}}\frac{\partial_{r_\star}\Omega^2}{\Omega^2}
\frac{L\sin\theta}{\Xi}\md\theta_\star\wedge\mmd\phi_\star+
2\frac{\Xi a}{r^2+a^2}JQ\nonumber\\
&=&\frac{1}{8\pi\Xi}\int_{{\cal S}}(2\kappa_+)\frac{r^2+a^2}{\Xi}\sin\theta
\frac{\partial\theta}{\partial\theta_\star}\md\theta_\star\wedge\mmd\phi_\star+
2\frac{\Xi a}{r^2+a^2}JQ\label{smarr}\\
&=&\frac{1}{2\Xi^2}\partial_r\Delta_r+2\frac{\Xi a}{r^2+a^2}JQ\nonumber\\
&=&\frac{1}{\Xi^2}\left(
\frac{1}{2} \left(-\,\frac{a^2 \left(\Lambda  r^2+3\right)}{3 r}
-\Lambda  r^3+r-\,\frac{e^2}{r}\right)
+\frac{2M a^2}{r^2+a^2}
-\,\frac{e^2 \left(\left(r^2+a^2\right)^2 
	\arctan\left(\frac{a}{r}\right)-a r \left(r^2-a^2\right)\right)}{2 a r^2 \left(r^2+a^2\right)}\right)
\nonumber\\
&=&
\frac{1}{\Xi^2}\left(
-\,\frac{\left(r^2+a^2\right) \left(a r \left(\Lambda  r^2-1\right)+e^2 \arctan\left(\frac{a}{r}\right)\right)}{2 a r^2}\right)
\nonumber\\
&=&
\frac{1}{\Xi^2}\left(
M-\,\frac{\Lambda}{3}r(r^2+a^2)-\,\frac{e^2}{2\Xi a}
\frac{a r+\left(r^2+a^2\right) \arctan\left(\frac{a}{r}\right)}{r^2}
\right)\nonumber\\
&=&\frac{M}{\Xi^2}-\,\frac{\Lambda}{3\Xi}\frac{r_+(r_+^2+a^2)}{\Xi}-\,
\frac{e^2}{2\Xi^2a}{\rm AT}\left(\frac{a}{r_+}\right)\nonumber\\
&=&\frac{M}{\Xi^2}-\,\frac{\Lambda}{3\Xi}\frac{r_+(r_+^2+a^2)}{\Xi}-\,
\frac{Q^2}{2a}{\rm AT}\left(\frac{a}{r_+}\right),\label{mass}
\end{eqnarray}
with
$$
{\rm AT}(x)=x+(1+x^2)\arctan x.
$$
\begin{rmk} The quantity
	$$\frac{4\pi}{3}\frac{r_+(r_+^2+a^2)}{\Xi}$$
	is the\/ {\rm Parikh~\cite[(10)]{Parikh}} volume of the black hole (see also\/~{\rm \cite[(91)]{BL}}).
\end{rmk} 
Of course, \eqref{smarr} is Smarr's formula (recall~\eqref{angular})
\begin{eqnarray*}
{\cal M}&=&\frac{\kappa_+}{4\pi\Xi}\left(4\pi\frac{r_+^2+a^2}{\Xi}\right)+2\Omega_H JQ
\ =\ \frac{\kappa_+}{4\pi\Xi}{\cal A}+2\Omega_H JQ\\
&=&\frac{\kappa_+}{4\pi\Xi}{\cal A}+2\Omega_H\frac{Ma}{\Xi^2}-\frac{Q^2}{2}\frac{\Xi a}{r_+^2+a^2}
{\rm BT}\left(\frac{a}{r_+}\right)
\end{eqnarray*}
(see, for example, \cite[(9)]{CCK} for the value of ${\cal A}$, and see~\cite{dolan} for much more on Smarr's formula).

\section*{Acknowledgements}

The authors were both partially supported by FCT/Portugal through UID/MAT/04459/2019. 
A.T.~Franzen was also supported by SFRH/BPD/115959/2016.


\section*{Conflict of interest statement}
On behalf of all authors, the corresponding author states that there is no conflict of interest.


\begin{thebibliography}{99}
 	

\bibitem{akcay}  Akcay, S.\ and Matzner, R.A.,
The Kerr-de Sitter universe. 
Classical Quantum Gravity 28 (2011), no.~8, 085012.

\bibitem{BL} Ballik, W.\ and Lake, K.,
Vector volume and black holes.
Phys.\ Rev.\ D 88 (2013), no.\ 10, 104038.

\bibitem{Mann1}
Balushi, A.A.\ and Mann, R.B.,
Null hypersurfaces in Kerr-(A)dS spacetimes.
Classical Quantum Gravity\ 36 (2019), 245017


\bibitem{borthwick} Borthwick, J.,
Maximal Kerr--de Sitter spacetimes. 
Classical Quantum Gravity 35 (2018), no. 21, 215006.

\bibitem{CCK}
Caldarelli, M.M., Cognola, G.\ and Klemm, D.,
Thermodynamics of Kerr-Newman-AdS black holes and conformal field theories. 
Classical Quantum Gravity 17 (2000), no.\ 2, 399–420.

\bibitem{carter1} Carter, B.,
Black hole equilibrium states. Black holes/Les astres occlus (École d'Été Phys.\ Théor., Les Houches, 1972), pp.\ 57–214. Gordon and Breach, New York, 1973.

\bibitem{CF} Costa, J.L.\ and Franzen, A.T.,
Bounded energy waves on the black hole interior of Reissner-Nordstr\"{o}m--de Sitter. 
Ann.\ Henri Poincaré 18 (2017), no.~10, 3371–3398.


\bibitem{DHR} Dafermos, M., Holzegel, G.\ and Rodnianski., I., 
A scattering theory construction of dynamical vacuum black holes (2013).
arXiv:1306.5364.



\bibitem{m-luk} Dafermos, M.\ and Luk, J.,
The interior of dynamical vacuum black holes I: The $C^0$-stability of the Kerr Cauchy horizon 
(2017).
arXiv:1710.01722.


\bibitem{DR} Dafermos, M.\ and Rodnianski, I.,
The red-shift effect and radiation decay on black hole spacetimes. 
Comm.\ Pure Appl.\ Math.\ 62 (2009), no.~7, 859–919.

\bibitem{m-lec} Dafermos, M.\ and Rodnianski, I.,
Lectures on black holes and linear waves. 
Evolution equations, 97–205, Clay Math.\ Proc., 17, Amer.\ Math.\ Soc., Providence, RI, 2013.


\bibitem{DS} Dafermos, M.\ and Shlapentokh-Rothman, Y.,
Rough initial data and the strength of the blue-shift instability on cosmological black holes with $\Lambda>0$. 
Classical Quantum Gravity 35 (2018), no.~19, 195010.

\bibitem{dolan}
Dolan, B.P.,
Kastor, D.,
Kubiz\v{n}ák, D.,
Mann, R.B.\ and
Traschenk, J.,
Thermodynamic volumes and isoperimetric inequalities for de Sitter black holes.
Physical Review D 87 (2013), 104017.


\bibitem{anne_kerr} Franzen, A.T.,
Boundedness of massless scalar waves on Kerr interior backgrounds. 
Ann.\ Henri Poincaré~21 (2020), no.~4, 1045–1111.

\bibitem{Gibbons2}
Gibbons, G.W., Yi Pang, Y.\ and Pope, C.N.,
Thermodynamics of magnetized Kerr-Newman black holes.
Physical Review D 89 (2014), 044029.

\bibitem{Gibbons} Gibbons, G.W., Perry, M.J.\ and Pope, C.N.,
The first law of thermodynamics for Kerr--anti-de Sitter black holes. 
Classical Quantum Gravity 22 (2005), no.\ 9, 1503–1526.

\bibitem{peter1}
Hintz, P.,
Non-linear stability of the Kerr--Newman--de Sitter family of charged black holes.
Ann.\ PDE 4 (2018), no.\ 1, paper no.\ 11.

\bibitem{peter2} Hintz, P.\ and Vasy, A.,
Analysis of linear waves near the Cauchy horizon of cosmological black holes.
J.~Math.\ Phys.\ 58 (2017), no.~8, 081509.

\bibitem{mann2}
Imseis, M.T.N., Balushi, A.A.\ and Mann, R.B.,
Null Hypersurfaces in Kerr-Newman-AdS Black Hole and Super-Entropic Black Hole Spacetimes.
Classical Quantum Gravity 38 (2021), 045018,


\bibitem{John} John, F.,
Lower bounds for the life span of solutions of nonlinear wave equations in three dimensions.
Comm.\ Pure Appl.\ Math.\ 36 (1983), no.~1, 1–35.

\bibitem{K} Klainerman, S.,
Uniform decay estimates and the Lorentz invariance of the classical wave equation.
Comm.\ Pure Appl.\ Math.\ 38 (1985), no.~3, 321–332.

\bibitem{KL} Klainerman, S,\ and Rodnianski, I.,
Rough solutions of the Einstein-vacuum equations.
Ann.\ of Math.~(2)~161 (2005), no.~3, 1143–1193.


\bibitem{krani} Kraniotis, G.V.,
Precise relativistic orbits in Kerr and Kerr-(anti) de Sitter spacetimes. 
Classical Quantum Gravity 21 (2004), no.~19, 4743–4769.

\bibitem{lake} Lake, K.\ and Zannias, T.,
Global structure of Kerr–de Sitter spacetimes. 
Phys.\ Rev.\ D 92 (2015), no.~8, 084003.

\bibitem{luk-sbierski} Luk, J.\ and Sbierski, J., Instability results for the wave equation
in the interior of Kerr black holes. J.\ Funct.\ Anal.\ 271 (2016), no.\ 7, 1948–1995.

\bibitem{Morawetz} Morawetz, C.S.,
The limiting amplitude principle.
Comm.\ Pure Appl.\ Math.\ 15 (1962), 349–361.


\bibitem{Parikh} Parikh, M.K.,
Volume of black holes. 
Phys. Rev. D (3) 73 (2006), no. 12, 124021.

\bibitem{pretorius} Pretorius, F.\ and Israel, W.,
Quasi-spherical light cones of the Kerr geometry. 
Classical Quantum Gravity 15 (1998), no.~8, 2289–2301.

\bibitem{Sbierski} Sbierski, J., 
On the initial value problem in general relativity and wave propagation in black-hole spacetimes. Doctoral thesis (2014). 
https://doi.org/10.17863/CAM.16140.
\end{thebibliography}
\end{document}